\documentclass[acmsmall,screen]{acmart}

\usepackage[noxcoloropts]{jhupllab}

\usepackage{jhupllab}
\usepackage{amsmath}
\usepackage{mathpartir}
\usepackage{accents}
\usepackage{wasysym}
\SetSymbolFont{wasy}{bold}{U}{wasy}{m}{n} 
\usepackage{stmaryrd}
\SetSymbolFont{stmry}{bold}{U}{stmry}{m}{n} 
\usepackage{mathtools}
\usepackage{amsthm}
\usepackage{thmtools}
\usepackage{thm-restate}
\usepackage{ifthenx}
\usepackage{relsize}
\usepackage{array}
\usepackage{wrapfig}
\usepackage{xargs}
\usepackage{soul}
\usepackage{booktabs} 
\usepackage[ruled]{algorithm2e} 
\usepackage{microtype} 
\usepackage[pro]{fontawesome5}

\newcommand{\bbrule}[4][]{\inferrule*[left={\bbrulename{#2}},#1]{#3}{#4}}
\newcommand{\bbrulename}[1]{{\sc #1}}

\newcommand{\fnstyle}[1]{\text{\textsmaller{\sc{#1}}}}
\newcommand{\deffn}[2]{%
    \expandafter\newcommand\expandafter{\csname #1\endcsname}[1][]{\fnstyle{#2}\ifthenelse{\isempty{##1}}{}{(##1)}}%
}

\renewcommand{\listof}[2][]{{\color{red}\textrm{\bfseries We are not using \texttt{\char`\\listof} in this document!}}}
\renewcommand{\setof}[2][]{{\color{red}\textrm{\bfseries We are not using \texttt{\char`\\setof} in this document!}}}

\newsavebox{\codeBoxA}
\newsavebox{\codeBoxB}
\newsavebox{\someBoxA}
\newsavebox{\someBoxB}

\makeatletter
\newcommand{\setlabel}[2]{\def\@currentlabel{#2}\label{#1}}
\makeatother

\definecolor{tableZebra}{gray}{0.9}

\newcommand{\tttilde}{{\raisebox{-0.1em}{\texttt{\char`\~}}}}
\newcommand{\codetilde}{\tttilde}

\defgt[gtcolon]{\texttt{\upshape :}}
\defgt[gttilde]{\texttt{\upshape \codetilde}}
\defgt[gtquestion]{\texttt{\upshape ?}}
\defgt[gtarrow]{\texttt{\upshape ->}}

\defgt[gtob]{\upshape \char`\{}
\defgt[gtcb]{\upshape \char`\}}
\defgt[gtop]{\upshape (}
\defgt[gtcp]{\upshape )}
\defgt[gtis]{\texttt{\upshape =}}
\defgt[gtfun]{\upshape fun}
\defgt[gtref]{\upshape ref}
\defgt[gtset]{\upshape <-}
\defgt[gtlbrack]{\upshape \lbrack}
\defgt[gtrbrack]{\upshape \rbrack}
\defgt[gtdoublequote]{\upshape "}
\defgt[gttrue]{\upshape true}
\defgt[gtfalse]{\upshape false}
\defgt[gtstring]{\upshape string}
\defgt[gtany]{\upshape any}
\defgt[gtderef]{\upshape !}

\defgt[gtint]{\upshape int}
\defgt[gtnot]{\upshape not}
\defgt[gtxor]{\upshape xor}
\defgt[gtinput]{\upshape input}

\defgn[binop]{\odot}
\defgn[unop]{\Box}

\defgt[gtplus]{\texttt{\upshape +}}
\defgt[gtminus]{\texttt{\upshape -}}
\defgt[gtless]{\texttt{\upshape <}}
\defgt[gtand]{\texttt{\upshape and}}
\defgt[gtor]{\texttt{\upshape or}}
\defgt[gtleq]{\texttt{\upshape <=}}
\defgt[gteq]{\texttt{\upshape =}} 
\defgt[gtat]{\texttt{\upshape @}}

\defgt[gtcreate]{\upshape create}
\defgt[gtread]{\upshape read}
\defgt[gtwrite]{\upshape write}
\defgt[gtfork]{\upshape fork}
\defgt[gtjoin]{\upshape join}
\defgt[gtend]{\upshape end}

\defgn[expr]{e}
\defgn[exprcon]{C}
\defgn[gexpr]{\expr_\textrm{glob}}
\defgn[env]{E}
\defgn[ecl]{c}
\defgn[elb]{b}
\defgn[eval]{v}
\defgn[evalcon]{\bar{v}}
\defgn[elbl]{\ell}
\defgn[efunc]{f}
\defgn[erec]{r}
\defgn[epat]{p}
\defgn[ev]{x}
\defgn[redex]{r}
\defgn[redexcon]{\bar{r}}
\defgn[reductioncon]{R}
\defgn[rccon]{\bar{R}}
\defgn[eint]{n}
\defgn[ebool]{\beta}
\defgn[estring]{\mathbb{S}}
\defgn[inlist]{I}
\defgn[inmap]{\iota}

\deffn{dom}{Dom}

\defgn[syntype]{\tau}
\defgn[tfun]{\tau \texttt{ -> } \tau}
\defgn[tpoly]{\alpha}
\defgn[vpoly]{a}

\newcommand{\tint}{\texttt{\small int}}
\newcommand{\tbool}{\texttt{\small bool}}
\newcommand{\vtrue}{\texttt{\small true}}
\newcommand{\vfalse}{\texttt{\small false}}
\newcommand{\pfun}{\texttt{\small fun}}
\newcommand{\picki}{\texttt{\small pick\_i}}
\newcommand{\pickb}{\texttt{\small pick\_b}}
\newcommand{\mzero}{\texttt{\small mzero}}
\newcommand{\retag}{\texttt{\small retag}}

\newcommand{\hastype}[2]{#1\ \colon #2}
\newcommand{\mkfun}[2]{#1 \texttt{\small \ -> } #2}

\newcommand{\mkfunv}[2]{\texttt{\small fun }#1 \texttt{\small \ -> } #2}

\newcommand{\mkdfun}[3]{\ttop\hastype{#1}{#2}\ttcp \texttt{\small \ -> } #3}

\newcommand{\mkvariant}[2]{V_1 \texttt{\small\ of }#1 | \cdots | V_n \texttt{\small\ of }#2}

\newcommand{\mkintersect}[2]{#1 \cap \cdots \cap #2}
\newcommand{\mktset}[2]{\ttob#1 \mathrel{\vert} #2\ttcb}
\newcommand{\mkmiu}[2]{\mu#1 . #2}
\newcommand{\mkgen}[1]{\fnstyle{Generator}(#1)}
\newcommand{\mkche}[2]{\fnstyle{Checker}(#1, #2)}
\newcommand{\mkwrap}[2]{\fnstyle{Wrapper}(#1, #2)}

\newcommand{\letfun}{\texttt{\small let f x = e in e}}
\newcommand{\letfunt}{\texttt{\small let f (x : \syntype) : \syntype \ = e in e}}
\newcommand{\eerror}{\texttt{\small ERROR}}

\newcommand{\stackname}[1]{\fnstyle{StackName}\ifthenelse{\isempty{#1}}{}{(#1)}}
\newcommand{\extract}[2]{\fnstyle{Extract}\ifthenelse{\isempty{#1}}{}{(#1,#2)}}
\newcommand{\rawval}[1]{\fnstyle{RawVal}\ifthenelse{\isempty{#1}}{}{(#1)}}

\newcommand{\smallstep}[1][]{\mathrel{\longrightarrow^1_{#1}}}
\newcommand{\smallsteps}[1][]{\mathrel{\longrightarrow^*_{#1}}}

\newcommand{\nsmallsteps}[1][]{\mathrel{\longarrownot\longrightarrow^*_{#1}}}

\deffn{freshen}{Freshen}

\newcommand{\freshvarA}[2]{\ensuremath{\freshen\ifthenelse{\isempty{#1}}{}{_{#1}(#2)}}}
\newcommand{\freshvarC}[2]{\ensuremath{\freshen^{\text{cs}}\ifthenelse{\isempty{#1}}{}{_{#1}(#2)}}}
\deffn{oldtonew}{OldToNew}

\defgn[senvterm]{l}
\defgn[senv]{L}
\defgn[sinstr]{\iota}
\defgn[sinstrs]{I}
\newcommand{\splitEnv}[3]{\fnstyle{SplitEnv}\ifthenelse{\isempty{#1}}{}{(#1\ifthenelse{\isempty{#2}}{}{,#2,#3})}}
\newcommand{\extractEnv}[1]{\fnstyle{ExtractEnv}\ifthenelse{\isempty{#1}}{}{(#1)}}
\newcommand{\extractEnvOld}[2]{\toDeprecate{\fnstyle{ExtractEnv}\ifthenelse{\isempty{#1}}{}{(#1\ifthenelse{\isempty{#2}}{}{,#2})}}}

\defgn[lenvterm]{z}
\defgn[lenv]{Z}
\defgn[lcl]{w}
\defgn[lexpr]{W}
\newcommand{\llookup}[3][\lenv]{#1(#2\ifthenelse{\isempty{#3}}{}{,#3})}

\deffn{rv}{RetV}
\deffn{firstv}{FirstV}

\defgn[lexprs]{\boldsymbol{W}}
\newcommand{\currentStack}[1]{\fnstyle{CurStack}\ifthenelse{\isempty{#1}}{}{(#1)}}

\newcommand{\lbbr}{[\mkern-6.3mu [\mkern-6.35mu [}
\newcommand{\rbbr}{]\mkern-6mu ]\mkern-6mu ]}
\renewcommand{\lbbr}{[\mkern-6.7mu [\mkern-6.7mu [}
\renewcommand{\rbbr}{]\mkern-6.3mu ]\mkern-6.3mu ]}
\newcommand{\substitute}[3]{{\small #1[#2/#3]}}
\newcommand{\matches}[2]{#1 \ \tttilde \ #2}
\newcommand{\ife}[3]{\texttt{\small if } #1 \texttt{\small \ then } #2 \texttt{\small \ else }#3}
\newcommand{\letin}[3]{\texttt{\small let }#1\texttt{\small\ = }#2\texttt{\small\ in }#3}
\newcommand{\letint}[4]{\texttt{\small let (}#1\texttt{\small\ : #2) = }#3\texttt{\small\ in }#4}
\newcommand{\redcon}[2]{\bar{R}[#1]\lbbr#2\rbbr}

\newcommand{\tc}[2]{\fnstyle{TC}(#1,#2)}

\newcommand{\ruleref}[1]{\textsc{#1}}

\usepackage{listings}

\newcommand{\keywordcolor}{\color{NavyBlue}}
\newcommand{\symbolcolor}{\color{RoyalPurple}}
\newcommand{\commentcolor}{\color{ForestGreen}}
\newcommand{\codesize}{\footnotesize} 
\newcommand{\plangbasicstyle}{\linespread{0.9}\ttfamily\codesize}
\newcommand{\plangkeywordstyle}{\plangbasicstyle\keywordcolor}
\newcommand{\plangsymbolstyle}{\plangbasicstyle\symbolcolor}
\newcommand{\plangcommentstyle}{\plangbasicstyle\commentcolor}

\newcommand{\ttob}{\text{\upshape\ttfamily\char`\{}}
\newcommand{\ttcb}{\text{\upshape\ttfamily\char`\}}}
\newcommand{\ttop}{\text{\upshape\ttfamily(}}
\newcommand{\ttcp}{\text{\upshape\ttfamily)}}

\newcommand{\codefigurestart}[1][.4\textwidth]{\hspace*{15pt}\begin{minipage}{#1-20pt}\camlset\footnotesize}
\newcommand{\codefigurestop}{\end{minipage}}

\long\def\tossit#1{}
\long\def\keepit#1{#1}

\def\fullversion{\long\def\infull{\keepit}\long\def\inshort{\tossit}}
\fullversion 

\def\trimmedversion{\long\def\inregular{\tossit}\long\def\intrimmed{\keepit}}

\trimmedversion 

\definecolor{ForestGreen}{rgb}{.132,.545,.132}
\definecolor{Plum}{rgb}{.868,.628,.868}
\definecolor{RoyalPurple}{rgb}{.38,.25,.6}
\definecolor{NavyBlue}{rgb}{0,0,.5}
\definecolor{VioletRed}{rgb}{.816,.125,.565}

\SetAlFnt{\small}
\SetAlCapFnt{\small}
\SetAlCapNameFnt{\small}
\SetAlCapHSkip{0pt}
\IncMargin{-\parindent}

\newcommand*\rot{\rotatebox{90}}
\newcommand*\red{\textcolor{red}}

\nonotes

\setcopyright{rightsretained} 
\acmDOI{10.1145/3689788}
\acmYear{2024}
\acmJournal{PACMPL}
\acmVolume{8}
\acmNumber{OOPSLA2}
\acmArticle{348}
\acmMonth{10}
\acmSubmissionID{oopslab24main-p743-p}
\received{2024-04-05}
\received[accepted]{2024-08-18}

  
\begin{document}

\title[Semantic-Type-Guided Bug Finding]{Semantic-Type-Guided Bug Finding}

\author{Kelvin Qian}
\orcid{0009-0002-1768-4018}
\affiliation{%
  \institution{Johns Hopkins University}
  \city{Baltimore}
  \country{USA}
}
\email{kqian2@alumni.jh.edu}

\author{Scott Smith}
 \orcid{0009-0005-0495-2716}
 \affiliation{%
   \institution{Johns Hopkins University}
   \city{Baltimore}
   \country{USA}
 }
\email{scott@cs.jhu.edu}

\author{Brandon Stride}
\orcid{0009-0000-9369-204X}
\affiliation{%
   \institution{Johns Hopkins University}
   \city{Baltimore}
   \country{USA}
}
\email{bstride1@jhu.edu}

\author{Shiwei Weng}
\orcid{0000-0001-5440-4471}
 \affiliation{%
   \institution{Johns Hopkins University}
   \city{Baltimore}
   \country{USA}
}
\email{wengshiwei@jhu.edu}

\author{Ke Wu}
\orcid{0009-0008-2647-4198}
\affiliation{%
  \institution{Johns Hopkins University}
  \city{Baltimore}
  \country{USA}
}
\email{kwu48@jhu.edu}

\begin{abstract}
    In recent years, there has been an increased interest in tools that establish \emph{incorrectness} rather than correctness of program properties. In this work we build on this approach by developing a novel methodology to prove incorrectness of \emph{semantic typing} properties of functional programs, extending the incorrectness approach to the model theory of functional program typing. We define a semantic type refuter which refutes semantic typings for a simple functional language. We prove our refuter is co-recursively enumerable, and that it is sound and complete with respect to a semantic typing notion. An initial implementation is described which uses symbolic evaluation to efficiently find type errors over a functional language with a rich type system.
\end{abstract}

\maketitle
\section{Introduction}

Program verification tools verify that programs meet their specifications. However, recent works suggest that developers may find it more helpful to find \emph{incorrectness}: having a tool that shows a program fails to meet the specification.  Since a sound system for establishing correctness properties will sometimes fail (and report a potential error) on correct programs due to incompleteness, developers might lose trust and interest in the tool, causing them to avoid using it \cite{MakingStaticAnalysisWork}. 

With this in mind, some recent approaches find incorrectnesses rather than establish correctness \cite{10.1145/3371078, 10.1145/3632909, 10.1145/3527325, jakob2015falsificationviewsuccesstyping}. This ensures that users will only receive notifications when there is a provable failure. For example, Incorrectness Logic is a formalization of error-finding in a first-order, imperative context \cite{10.1145/3371078}. Similar ideas have also been investigated for statically typed functional languages: recent work on two-sided type systems explored incorrectness reasoning in higher-order program verification \cite{10.1145/3632909}. Within the context of dynamic languages, a notable example is \emph{success typing} in Elixir \cite{10.1145/1140335.1140356, jakob2015falsificationviewsuccesstyping}: Dialyzer, its static analysis tool, uses constraint-based type inference to report only provable type errors at compile time.

In this work, we present a type-directed bug finding technique that shares a similar philosophy with these approaches: our goal is to discover the type incorrectness of a program by finding provable type violations. There are several distinctive aspects of our work: (1) we study higher-order functional programs and a rich array of types including polymorphic, dependent, and refinement types; (2) our typing notion uses a theoretical foundation of \emph{semantic typing}, so it is a purely model-theoretic and non-proof-theoretic approach; and (3) we prove that the theoretical core of our type checker is both sound and complete with respect to this semantic typing notion: semantically ill-typed programs will be discovered type-incorrect by our algorithm, and programs discovered type-incorrect by our algorithm are semantically ill-typed.

\subsection{Semantic Typing}

Semantic typing \cite{DreyerBlog} is an approach to typing where, rather than relying on a static type system to define the meaning of types, type inhabitants are determined solely by the behaviors of expressions under an operational or denotational semantics of (untyped) expressions; this is more a model-theoretic notion as opposed to the proof-theoretic approach of type theory.  Semantic typing originated with Milner \cite{MILNER1978348}, where semantic types were used in place of a type system to show type soundness; however, they have largely been overshadowed by the proof-theoretic approach of \emph{type systems}.  The semantic approach to typing is a component of logical relations, dating back to Tait \cite{TaitRealizability}, but logical relations generally were used as a characterization of a type system to prove properties such as strong normalization, and not as the basis of meaning. There has been some resurgence of using semantic types / logical relations as the basis of meaning \cite{StepIndexedLogicalRelations,10.5555/1037736}. We believe our work is the first to use semantic types for type-directed bug finding by implementing a semantic-type-based bug finding method that is both sound and complete.
%
%
\subsection{Semantic-Type-Guided Refutation}

One fundamental challenge with semantic types is that type membership is undecidable: to type check a function, it must successfully type check when applied to all the elements of the function domain, which is often infinite in size and has computations with unbounded numbers of steps. However, if the focus is on finding \emph{incorrect} typings in the spirit of Incorrectness Logic, there is no infinite domain that needs to be positively verified; instead, only a single counterexample is needed.  We have produced a type incorrectness finder which is co-recursively enumerable (co-r.e.), i.e. if an expression does \emph{not} belong to a semantic type, our refutation tool will be able to find it in finite time. Inspired by ideas from dynamic contracts and related works \cite{Findler-Contracts,SoftContractVerif,RelCompleteCounterexamples,Meunier:2006,ParametricContracts}, as well as property-based testing \cite{10.1145/351240.351266}, we have designed a novel runtime representation of semantic types that allows type refutations to be encoded as expressions in the program itself and reduces type checking in our system to an error state reachability problem. 
%

We develop an implementation of our ideas that uses symbolic evaluation to refute typings:  symbolic evaluation can cover an infinite state space based on only one execution and so allows for more efficient type error finding. 



\subsection{Contributions}
The primary contributions of this paper are as follows:
\begin{enumerate}
    \item We define a novel embedding of semantic type refutation inside an untyped core functional language which is co-recursively enumerable.
    \item We prove that the embedding is sound and complete with respect to a standard semantic type definition for basic types (i.e. $\tint$, $\tbool$, and function types).
    \item We extend our refutation approach to other types, including refinement types, dependent function types, and parametric polymorphism, for which we also prove soundness and completeness with respect to semantic typing notions.
    \item We additionally present extensions for record, variant, intersection, and recursive types and subtyping, which we implement, but we do not establish correctness of these extensions.
    \item We describe our type refuter implementation and show it performs reasonably on benchmark examples.
\end{enumerate}

\section{Overview}
\label{sec_overview}

In this section we give a flavor of semantic-type-based refutation. Consider the following simple function containing a type error:

\begin{lstlisting}
let appl_int (fn : int -> int) : int = 
  let res = fn 1 in if res != 32767 then f 0 else (res - 1) < 0 
\end{lstlisting}

Since \lstinline|(res - 1)$ $  < 0| results in a boolean value, it clearly violates the specified \lstinline|int| return type. How can we quickly discover this type violation using a semantic typing notion?

A key aspect of semantic types is that membership is determined solely by program behaviors. For basic types such as \lstinline|int| or \lstinline|bool|, membership is decided by whether the given expression evaluates to an integer or boolean value. For function types, semantic typing dictates that for all argument values in the domain's semantic type, the application result evaluates to an element of the codomain's semantic type. In this example, the domain type is \lstinline|int ->$\ $int|, and the codomain type is \lstinline|int|, so to show type incorrectness, we need to find some element of \lstinline|int ->$\ $int| such that the application does not return an integer.

Given the unbounded nature of the function value space, how is it feasible to automatically generate a function in \lstinline|int ->$\ $int| that causes the above typing to fail?  The answer is surprisingly simple: in order to trigger a type error, it suffices to find only one input-output mapping performed by the function that will lead to the error. In particular, we can define a generator of arbitrary elements of type \lstinline|int ->$\ $int| as follows:

\mkgen{$\mkfun{\tint}{\tint}$} = \lstinline|fun i ->| \lstinline|if| \mkche{\lstinline|int|}{\lstinline|i|} \lstinline| then pick_i else ERROR|


This generated function will verify its argument is an integer, throw it away, and return a random integer. For \emph{some} run this can behave like \emph{any} fixed function on the integers and so gives a complete characterization of all the arguments for a function of e.g. type \lstinline|(int ->$\ $int) ->$\ $int|. Since these generated functions are very much a subset of all possible functions, the proof that they will suffice for the purpose of type refutation is subtle and makes our soundness proof challenging. 

Essentially, we have to prove that if there exists a concrete function that will trigger a type violation, then the generated function can reproduce the corresponding control flows leading to the error by picking the correct input and output pairs at every branching point, effectively simulating the behaviors of the concrete function without having to emulate its actual semantics. The full definition of generators and checkers has some additional nuance and is presented in Section \ref{sec_formalization}. We also prove that generators and checkers can be used to give a sound and complete characterization of semantic typing.

Another notable property of this construct is that the  generators and checkers are all directly implementable in the language itself, as can be seen in the example above. As a result, we can transform the type declarations into instrumented programs that are executable, checkable code. The instrumented programs only need \lstinline|pick_i| to obtain a random integer.  These instrumented programs can be directly run in an interpreter and may hit the error if the picked value happens to trigger it, but this is similar to property-based testing in that there is no completeness; on the example above, for instance, random choice is very unlikely to pick the path causing the type error.  However, if the instrumented program is executed in a symbolic evaluator, each \lstinline|pick_i| will initially be unconstrained, and constraints can accumulate as the program runs and can then be solved to find a path leading to the type error.


\subsection{Other Types}

This approach to typing would not be very useful if it only worked for simple function types.  Fortunately, it can be extended to incorporate a very rich grammar of types. In Sections \ref{sec_sound_extension} and \ref{sec_additional_extension}, we illustrate how refinement, dependent, polymorphic, variant, intersection, record, and recursive types, plus subtyping, can be added as extensions to the core system.

The following example showcases how our approach can refute a declaration that uses both refinement and dependent function types:
\begin{lstlisting}
let f (l : { list int | nonneg }) : { int | fun a -> a > 0 and (list_sum l 0 == a) } = list_sum l 0 
\end{lstlisting}
where \lstinline|list_sum| is a function that sums up a list of integers, and \lstinline|nonneg| is a boolean-valued function that constrains a list to contain only non-negative integers. The type requires that the returned value be the sum of the input list, which is positive. Using the generators and checkers for this function type, our type checker will produce list values that meet the \lstinline|nonneg| predicate and check whether the application result has the output type. Since \lstinline|nonneg| allows a list of zeroes, there will be executions where the application \lstinline!sum l 0! returns $0$, and our implementation can prove this fails to type check. 

To model recursive types, the generator/checker definitions above become non-well-founded as the type is arbitrarily deep. (It is also a problem in semantic typing and requires some subtlety to solve \cite{StepIndexedLogicalRelations}.) We resolve this issue by taking the further step to embed the checker definition \emph{in the language itself}: we embed types as pairs of generator and checker expressions.  Once types are embedded in the language, we can simply use the Y-combinator to define recursive types as fixed-points of type functions.  With types as values, parameterized types can also be easily defined.


\subsection{Use Checking}
We have up to now discussed how \emph{definitions} can be verified to be of the proper type; with just definition checking, however, a function could be \emph{used} at a type it was not declared at.  For example, the identity function typed as taking integers to integers could have a boolean passed to it, violating its declared interface. Simply checking definitions are correct will not lead to uses being correct.  Our implementation can insert use checks (so-called ``wraps'') to verify that type interfaces are not violated.  These checks are very similar to the use checks used in dynamic contracts \cite{Findler-Contracts}, so we will not include them in our formalization.
\section{Formalization}
\label{sec_formalization}

In this section we will first define a core functional language and its semantic typing rules. Then, we will define our type checker in terms of generator and checker embeddings as outlined in the previous section. We then give a theorem of soundness and completeness: our type checker fully and faithfully implements semantic typing.

Note that from this section on, we will use the term \textit{type checking} in place of semantic-type-guided bug finding. We use this term for its  broader meaning of validating the correctness of type declarations, not for a specifically type-system-based validation. 


\subsection{Core Language}

\begin{wrapfigure}{r}{.55\textwidth}
    \begin{grammar}
    \grule[values]{\eval}{
      \mathbb{Z}
      \gor \mathbb {B}
      \gor
      \mkfunv{\ev}{\expr}
    }
    \grule[expressions]{\expr}{
        \eval
        \gor \ev
        \gor \expr\ \expr
        \gor \expr\ \binop\ \expr
        \gline 
        \gor \ife{\expr}{\expr}{\expr}
        \gor \matches{\expr}{p}
        \gline
        \gor \picki 
        \gor \pickb
        \gor \eerror
    }
    \grule[variables]{\ev}{
                \textit{(identifiers)}
    }
    \grule[patterns]{p}{
      \tint
      \gor \tbool
      \gor \pfun
    }
    \grule[binops]{\binop}{
      \gtplus
      \gor \gtminus
      \gor \gtless
      \gor \gteq
      \gor \gtand
      \gor \gtor
      \gor \gtxor
    }
    \grule[types]{\syntype}{
      \tint
      \gor \tbool
      \gor \tfun
    }
    \grule[redexes]{r}{
      \eval \binop \eval
      \gor \ife{\vtrue}{\expr}{\expr}
      \gline
      \gor \ife{\vfalse}{\expr}{\expr}
      \gor \eval \ \eval
      \gline
      \gor \matches{\eval}{p}
      \gor \picki \gor \pickb
    }
    \grule[reduction contexts]{\reductioncon}{
      \bullet
      \gor \reductioncon \ \expr
      \gor \expr \ \reductioncon
      \gor \reductioncon \binop \expr
      \gor \eval \binop \reductioncon
      \gline
      \gor \ife{\reductioncon}{\expr}{\expr}
      \gor \matches{\reductioncon}{p}
  }
  \end{grammar}
  \caption{Core Language Grammar}
  \label{fig:core_grammar}
 \end{wrapfigure}

The grammar for the core language is given in Figure~\ref{fig:core_grammar}. It is mostly standard; the only non-traditional features are nondeterministic choice
for integers and booleans (\lstinline|$\picki$| and \lstinline|$\pickb$| respectively) and run-time typecasing
on integer, boolean, and function types (via the boolean-valued operator \lstinline|$\matches{\expr}{p}$|). The core language itself is untyped; we define types as properties of expressions over the untyped operational semantics.

We are using a small-step contextual operational semantics and include the definition of reduction (aka evaluation) contexts $\reductioncon$ and redexes $r$ in Figure~\ref{fig:core_grammar}. The black hole, $\bullet$, in a reduction context is unique. Filling a reduction context with an expression, denoted $R\lbbr\expr\rbbr$, is the operation of replacing the black hole, $\bullet$, in $\reductioncon$ with $\expr$.  \lstinline|ERROR| represents a runtime type error, and it is not considered a divergence. Finally, we will occasionally use the syntax \lstinline|let $x$ = $\expr_1$ in $\expr_2$|, which is syntactic sugar for \lstinline|(fun $x$ -> $\expr_2$) $\expr_1$|.

We define the following function to implement typecasing in the operational semantics.

\begin{definition}[Pattern Match]
  \label{def:pattern_match}
  \textsc{matches}$(\eval, p)$ is defined by the following clauses:
  \begin{description}
    \item \textsc{matches}$(\eval, \tint) =   
    \begin{cases}
      \gttrue & \text{if}\ \eval \in \mathbb{Z} \\
      \gtfalse & \text{if}\ \eval \not\in \mathbb{Z}
    \end{cases}$
    \ \ \ \ \ \ \ \ \ \textsc{matches}$(\eval, \tbool) =   
    \begin{cases}
      \gttrue & \text{if}\ \eval \in \mathbb{B} \\
      \gtfalse & \text{if}\ \eval \not\in \mathbb{B}
    \end{cases}$    \item \textsc{matches}$(\eval, \gtfun) =   
    \begin{cases}
      \gttrue & \text{if}\ \eval \text{ is a function value} \\
      \gtfalse & \text{if}\ \eval \text{ is not a function value}
    \end{cases}$
  \end{description}
\end{definition}


The operational semantics for this language is defined in Figure~\ref{fig:core_op_sem}. The definitions for bound, free, and closed are standard.

\begin{definition} Assuming closed expressions $\expr$ and $\expr'$, we define the following relations:
  \label{def:eval_relation}
  \begin{enumerate}
    \item $\expr \smallstep[] \expr'$ is the least relation satisfying the rules of Figure~\ref{fig:core_op_sem}.
    \item $\expr_1 \smallsteps[] \expr_n$ iff there exists a sequence $\expr_1 \smallstep[] \expr_2, \cdots, \expr_{n-1} \smallstep[] \expr_n$.
    \item $\expr \Uparrow$ iff for all $n > 0$, there exists a sequence $\expr \smallstep[] \expr_1, \cdots, \expr_{n-1} \smallstep[] \expr_n$.
  \end{enumerate}
\end{definition}

\begin{figure}[ht]
  \begin{mathpar}

      \bbrule{Red}{
        r \smallstep[] \expr
      }{
        R\lbbr r \rbbr \smallstep[] R'\lbbr\expr\rbbr
      }

      \bbrule{Err}{
        \expr \smallstep[] \eerror
      }{
        R\lbbr\expr\rbbr \smallstep[] \eerror
      }

      \bbrule{Add}{
        n_1, n_2 \in \mathbb{Z}
      }{
        n_1 + n_2 \smallstep[] \text{integer sum of } n_1 \text{ and } n_2
      }

      \bbrule{Add-Err}{
        \eval_1 \text{ or } \eval_2 \text{ is not an integer}
      }{
        \eval_1 + \eval_2 \smallstep[] \eerror
      }

      \bbrule{Appl}{
      }{
        (\mkfunv{\ev}{\expr_f}) \ \eval \smallstep[] \substitute{\expr_f}{\eval}{\ev}
      }

      \bbrule{Appl-Err}{
        \eval \text{ is not a function value}
      }{
        \eval \ \eval' \smallstep[] \eerror
      }

      \bbrule{If-True}{
        }{
          \ife{\vtrue}{\expr}{\expr'} \smallstep[] \expr
      }

      \bbrule{If-False}{
        }{
          \ife{\vfalse}{\expr}{\expr'} \smallstep[] \expr'
      }

      \bbrule{If-Err}{
        \eval \text{ is not a boolean value}
      }{
        \ife{\eval}{\expr}{\expr'} \smallstep[] \eerror
      }

      \bbrule{Nondet-int}{
        n \in \mathbb{Z}
      }{
        \picki \smallstep[] n
      }

      \bbrule{Nondet-bool}{
        b \in \mathbb{B}
      }{
        \pickb \smallstep[] b
      }

      \bbrule{Pattern}{
        \beta = \textsc{matches}(\eval, p)
      }{
        \matches{\eval}{p} \smallstep[] \beta
      }

  \end{mathpar}
  \caption{Operational Semantics for the Core Language}
  \label{fig:core_op_sem}
\end{figure}

\subsection{Modeling Types}

In this section, we first define a standard semantic type notion $\models \hastype{\expr}{\syntype}$ meaning $\expr$ semantically has type $\syntype$.   We then define our type checking relation $\tc{\expr}{\syntype}$ meaning $\expr$ type checks at type $\syntype$, and show it is co-r.e.

\subsubsection{Semantic Typing.} 
The semantic typing relation is defined as follows:

\begin{definition}[Semantic Types for the Core Language]
  \label{def:lr_core}
  \ \par
  \begin{enumerate}
      \item $\ \models \hastype{\expr}{\tint}$ iff $\expr \nsmallsteps[] \eerror$ and $\forall \eval.$ if $\expr \smallsteps[] \eval$, then $\eval \in \mathbb{Z}$.
      \item $\ \models \hastype{\expr}{\tbool}$ iff $\expr \nsmallsteps[] \eerror$ and $\forall \eval.$ if $\expr \smallsteps[] \eval$, then $\eval \in \mathbb{B}$.
      \item $\ \models \hastype{\expr}{\mkfun{\syntype_1}{\syntype_2}}$ iff $\expr \nsmallsteps[] \eerror$ and $\forall \eval_f.$ if $\expr \smallsteps[] \eval_f$, then $\forall{\eval}.$ if $\models \hastype{\eval}{\syntype_1}$, then $\models \hastype{\eval_f\ \eval}{\syntype_2}$.
   \end{enumerate}
\end{definition}

This is similar to the standard inductive definition of semantic typing \cite{MILNER1978348} and unary logical relations \cite{TaitRealizability,StepIndexedLogicalRelations,plotkin_lambda_1973}. There are a few differences from the standard presentation seen in these works: (1) we choose to only model closed expressions for simplicity and have no need to separate values from expressions in the definitions; (2) we choose to include soundness (i.e. well-typed expressions cannot evaluate to \lstinline|ERROR|) as part of the definition because this work is not trying to prove the correctness of these definitions but rather tries to use semantic typing to perform type checking over an untyped core language.

The semantic typing relation is not decidable because it requires enumeration of infinite function domains in the ``$\forall v$'' in the definition of a function type.  We will establish that it is in fact co-r.e. later in this section.


\subsubsection{Defining a Type Checker.}

Now we will give a co-r.e. definition of type checking based on the outline given in Section \ref{sec_overview}.  First, we will define the checkers and generators for each type $\syntype$.

\begin{definition}[Checker for Core]
  \label{def:checker_core}
  \ \par
  \begin{enumerate}
      \item \mkche{\tint}{\expr} = $\expr$ $\tttilde$ \tint
      \item \mkche{\tbool}{\expr} = $\expr$ $\tttilde$ \tbool
      \item $\mkche{\mkfun{\syntype_1}{\syntype_2}}{\expr}$ = 
      \lstinline|if $\expr\ \tttilde$ fun then let arg = |\mkgen{$\syntype_1$}\ \lstinline| in |\mkche{$\syntype_2$}{\lstinline|$\expr\ $ arg|} \lstinline| else false|
   \end{enumerate}
\end{definition}

\begin{definition}[Generator for Core]
  \label{def:generator_core}
  \ \par 
  \begin{enumerate}
    \item \mkgen{\tint} = \lstinline|pick_i|
    \item \mkgen{\tbool} = \lstinline|pick_b|
    \item \mkgen{$\mkfun{\syntype_1}{\syntype_2}$} = \lstinline|fun x -> $\ $if| \pickb \ \lstinline|then|
   \lstinline|if| $\mkche{\syntype_1}{\texttt{x}}$ \lstinline|then| $\mkgen{\syntype_2}$ \lstinline|else ERROR|
    
    \hspace*{1.75in}\lstinline|else| $\mkgen{\syntype_2}$
  \end{enumerate}
\end{definition}

Note that for a fixed type $\syntype$, the checker or generator is simply an expression in the core language.  Intuitively, checkers are functions that determine whether an expression has a declared type and return \lstinline|false| if the checking fails, and generators are expressions that produce arbitrary values from the specified type.

The base cases are straightforward: the checker performs typecasing on the expression, and the generator is the nondeterministic \lstinline|pick| for the corresponding type.
Function types are more interesting. To check a function type, we first generate an argument value of the input type and then check whether the application result has the output type. 
A complete enumeration is obtained via the nondeterminism in the generators.

There are two things of note regarding the generator definition. First, the generator checks whether the given argument has the correct input type, but only non-deterministically. The check is not always performed because for the case that the argument is a function, it might be that the argument is unused but diverges upon calling (e.g. the argument is \lstinline|fun x -> $\ \Omega \ \Omega$|), and we do not want the argument check to trigger a divergence that otherwise would not arise. The nondeterministic check ensures that the system will be able to catch both use errors and type errors that are only discoverable if the diverging argument is never invoked. Second, if the argument check passes or is not performed, the generator will produce an arbitrary value in the output type. As a result, the outputs of our generated functions are independent of their inputs. This implies that function generators only aim to capture the type correctness aspect of  functions. Essentially, for each run of a program containing a concrete function value of the specified type, if we replace this function with a generated function, there will be some set of specific choices for each nondeterministic \lstinline|pick| that will result in the same observable program behavior as the concrete function.

Finally, we can define the type checking relation, {\sc tc}.  This definition is straightforward: an expression has type $\syntype$ if and only if the checker for $\syntype$ at $\expr$ returns at most \lstinline|true|. 

\begin{definition}[Type Checker for Core]
  \label{def:tc_core}
  $\tc{\expr}{\syntype}$ iff \mkche{$\syntype$}{\expr} $\nsmallsteps[] \eerror$ and $\forall \eval. $ if \mkche{$\syntype$}{\expr}$ \smallsteps[] \eval$, then $\eval = \gttrue$.
\end{definition}

It is easy to show that this relation is co-r.e. given this definition.

\begin{restatable}{lemma}{tccore}
  \label{lem:tc_co_re}
  $\tc{\expr}{\syntype}$ is co-r.e.
\end{restatable}

\begin{proof}
  To show a relation is co-r.e.~we need to show all counterexamples can be exhaustively enumerated.  For this relation we dovetail enumeration of all $\expr,\syntype$ pairs with the computation of Definition~\ref{def:tc_core}, reporting each \lstinline|ERROR| or \lstinline|false| case as a failure of the relation as it arises.  Since all failures must terminate individually, the dovetailing will in the limit enumerate all such cases.
\end{proof}

\subsection{Soundness and Completeness}

We now show that the two typing definitions above are equivalent by proving the following theorem. This theorem establishes that our $\tc{\expr}{\syntype}$ definition fully captures the ``typedness'' property dictated by the semantic typing model and thus can be safely used as the basis for type refutation.

\begin{theorem}[Soundness and Completeness]
  \label{thm:sound_and_complete}$\forall \expr, \syntype.$ $\tc{\expr}{\syntype}$ iff $\models \hastype{\expr}{\syntype}$.
\end{theorem}

We will first establish completeness by proving the following more general lemma.  Note that proofs not given here are found in \inshort{Appendix~A of the supplementary material}\infull{Appendix~\ref{app_proofs}}.

  \begin{restatable}{lemma}{completenessextra}
  \label{lemma:completeness_extra}
  For all types $\syntype$,
  \begin{enumerate}
      \item $\models$ \mkgen{\syntype} $\colon$\syntype, and 
      \item $\forall \expr.$ if \mkche{\syntype}{\expr} $\smallsteps[] \eerror$ or if \mkche{\syntype}{\expr}$\smallsteps[] \gtfalse$, then $\not \models$ \expr $\colon$\syntype.
  \end{enumerate}
\end{restatable}

\begin{restatable}[Completeness]{lemma}{completeness}
  \label{lemma:completeness}$\forall \expr,\syntype.$ if $\models \hastype{\expr}{\syntype}$, then $\tc{\expr}{\syntype}$.
\end{restatable}

\begin{proof}
  This is equivalent to showing: if \mkche{\syntype}{\expr} $\smallsteps[] \eerror$ or if $\exists\eval. \mkche{\syntype}{\expr} \smallsteps[] \eval$ and $\eval \neq \gttrue$, then $\not\models \hastype{\expr}{\syntype}$. By examining Definition~\ref{def:checker_core}, we can see that \textsc{checker} can only return \lstinline|ERROR| or boolean values, making the completeness statement follow immediately from Lemma~\ref{lemma:completeness_extra}.
\end{proof}

Soundness is considerably more challenging than completeness: we need to show that the generators indeed exhaustively simulate all correctly-typed concrete function arguments.  This requires some special notation to factor a computation into the uniform parts between the actual/simulated run and the holes which are filled by either actual or simulated run. The full definition of additional notations can be found in \inshort{Appendix~A of the supplementary material}\infull{Appendix~\ref{app_proofs}}.  With these definitions we can establish the following lemma:

\begin{restatable}{lemma}{soundnessextra}
  \label{lemma:soundness_extra}
  For all types $\syntype$, 
  \begin{enumerate}
    \item $\forall \eval.$ if $\models\hastype{\eval}{\syntype}$, then $\forall\exprcon.$ if $\exprcon[\eval] \smallsteps[] \eerror$, then $\exprcon[\mkgen{\syntype}] \smallsteps[] \eerror$.
    \item $\forall\expr.$ if $\expr \smallsteps[] \eval$ and $\not\models\hastype{\eval}{\syntype}$, then $\neg\tc{\expr}{\syntype}$.
  \end{enumerate}
\end{restatable}

The $\exprcon$'s in this Lemma are standard expression contexts. Informally, expression contexts are expressions with holes in which expressions can be placed. Their formal definition as well as the proof of this Lemma can be found in \inshort{Appendix~A of the supplementary material}\infull{Appendix~\ref{app_proofs}}. With this Lemma we then may establish soundness as a direct corollary.

\begin{restatable}[Soundness]{lemma}{soundness}
  \label{lemma:soundness}
  $\forall \expr, \syntype.$ if $\tc{\expr}{\syntype}$, then $\models \hastype{\expr}{\syntype}$.
\end{restatable}

So, we finally have both soundness and completeness.
\begin{proof}[Proof of Theorem \ref{thm:sound_and_complete}]
  The forward implication follows from Lemma~\ref{lemma:completeness} and the reverse from Lemma~\ref{lemma:soundness}.
\end{proof}

Given the equivalence of the two definitions, we can obtain co-r.e.-ness of the semantic typing relation as a corollary.

\begin{corollary}
  The $\models \hastype{\expr}{\syntype}$ relation is co-r.e.
\end{corollary}
\begin{proof}
  Immediate by combining Theorem \ref{thm:sound_and_complete} and Lemma \ref{lem:tc_co_re}.
\end{proof}

This property would not be straightforward to prove from the semantic typing definition alone: e.g. for the definition of an expression \lstinline|g| having type \lstinline|(int ->$\ $int)$\ $->$\ $int| we would need to enumerate all functions \lstinline|f| with $\models$ \lstinline|f $\colon$int -> int|, but co-r.e.-ness only provides a refuter for such a judgment.  Fortunately, Theorem \ref{thm:sound_and_complete} shows semantic typing to be isomorphic to our checker, which is co-r.e. via Lemma \ref{lem:tc_co_re}, so the result immediately follows.  A nontrivial compactness property is thus hiding in Theorem \ref{thm:sound_and_complete}; in particular, the nondeterminism of our generator provides that compactness.
\section{Sound Language Extensions}
\label{sec_sound_extension}

In the last section, we showed that type checking using generators and checkers is equivalent to semantic typing for integer, boolean, and function types. 
In this section, we will demonstrate how this technique can be extended to refinement types, dependent function types, and parametric polymorphism.  We will establish soundness for these extensions using similar techniques as was used for the core language of the previous section.  This section does not include all of the extensions in our implementation; in Section \ref{sec_additional_extension}, we define several other extensions we have implemented but have not yet proven sound.

\begin{wrapfigure}{r}{.5\textwidth}
  \begin{grammar}
    \grule[expressions]{\expr}{
      \ldots
      \gor \mzero
    }
    \grule[types]{\syntype}{
      \ldots
      \gor \mktset{\syntype}{\expr}
      \gor \mkdfun{\ev}{\syntype}{\syntype}
    }
  \end{grammar}
  \caption{Extended Grammar with Refinement and Dependent Function Types}
  \label{fig:grammar_refinement_dependent}
\end{wrapfigure}

\subsection{Refinement and Dependent Function Types}

To introduce refinement and dependent function types into the system, we need to extend the language grammar as is shown in Figure~\ref{fig:grammar_refinement_dependent}.

\begin{wrapfigure}{r}{.25\textwidth}
  \begin{minipage}{0in} 
      \begin{mathpar}
            \bbrule{Mzero}{
          \expr \smallstep[] \mzero
        }{
          R\lbbr \expr \rbbr \smallstep[] \mzero
        }
    \end{mathpar}
  \end{minipage}
  \caption{Additional Operational Semantics Rules for $\mzero$}
  \label{fig:mzero_op_sem} 
\end{wrapfigure}

The new \lstinline|mzero| expression defined in Figure~\ref{fig:mzero_op_sem} is a dual to \lstinline|ERROR|. Instead of indicating the presence of a runtime error caused by the user, it signifies that the system itself has made a mistake, and this particular execution is invalid and thus can be safely discarded.  \lstinline|mzero| plays a similar role to \lstinline|assume| in other programming languages.  It is in fact equivalent to divergence, so \lstinline|mzero| can take on any type in our theory just like a diverging program can. We can use \lstinline|mzero| instead of an acually-diverging computation to immediately detect the divergence.

Now we will present semantic typing relations for refinement and dependent function types:

\begin{definition}[Semantic Typing for Refinement and Dependent Function Types] We extend Definition~\ref{def:lr_core} with the following clauses:
  \label{def:lr_refinement_dependent}
  \ \par
  \begin{enumerate}
    \setcounter{enumi}{3}
    \item $\ \models \hastype{\expr}{\mktset{\syntype}{\expr_p}}$ iff $\expr \nsmallsteps[] \eerror$, $\models \hastype{\expr_p}{\mkfun{\syntype}{\tbool}}$, and $\forall \eval.$ if $\expr \smallsteps[] \eval$, then $\ \models \hastype{\eval}{\syntype}$ and $\forall \eval_p.$ if $\expr_p \ \eval \smallsteps[] \eval_p$, then $\eval_p = \gttrue$.
    
    \item $\ \models \hastype{\expr}{\mkdfun{\ev}{\syntype_1}{\syntype_2}}$ iff $\expr \nsmallsteps[] \eerror$, and $\forall{\eval_f}.$ if $\expr \smallsteps[] \eval_f$, then $\forall \eval.$ if $\models \hastype{\eval}{\syntype_1}$, then $\models \hastype{\eval_f\ \eval}{\substitute{\syntype_2}{\eval}{\ev}}$.
  \end{enumerate}
\end{definition}

This extended definition is unsurprising: for refinement types, we require the expression to both be in the base type $\syntype$ as well as passing the predicate $p$. For dependent function types, the definition is very similar to that of regular function types, where the only difference is that the argument value is substituted into the output type. In this version of the theory, we also restrict the predicate functions in refinement types to be total, deterministic functions for simplicity (while it may be sound to include such predicates, it complicates the proofs).


We also need to extend Definitions~\ref{def:checker_core} and \ref{def:generator_core} with the following new clauses, respectively.

\begin{definition}[Type Checker with Refinement and Dependent Function Types] 
  \label{def:checker_refinement_and_dependent}
  \ \par 
  \begin{enumerate}
    \setcounter{enumi}{3}
    \item \mkche{$\mktset{\syntype}{\expr_p}$}{\expr} = \mkche{\syntype}{\expr}\ \ \lstinline| and | ($\expr_p$ \ \expr)
    \item \mkche{$\mkdfun{\ev}{\syntype_1}{\syntype_2}$}{\expr} = \\\lstinline|if $\expr\ \tttilde$ fun then let arg =|\ \mkgen{$\syntype_1$}\ \lstinline|in| \mkche{\substitute{$\syntype_2$}{\lstinline|arg|}{\ev}}{\lstinline|($\expr\ $ arg)|}\ \lstinline| else false|
  \end{enumerate}
\end{definition}

Notably, in the refinement type case above, the checker is defined by performing a logical \lstinline|and| on the result of the base type check and the predicate check.

\begin{definition}[Type Generator with Refinement and Dependent Function Types]
  \label{def:generator_refinement_and_dependent}
  \ \par 
  \begin{enumerate}
    \setcounter{enumi}{3}
    \item \mkgen{$\mktset{\syntype}{\expr_p}$} = \lstinline|let gend = | \mkgen{\syntype} \lstinline| in if ($\expr_p$ gend) then gend else mzero|
    
    \item \mkgen{$\mkdfun{\ev}{\syntype_1}{\syntype_2}$} = \lstinline|fun $\ev'$ ->$\ $if| \pickb \ \lstinline|then|
    
    \quad \lstinline|if| $\mkche{\syntype_1}{\ev'}$ \lstinline|then| $\mkgen{\substitute{\syntype_2}{\ev'}{\ev}}$ \lstinline|else ERROR| 
    
    \lstinline|else| $\mkgen{\substitute{\syntype_2}{\ev'}{\ev}}$
  \end{enumerate}
\end{definition}

For refinement types, the generator will first produce an arbitrary value in the base type, and then it will check whether this value satisfies the predicate. If the predicate check fails then we return \lstinline|mzero|, which signals that this particular execution is invalid. The generator for dependent function types is similar to that of non-dependent function types--the only change needed is to make sure that the argument value can be used in the output type, $\syntype_2$.

We will prove the extended theory sound and complete for a combined system that incorporates refinement and dependent function types in the next section.

\subsection{Parametric Polymorphism}

This section shows how to extend the system with parametric polymorphism.  We only type check prenex polymorphism because full higher-order polymorphism is beyond co-r.e. complexity. We follow OCaml convention and leave the quantifiers implicit in the syntax here.

%
\begin{wrapfigure}{r}{.4\textwidth}
  \begin{grammar}
          \grule[expressions]{\expr}{
            \ldots
            \gor \expr \simeq \tpoly
          }
          \grule[values]{\eval}{
            \ldots
            \gor \texttt{\small V(}\alpha\texttt{\small)}
          }
          \grule[type variables]{\tpoly}{
            $\lstinline`'a`\ $
            \gor $\lstinline`'b`\ $
            \gor \ldots
          }
          \grule[types]{\syntype}{
            \ldots
            \gor \tpoly
        }
      \end{grammar}
  \caption{Extended Grammar for Parametric Polymorphism}
  \label{fig:para_poly_grammar}
\end{wrapfigure}

Intuitively, since parametric functions must work uniformly for inputs of arbitrary type, we enforce uniformity by using singleton ``untouchable'' placeholder values \lstinline`V($\alpha$)`, which are essentially black boxes that can only be passed around but not examined or operated upon. For each type variable $\alpha$, the singleton member is \lstinline`V($\alpha$)`. 

\begin{wrapfigure}{r}{.4\textwidth}
  \begin{minipage}{0in} 
      \begin{mathpar}
            \bbrule{Poly-Check-True}{
          \eval = \texttt{V}\ttop\alpha\ttcp
        }{
          \eval \simeq \alpha \smallstep[] \vtrue
        }
  
        \bbrule{Poly-Check-False}{
          \eval \neq \texttt{V}\ttop\alpha\ttcp
        }{
          \eval \simeq \alpha \smallstep[] \vfalse
        }
  
        \bbrule{Opaque Pattern}{
        }{
          \matches{\texttt{V}\ttop\alpha\ttcp}{p} \smallstep[] \eerror
        }
    \end{mathpar}
  \end{minipage}
  \caption{Additional Operational Semantics Rules for Polymorphism}
    \label{fig:para_poly_op_sem} 
  \end{wrapfigure}

These untouchable values serve only for type checking and should not be found in user code.  To verify parametricity, we need to make sure the correct singleton that went in came out; this correspondence will be checked by the new expression, $\expr \simeq \alpha$ in Figure~\ref{fig:para_poly_grammar}, which also requires the addition of the operatational semantics rules in Figure \ref{fig:para_poly_op_sem} to the rules of Figure \ref{fig:core_op_sem}.

The approach we take here bears some resemblance to the \emph{sealing} used in dynamic contracts to preserve parametricity \cite{ParametricContracts}.  In a dynamic contract system, the instantiating types are concrete: for example the identity function is never type checked without being instantiated at a particular concrete type. Such an approach is not viable in our modular, statically-checked approach: to refute that a function $f$ has a type such as $\forall \alpha. \alpha \rightarrow \alpha$ would require finding a particular concrete type $\tau$ such that $f$ fails to have the type $\tau \rightarrow \tau$, and that would require an enumeration of all types.  For this reason, we abstract it further to an untouchable singleton.

With these new definitions in place, we are ready to define the semantic typing relation for parametric polymorphism.

\begin{definition}[Semantic Typing for Parametric Polymorphism] We extend Definition~\ref{def:lr_core} with the following clause:
  \label{def:lr_poly}
  \ \par
  \begin{enumerate}
    \setcounter{enumi}{3}
    \item $\ \models \hastype{\expr}{\tpoly}$ iff $\expr \nsmallsteps[] \eerror$ and $\forall \eval.$ if $\expr \smallsteps[] \eval$, then $\eval \simeq \alpha \smallsteps[] \vtrue$.
  \end{enumerate}
\end{definition}

Consider refuting that the function \lstinline!f = $\ $fun x -> if x ~ int then x + 1 else x! has type \lstinline`'a ->$\ $'a`. According to Definition~\ref{def:lr_core}, we have to check whether for any value in the input type, \lstinline`'a`, the application result will also be in the output type. That is to say, if we pass an untouched value into a polymorphic function and get the same value in return, then we know that the function behavior must be independent of the input's actual type. Since \lstinline`'a` has only one unique value, \lstinline`V('a)`, we only need to check whether \lstinline`f$\ $V('a) $\smallsteps[]$ V('a)`.  However, by the \ruleref{Opaque Pattern} rule of Figure \ref{fig:para_poly_op_sem}, \lstinline`V('a) $\ \tttilde$ int` returns \lstinline|ERROR| and so the overall result will be \lstinline|ERROR|, and the typing fails.  Note that the program is in fact ``fine'' in the sense that it will not generate any runtime errors, but it violates the required parametricity of polymorphism.  An alternative definition of semantic typing for polymorphic types is that the typing holds for all (infinite) possibilities of concrete type instantiations of \lstinline|'a|, and under that notion of semantic typing, this program would be type-correct.  Unfortunately the infinite nature of the above assertion makes refutation challenging: all concrete types must be enumerated.

Now, we will provide the checker and generator definitions for polymorphic types.

\begin{definition}[Type Checker with Parametric Polymorphism] We extend Definition~\ref{def:checker_core} with the following clause:
  \ \par 
  \label{def:checker_poly}
  \begin{enumerate}
    \setcounter{enumi}{3}
    \item $\mkche{\tpoly}{\expr} = \expr \simeq \tpoly$
  \end{enumerate}
\end{definition}

\begin{definition}[Type Generator with Parametric Polymorphism] We extend Definition~\ref{def:generator_core} with the following clause:
  \ \par 
  \label{def:gen_poly}
  \begin{enumerate}
    \setcounter{enumi}{3}
    \item $\mkgen{\tpoly} = $ \lstinline`V($\alpha$)`
  \end{enumerate}
\end{definition}

The definitions match the intuitions discussed earlier, where the checker will only pass if the expression is the corresponding unique untouched value, and the generator in turn produces that unique value for each polymorphic type variable. 

To avoid redundancy and to demonstrate viability of multiple feature interactions, we prove equivalence between semantic typing and type checking for the core language extended with refinement, dependent function types and parametric polymorphism by establishing the following theorem:

\begin{theorem}[Soundness and Completeness of Extended System]
  \label{thm:sound_and_complete_ext}For all types $\syntype$ defined in Definitions~\ref{def:lr_core}, \ref{def:lr_refinement_dependent}, and \ref{def:lr_poly}, $\forall \expr.$ $\tc{\expr}{\syntype}$ iff $\models \hastype{\expr}{\syntype}$.
\end{theorem}

The proof of this theorem can be found in \inshort{Appendix~A.2 of the supplementary material}\infull{Appendix~\ref{app_proofs_ext}}.

\section{Additional Extensions}
\label{sec_additional_extension}

In this section we describe additional extensions we have implemented.  Several of them will be very challenging to prove sound, so we have left the task of soundness for the extensions in this section to future work.  We first define variant and intersection types. Next, we define record types and show how to treat types themselves as expressions. With types as expressions, we also show how we can model recursive types.  Lastly, we show how subtyping on records can be incorporated.  For simplicity, we add each feature independently and do not address interactions between these extensions (the implementation, of course, must deal with such feature interactions).

\subsection{Variant and Intersection Types}

We now describe how variants and restricted intersection types can be added.  We originally thought that general unions and intersections would be the simplest and most elegant approach, but they proved surprisingly difficult to model.  Before getting into our solution, we briefly discuss why they are so difficult.

The primary problem arises with positive (covariant) unions, and negative (contravariant) intersections.  Consider a simple positive union type $\mkche{\syntype_1 \cup \syntype_2}{\expr}$.  This checker will need to see if $\expr$ passes \emph{either} $\syntype_1$'s \emph{or} $\syntype_2$'s checker.  Since we have no parallelism, we must arbitrarily start with one or the other checker.  We cannot simply sequence the checkers because the first checker could return \lstinline|ERROR|, and since we lack exception handling, the union typing has already been rejected before it could complete.  Rather than add exception handling, we elected to restrict unions to a tagged form only, i.e. variants.

If we supported general intersections occuring negatively, for example $\mkche{\mkfun{\syntype_1 \cap \syntype_2}{\tint}}{\expr}$, we would need to invoke $\fnstyle{Generate}(\syntype_1 \cap \syntype_2)$ for such a checker given how function checkers are defined.  Consider how this generator could be defined: the only general approach is to (1) arbitrarily generate an element of one of the two types and (2) check if it is in the other type and \lstinline|mzero| if not.  But again, as with positive unions, the issue is that a typing failure could be an \lstinline|ERROR| state, and there is no way in our current language to catch an \lstinline|ERROR| and turn it into an \lstinline|mzero|.  To address this issue, we will restrict intersections to the narrow case of the intersection of functions where the domains are each distinct variants. Note that we could have also simply disallowed intersections instead of allowing this narrow form, but since this form is useful for object-oriented programming with a variant dispatcher view of an object \cite{first-class-messages}, we elected to include it.

We provide the syntax for the variant and intersection type extensions in Figure~\ref{fig:grammar_variant_intersect}. We require that each clause $V_i$ be distinct in any variant type and in any intersection of functions.

Below are the semantic typing definitions for variant and intersection types.

\begin{wrapfigure}{r}{.62\textwidth}
  \begin{grammar}
    \grule[constructors]{V}{
      \textit{(identifiers)}
    }
    \grule[types]{\syntype}{
      \ldots
      \gor (\mkvariant{\syntype_1}{\syntype_n})
      \gline
      \gor \mkintersect{(\mkfun{(V_1\ \texttt{of}\ \syntype_{1})}{\syntype_{1}'})}{(\mkfun{(V_n\ \texttt{of}\ \syntype_{n})}{\syntype_{n}'})}
    }
  \end{grammar}
  \caption{Extended Grammar with Variant and Intersection Types}
  \label{fig:grammar_variant_intersect}
\end{wrapfigure}

\begin{definition}[Semantic Typing for Variant and Intersections] We extend Definition~\ref{def:lr_core} with the following clauses:
  \ \par
  \label{def:lr_variant_and_intersection}
  \begin{enumerate}
    \setcounter{enumi}{3}
    \item $\ \models \hastype{\expr}{\mkvariant{\syntype_1}{\syntype_n}}$ iff $\expr \nsmallsteps[] \eerror$, and $\forall \eval.$ if $\expr \smallsteps[] \eval$, then $\exists 1\leq i \leq n.\ \eval = V_i (\eval')$ and $\models \hastype{\eval'}{\syntype_i}$.
    \item $\ \models \hastype{\expr}{\mkintersect{\syntype_1}{\syntype_n}}$ iff $\expr \nsmallsteps[] \eerror$, and $\forall \eval.$ if $\expr \smallsteps[] \eval$, then $\forall 1\leq i \leq n. \models \hastype{\eval}{\syntype_i}$.
  \end{enumerate}
\end{definition}

The semantic typing definitions are standard: an expression is in a variant type if it matches one of the constructors and if its value has the corresponding type; an expression is in the intersection of some types if it is in all of these types. 

We will now extend the checker definition for variant and intersection types.

\begin{definition}[Type Checker for Variant and Intersection Types] We extend Definition~\ref{def:checker_core} with the following clauses:
  \label{def:checker_union}
  \begin{enumerate}
    \setcounter{enumi}{3}
    \item \mkche{$\mkvariant{\syntype_1}{\syntype_n}$}{\expr} = 
    \vskip -.04in\begin{lstlisting}
match $\expr$ with | $V_1(\eval_1)$ -> $\mkche{\syntype_1}{\eval_1}$ | $\cdots$ | $V_n(\eval_n)$ -> $\mkche{\syntype_n}{\eval_n}$
    \end{lstlisting}
    \item $\mkche{\mkintersect{(\mkfun{(V_1\ \texttt{of}\ \syntype_{1})}{\syntype_{1}'})}{(\mkfun{(V_n\ \texttt{of}\ \syntype_{n})}{\syntype_{n}'})}}{\expr} = $ \\\lstinline| let i = |$\ \picki$ \lstinline| in if i = $ $ 1 then |$\mkche{\mkfun{(V_1\ \texttt{of}\ \syntype_{1})}{\syntype_{1}'}}{\expr}$
    \\\lstinline| else ... else if i = $\ n$ then |$\mkche{\mkfun{(V_n\ \texttt{of}\ \syntype_{n})}{\syntype_{n}'}}{\expr}$
    \\\lstinline| else mzero|
  \end{enumerate}
\end{definition}

The above definition illustrates why we choose to include variants rather than unions: having distinct constructors means we can always safely determine which checker needs to be run on the given expression and thus circumvent the problem that the type checking fails due to a prematurely raised error from an incorrectly chosen checker function. 

Next, let us look at the updated generator definitions.

\begin{definition}[Type Generator with Variant and Intersection Types] We extend Definition~\ref{def:generator_core} with the following clauses:
  \ \par 
  \label{def:generator_variant_intersect}
  \begin{enumerate}
    \setcounter{enumi}{3}
    \item \mkgen{$\mkvariant{\syntype_1}{\syntype_n}$} = 
    
    \lstinline|if pick_b then |$V_1(\mkgen{\syntype_1})$ \ \lstinline| else if pick_b then ... else |$V_n(\mkgen{\syntype_n})$    
    \item \mkgen{$\mkintersect{(\mkfun{(V_1\ \texttt{of}\ \syntype_{1})}{\syntype_{1}'})}{(\mkfun{(V_n\ \texttt{of}\ \syntype_{n})}{\syntype_{n}'})}$} = 
    \begin{lstlisting}
fun $\ev$ -> match $\ev$ with
| $V_1(\eval_1)$ -> if $\pickb$ then 
    if $\mkche{\syntype_1}{\eval_1}$ then $\mkgen{\syntype_1'}$ else ERROR
  else
    $\mkgen{\syntype_1'}$
...
| $V_n(\eval_n)$ -> if $\pickb$ then 
    if $\mkche{\syntype_n}{\eval_1}$ then $\mkgen{\syntype_n'}$ else ERROR
  else
    $\mkgen{\syntype_n'}$
| _ -> ERROR
    \end{lstlisting}
    \end{enumerate}
\end{definition}

The generator for variant types is intuitive: it nondeterministically generates one of the potential constructors of the given variant type. The generator for intersections of variant argument functions can use which particular variant $x$ shows up to dispatch on which kind of function to generate.

\subsection{Record Types and Types as Expressions}

\begin{wrapfigure}{r}{.5\textwidth}
  \begin{grammar}
  \grule[values]{\eval}{
    \ldots
    \gor \ttob l = \eval; \ \ldots\ttcb
  }
  \grule[expressions]{\expr}{
    \ldots
    \gor \ttob l = \expr; \ \ldots\ttcb
    \gor \expr.l
    \gor \syntype
  }
  \grule[patterns]{p}{
    \ldots
    \gor \ttob l; \ \ldots \ttcb
  }
  \grule[types]{\syntype}{
    \tint
    \gor \tbool
    \gor \tfun
    \gor \ttob \hastype{l}{\syntype}; \ \ldots \ttcb
  }
\end{grammar}
\caption{Language Grammar with Type as Expressions}
\label{fig:type_as_expr_grammar} 
\end{wrapfigure}

For more advanced types, it is necessary to model functions from types to types.  Two examples are recursive types (which are fixed points of functions from type to type) and parametric types.  Our approach has been to embed typehood judgements in the language, and continuing in this spirit we will also embed type functions in the language itself.  This will be achieved by embedding the generators and checkers themselves: a type will now be modeled as a tuple of generator and checker expressions. A type function then maps such tuples to tuples.

Concretely, we will introduce records to the core language here to express tuples as well as other data structures.  For simplicity, we will only modify and extend the core grammar of Definition~\ref{fig:core_grammar} and leave out the previous extensions of this section. Record types and the operational semantics for record creation and projection are standard. 

Starting from this subsection, there will be no semantic type basis defined; we leave investigation of the subject to future work. 

Now we define the translation from syntactic type $\syntype$ to an expression, denoted as $\llbracket\syntype\rrbracket$.

\begin{definition}[Embedding Types as Expressions]
  \label{def:sem_of_types}
  \ \par
  \begin{enumerate}
    \item $\llbracket \tint \rrbracket$ = \lstinline`{ gen =$\ $fun _ ->  pick_i; check = fun e ->  e $\tttilde$ int }`
    \item $\llbracket \tbool \rrbracket$ = \lstinline`{ gen =$\ $fun _ ->  pick_b; check = fun e ->  e $\tttilde$ bool }`
    \item $\llbracket \mkfun{\syntype_1}{\syntype_2} \rrbracket$ = 
    \begin{lstlisting}[gobble=4]
    { gen = fun _ -> fun arg -> if pick_b then 
        if ($\llbracket\syntype_1\rrbracket$.check arg) then ($\llbracket\syntype_2\rrbracket$.gen 0) else ERROR else ($\llbracket\syntype_2\rrbracket$.gen 0);
      check = fun e ->  
        if e ~ fun then let arg = ($\llbracket\syntype_1\rrbracket$.gen 0) in ($\llbracket\syntype_2\rrbracket$.check) (e arg) 
          else false }
    \end{lstlisting}
    \item $\llbracket\ttob\hastype{l_1}{\syntype_1}; \ldots; \hastype{l_n}{\syntype_n}\ttcb\rrbracket$ = 
    \begin{lstlisting}[gobble=4]
      { gen = fun _ -> {$l_1$ = ($\llbracket\syntype_1\rrbracket$.gen 0); $\ldots$; $l_n$ = ($\llbracket\syntype_n\rrbracket$.gen 0)};
        check = fun e ->
          if e ~ {$l_1$; $\ldots$; $l_n$} then
            if $\llbracket\syntype_1\rrbracket$.check e.$l_1$ then $\ldots$
              if $\llbracket\syntype_n\rrbracket$.check e.$l_n$ then true else false
              $\ldots$
          else false }
      \end{lstlisting}
  \end{enumerate}
\end{definition}

The clauses for $\tint$, $\tbool$, and \lstinline|$\syntype_1$ -> $\syntype_2$| are faithful to their corresponding checkers and generators, Definitions~\ref{def:checker_core} and \ref{def:generator_core}. For record types, the generator  produces values in the type by invoking each label's corresponding generator. The checker first uses pattern matching to ensure that the expression indeed has the right record layout and then proceeds to check whether each label contains values of the correct types. 


With this embedding, we redefine $\tc{\expr}{\syntype}$ to be as follows.
\begin{definition}
  \label{def_typecheck_new}
  $\tc{\expr}{\syntype}$ iff $\llbracket \syntype \rrbracket$\lstinline|.check $\expr$| $\nsmallsteps[]$ \lstinline|ERROR| and $\forall \eval.$ if $\llbracket \syntype \rrbracket$\lstinline|.check $\expr$| $\smallsteps[] \eval$, then $\eval = \vtrue$. 
\end{definition}

We will now go over the addition of recursive types and will provide examples of type functions and parametrized types when we discuss the implementation in Section \ref{sec_impl}. 

  \begin{wrapfigure}{r}{.4\textwidth}%
    \begin{grammar}
    \grule[type variables]{\beta}{
      \textit{(identifiers)}
    }
      \grule[types]{\syntype}{
        \ldots
        \gor \beta
        \gor \mkmiu{\beta}{\syntype}
    }
  \end{grammar}
  \caption{Extended Grammar for Recursive Types}
  \label{fig:recursive_grammar}
\end{wrapfigure}

\subsection{Recursive Types}

The syntax for recursive types is defined in Figure~\ref{fig:recursive_grammar}.

Recursive types are well-known to present challenges in semantic typing, in particular for the case of contravariant recursion leading to non-monotonically-increasing types, and step-indexed logical relations were developed for this purpose \cite{StepIndexedLogicalRelations}.  Using types-as-expressions, a recursive type checker can be defined simply by using the Y-combinator to take a fixed-point of a type function; the open question of the soundness of this approach may require something like step-indexed logical relations to resolve.

\begin{definition}[Embedding Recursive Types] 
    \label{def:recursive_type_embedding}
    We extend 
    Definition~\ref{def:sem_of_types} with the following clauses:
  \begin{enumerate}
  \setcounter{enumi}{4}
    \item $\llbracket \beta \rrbracket$ = $\beta$
    \item $\llbracket \mkmiu{\beta}{\syntype} \rrbracket$ = 
    \begin{minipage}{.75\textwidth}
      \begin{lstlisting}[gobble=3]
    Y (fun self -> fun _ -> 
      { gen = fun _ -> (fun $\beta$ -> $\llbracket\syntype\rrbracket$.gen 0) (self 0);
        check = fun e -> (fun $\beta$ -> $\llbracket\syntype\rrbracket$.check e) (self 0) })
    \end{lstlisting}
  \end{minipage}
  \end{enumerate}
\end{definition}

\subsection{Record Subtyping}
\label{subsec:rec_and_subtyping}

The subtyping syntax is an extension of Figure~\ref{fig:type_as_expr_grammar} and is given in Figure~\ref{fig:record_grammar}. Like OCaml module types, our system treats subtyping as abstraction, where non-listed fields are fundamentally unobservable.

To model subtyping, record values now contain an extra piece of information: the \emph{declared labels} shown as a superscript \ttob$l_1; \ldots; l_m$\ttcb. This additional set of labels represent which labels can be safely accessed. By default, an untyped record value will contain all actual labels in the declared labels set. More details on the updated operational semantics can be found in \inshort{Figure~18 in Appendix~C of the supplementary material}\infull{Figure~\ref{fig:records_op_sem} in Appendix~\ref{app_record}}.

\begin{wrapfigure}{r}{.55\textwidth}
  \begin{grammar}
  \grule[redexes]{r}{
      \ldots
      \gor \eval.l
      \gor \retag \ttop \eval, \ttob l; \ldots \ttcb\ttcp
  }
  \grule[reduction contexts]{\reductioncon}{
      \ldots
      \gor R.l 
      \gor \retag \ttop \reductioncon, \ttob l; \ldots \ttcb\ttcp
      \gline
      \gor \ttob l = \eval; \ldots; l = \eval; l = \reductioncon; l = \expr; \ldots; l = \expr\ttcb
  }
  \grule[expressions]{\expr}{
    \ldots
    \gor \retag \ttop\expr, \ttob l; \ldots \ttcb\ttcp
  }
  \grule[values]{\eval}{
    \ldots
    \gor \ttob l = \eval; \ldots \ttcb^{\ttob l; \ldots \ttcb}
  }
\end{grammar}
\caption{Extended Grammar for Record Types and Subtyping}
\label{fig:record_grammar}
\end{wrapfigure}

Additionally, we introduce a new operation \retag\ into the language. This operation is used in the type-checking instrumentation only, meaning users will not be able to use it in their source code. \retag\ allows a record value to take on a new declared labels set, but only if the new set is a subset of the actual labels in the record. This operation is essential for ensuring soundness with respect to the semantic notion of subtyping-as-subsetting. The operational semantics rule for this new operation is shown in Figure~\ref{fig:retag_op_sem}. Additional rules pertaining to error cases can be found in \inshort{Figure~18 in Appendix~C in the supplementary material}\infull{Figure~\ref{fig:records_op_sem} in Appendix~\ref{app_record}}.

The declared labels set will always be a subset of all the actual labels in a record value. This invariance is guaranteed by the operational semantics: the only means by which a record value can change its declared labels is \retag, and the rule for \retag \ will reject any relabeling where the new set is not a subset of actual labels in the record. The reader can assume that any record value appearing in the operational semantics rules have this invariant holding.

The definition for \textsc{matches} also needs an update to accommodate subtyping. 

\begin{definition}We extend Definition~\ref{def:pattern_match} with the following clause:
  \begin{description}
    \item \textsc{matches}$(\eval, \ttob l_1; \ldots; l_m \ttcb) =   
    \begin{cases}
      \gttrue & \text{if}\ \eval = \ttob l_1 = \eval_1; \ldots; l_n = \eval_n \ttcb^{\ttob l_1; \ldots; l_k \ttcb} \text{ and } m \leq k\\ 
      \gtfalse & \text{if otherwise}
    \end{cases}$
  \end{description}
\end{definition}

Intuitively, this revised \textsc{matches} definition allows record values to match on record patterns that are ``less specific'' than their declared labels set. For example, \textsc{matches}$(\ttob l_1 = 1; l_2 = 2\ttcb^{\ttob l_1; l_2 \ttcb}, \ttob l_1 \ttcb)$ is $\gttrue$ because the pattern set $\ttob l_1 \ttcb$ strictly contains fewer labels than the declared labels set $\ttob l_1; l_2 \ttcb$.


\begin{figure}
  \begin{mathpar}
    
      \bbrule{Retag}{
        \eval = \ttob l_1 = \eval_1; \ldots; l_n = \eval_n \ttcb^{\ttob l'_1; \ldots; l'_k \ttcb} \\
        m \leq n
      }{
        \retag \ttop\eval, \ttob l_1; \ldots; l_m \ttcb\ttcp \smallstep[]  \ttob l_1 = \eval_1; \ldots; l_n = \eval_n \ttcb^{\ttob l_1; \ldots; l_m \ttcb}
      }
  
  \end{mathpar}
\caption{Operational Semantics Rule for Retag}
  \label{fig:retag_op_sem} 
\end{figure}
\section{Implementation}
\label{sec_impl}
In this section we describe the current status of our type checker implementation.  There are two components to the type checker: the front-end translator, which produces instrumented code as described in Sections \ref{sec_formalization} through \ref{sec_additional_extension}, and the back end, which takes an instrumented program and searches for inputs that produce \lstinline|ERROR|. 


It is worth noting that our type checking framework is not dependent on the choice of underlying back end. The translator produces programs for which all that is needed is a back end to discover integer input streams that lead these programs to runtime \lstinline|ERROR|.  Thus far, we have experimented with symbolic and concolic back-ends; potential additions for future work include property-based testing and abstract interpretation. In practice, a realistic implementation may need to use all of these approaches in unison because they each have trade-offs. 

\subsection{The Bluejay Language}
The implemented Bluejay language includes all the features defined in Sections \ref{sec_formalization} through \ref{sec_additional_extension} as well a built-in list type (polymorphic lists can be defined as a parameterized recursive type, but we include a built-in list type for efficiency). The language also includes syntactic sugar to make it easier for users to write programs, such as replacing $\expr \ \tttilde \ p$ with \lstinline`match $\expr$ with $p$ -> $\cdots$`. Bluejay syntax is similar to OCaml on the features they have in common. The full grammar can be found in \inshort{Figure~17}\infull{Figure~\ref{fig:bluejay_grammar}} of \inshort{Appendix~B in the supplementary material.}\infull{Appendix~\ref{app_bjy_syntax}.} Bluejay is untyped by default, but it allows users to selectively provide type annotations on expressions that they wish to statically type check.
    
\subsubsection*{Primitive Operations.}To align our implementation with the theory, we use type instrumentations to guard against misuse of primitive operations: each operation will first check to make sure the types of their operands are as expected, and if the check fails, the expression will evaluate to \lstinline|ERROR|. For example, the expression \lstinline|not e| will be transformed into \lstinline`match e with | bool ->$\ $not e | any ->$\ $ERROR` after the instrumentation. As with user-declared-type checking, we can use the back end to find whether any such \lstinline|ERROR| is reachable to detect primitive type errors.

\subsection{The Type Checking Process}

In this section, we give an overview of the process of type checking. As outlined above, the two steps are running the instrumentation front-end, followed by running the \lstinline|ERROR|-finding back end on the resulting instrumented code.

In the instrumentation step, we translate all declared types in a program into expressions along the lines of  Definition~\ref{def:sem_of_types}. We will then transform the type declarations into invocations on the checker expressions for the declared types.

Take the following very simple program as an example:

\begin{lstlisting}
  let id (x : bool) : bool = 1 in id
\end{lstlisting}

This essentially translates to checking the statement $\tc{\mbox{\tt \small id}}{\mbox{\tt \small bool -> bool}}$. By Definition~\ref{def_typecheck_new}, the type checking code will be:

\begin{lstlisting}
  let id x = 1 in let check_id = $\llbracket$bool -> bool$\rrbracket$.check id in if check_id then id else ERROR
\end{lstlisting}


\begin{figure}
  \begin{minipage}{0.49\textwidth}
    \begin{lstlisting}
      let id x = 1 in 
      let check_id = 
        let arg =  
         { check = ... ; 
           gen = fun _ -> input >= 0 }
        .gen 0  
        in
        { check = fun expr ->
            match expr with
            | bool -> true
            | any -> false
        ; gen = ...
        }.check (id arg)
      in if check_id then id else ERROR
    \end{lstlisting}
    \caption{Fully Transformed Simple Example}
    \label{fig:tc_full_code}
  \end{minipage}
  \begin{minipage}{0.49\textwidth}
    \begin{lstlisting}
      let length x =
        let rec loop l acc = 
            match l with
            | [] -> acc
            | hd :: tl -> loop tl (acc + 1)
        in loop x 0
      in 
      let rec prepend (type a) (x : list a) : 
        ((y : list a) -> 
         { list a | 
           (fun r -> (length r) == 
              (length x) + (length y)) }) 
      = fun y -> let rec loop l acc = 
            match l with
            | [] -> acc
            | hd :: tl -> 
                loop tl (hd :: (hd :: acc))
        in loop x y 
      in prepend
      \end{lstlisting}
      \caption{List Prepend Example}
      \label{fig:list_prepend}  
  \end{minipage}
  \end{figure}

If we expand $\llbracket$\lstinline|bool ->$\ $bool|$\rrbracket$ according to Definition~\ref{def:sem_of_types}, we obtain the fully transformed program, shown in Figure \ref{fig:tc_full_code}.
After the transformation, we run the back-end analysis engine to find if there is a viable path that can reach \lstinline|ERROR| from the top of the program. The back-end engines we experimented with will be covered below.


The back end automatically conducts searches for any input sequence leading to an \lstinline|ERROR|. There are three potential outcomes: (1) inconclusive, because the search times out; (2) conclusive and no errors found, meaning the analysis has exhausted all possible execution paths and can safely conclude that there are no type errors in the program; (3) conclusive and an error is found, where a valid execution trace is discovered. In the latter case, our type checker will report the type error to the user with information about the error location (i.e. which type declaration failed to type check), the expected type, and the actual type.  The output from the previous example is:

\bgroup\scriptsize
\begin{verbatim}
** Bluejay Type Errors **
- Found at clause : let id (x : bool) : bool = 1 in id
--------------------
* Value    : id
* Expected : (bool -> bool)
* Actual   : (bool -> int)
\end{verbatim}
\egroup

\subsection{Use Checking with Wrappers}
Recall that the semantic typing definition does not specify the function's behavior when given an argument that is not in the domain type; it clearly is asserting \emph{if} a value is in the function domain \emph{then} application produces a value in the codomain, but the definition is completly silent when used at values \emph{not} in its domain.  In other words, semantic typing only verifies if a typed function is \emph{defined} properly but is silent on whether it is \emph{used} properly.  For example, an identity function typed on integers but used on a boolean, \lstinline!let id x : int  = $ $ x in id true!, does not constitute an error with respect to our formalization of the meaning of types of the previous sections. Whether such programs should constitute type errors depends on how seriously the programmer wishes to take type interfaces; a strict view, and the one realized in all standard type systems, is that the interface must not be violated, and the above program should be a type error.

Fortunately, it is not hard to add an additional layer of checking to verify that semantically-typed functions must only be passed values that are in their domain types.  Such use checking has been thoroughly studied because it is a key component of contract checking \cite{Findler-Contracts,ParametricContracts,nguyen_tobin-hochstadt_vanhorn_2017}: contracts don't verify that a function typing is correct across its whole domain, contracts only verify correctness for specific uses of the function. We generally follow contracts in how uses are verified, so we did not formalize it; we will give an overview here of our implementation of use checking.  Note that the implementation has a command-line flag to let the programmer choose whether or not they want uses to be verified -- some programmers may wish to take a looser view of what type interfaces mean and forego use checking.

\subsubsection*{Core Language Wrappers.}

We will model use checking by taking each typed function (or value) definition and export only a wrapped version, which will check all of the arguments passed.  Here is the definition of the wrappers added to typed values of the core language.
\begin{definition}[Function Use Wrappers]
  \label{def:wrap_core}
  \ \par
  \begin{enumerate}
      \item \mkwrap{\tint}{\expr} = \expr
      \item \mkwrap{\tbool}{\expr} = \expr
      \item \mkwrap{$\mkfun{\syntype_1}{\syntype_2}$}{\expr} = \lstinline|fun x ->$\ $if| $\pickb$ \lstinline|then|
      
      \quad \lstinline|if| $\mkche{\syntype_1}{\texttt{x}}$ \lstinline|then| $\mkwrap{\syntype_2}{\expr \ \mkwrap{\syntype_1}{\texttt{x}}}$ \lstinline|else ERROR|\\
      \lstinline|else| $\mkwrap{\syntype_2}{\expr \ \mkwrap{\syntype_1}{\texttt{x}}}$ 
   \end{enumerate}
\end{definition}

Wrapping only checks function arguments and so is a no-op for integers and booleans because they are not functions.  Note that like the function type generator, the argument check in function wrappers is also performed non-deterministically to prevent triggering divergence in arguments that might prevent us from finding  errors.  The function case is recursive because the argument and return types may themselves be functions that need to have their uses verified; this recursive wrapping follows the methodology of dynamic contracts on higher-order functions \cite{Findler-Contracts}.
To show how this \fnstyle{Wrapper} function is used, we will give a small example of a source program and how it is wrapped.

\begin{lstlisting}
  let f : (int -> int) = fun x -> x + 1 in f true
\end{lstlisting}

The function itself will be typechecked using the principles of the previous section. Here we only focus on how uses are checked.  To check the \lstinline!f true! application we will export \lstinline!f! as a wrapped function so the application will check the argument type:

\begin{lstlisting}
  let _f : (int -> int) = fun x -> x + 1 in (* internal version of f *)
  let f = $\fnstyle{Wrapper}$(int -> int,_f) in (* wrapped version exported for use *)
  f true  (* argument will be checked against int via $\fnstyle{Checker}$(int,true) in $\fnstyle{Wrapper}$, so typing fails *)
\end{lstlisting}

All user-defined typed values will be so wrapped. We will now briefly describe how Definition \ref{def:wrap_core} is extended for the additional features in the implementation.

\subsubsection*{Wrapping of Dependent and Refinement Types.}
For this and other extensions, we will show what clauses are added to the definition of \fnstyle{Wrapper} of Definition \ref{def:wrap_core}.
\begin{definition}[Wrapper for Refinement and Dependent Function Types]
  \label{def:wrap_ref_and_dep}
  \ \par 
  \begin{enumerate}
    \setcounter{enumi}{3}
    \item $\mkwrap{\mktset{\syntype}{\expr_p}}{\expr} = \mkwrap{\syntype}{\expr}$
    \item $\mkwrap{\mkdfun{\ev}{\syntype_1}{\syntype_2}}{\expr} = $ \lstinline|fun $\ev'$ ->$\ $if| $\pickb$ \lstinline|then|
      
    \quad \lstinline|if| $\mkche{\syntype_1}{\ev'}$ \lstinline|then| $\mkwrap{\substitute{\syntype_2}{\ev'}{\ev}}{\expr \ \mkwrap{\syntype_1}{\ev'}}$ \lstinline|else ERROR|\\
    \lstinline|else| $\mkwrap{\substitute{\syntype_2}{\ev'}{\ev}}{\expr \ \mkwrap{\syntype_1}{\ev'}}$ 
  \end{enumerate}
\end{definition}

\subsubsection*{Wrapping for Variant and Intersection Types.}
Variant wrapping is straightforward; the subtle case is for intersections.  Recall that intersections appearing immediately in function domains are restricted to be intersections of functions taking distinct variant arguments.  This restriction in turn allows for a natural definition of wrapping: simply case on the variant passed in and constrain it to its underlying type.
\begin{definition}[Wrapper for Variant and Intersection Types]
  \label{def:wrap_variant_intersect}
  \ \par 
  \begin{enumerate}
    \setcounter{enumi}{5}
    \item \mkwrap{($\mkvariant{\syntype_1}{\syntype_n}$)}{\expr} = 
    \begin{lstlisting}
match e with | $V_1(v_1)$ -> $V_1(\mkwrap{\syntype_1}{\eval_1})$...| $V_n(v_n)$ -> $V_n(\mkwrap{\syntype_n}{\eval_n})$
    \end{lstlisting} 
    \item \mkwrap{$\mkintersect{(\mkfun{(V_1\ \texttt{of}\ \syntype_{1})}{\syntype_{1}'})}{(\mkfun{(V_n\ \texttt{of}\ \syntype_{n})}{\syntype_{n}'})}$}{\expr} = 
\begin{lstlisting}
fun $\ev$ -> match $\ev$ with
  | $V_1(\eval_1)$ -> if pick_b then 
      if $\mkche{\syntype_1}{\eval_1}$ then $\mkwrap{\syntype_1'}{\expr \ V_1(\mkwrap{\syntype_1}{\eval_1})}$ else ERROR
    else $\mkwrap{\syntype_1'}{\expr \ V_1(\mkwrap{\syntype_1}{\eval_1})}$
  | ...
\end{lstlisting}
  \end{enumerate}
\end{definition}

\subsubsection*{Wrapping Types as Expressions.}
To extend wrapping to types-as-expressions, we simply add a third field \lstinline!wrap = ...! to the encoding of types as expressions.
\begin{definition}[Types as Expressions With Wrappers]
  Modify Definitions \ref{def:sem_of_types} and \ref{def:recursive_type_embedding} by adding the following \lstinline|wrap| clauses to those gen/check records:
    \label{def:sem_of_types_wrapped}
    \ \par
    \begin{enumerate}
      \item $\llbracket \tint \rrbracket$ = \lstinline`{ gen =$\ldots$; check = $\ldots$; wrap = fun e -> e }`
      \item $\llbracket \tbool \rrbracket$ = \lstinline`{ gen =$\ldots$; check = $\ldots$; wrap = fun e -> e }`
      \item $\llbracket \mkfun{\syntype_1}{\syntype_2} \rrbracket$ = 
      \begin{minipage}{.75\linewidth}\begin{lstlisting}[gobble=4]
      { gen = $\ldots$; check = $\ldots$;
        wrap = fun e -> fun x -> if pick_b then 
          if ($\llbracket\syntype_1\rrbracket$.check x) then $\llbracket\syntype_2\rrbracket$.wrap (e ($\llbracket\syntype_1\rrbracket$.wrap x)) else ERROR
        else $\llbracket\syntype_2\rrbracket$.wrap (e ($\llbracket\syntype_1\rrbracket$.wrap x)) }
      \end{lstlisting}\end{minipage}
      \item $\llbracket\ttob\hastype{l_1}{\syntype_1}; \ldots; \hastype{l_n}{\syntype_n}\ttcb\rrbracket$ = 
      \begin{minipage}{.75\linewidth}\begin{lstlisting}[gobble=4]
      { gen = $\ldots$; check = $\ldots$; 
        wrap = fun e -> { $l_1$ = $\llbracket \syntype_1 \rrbracket$.wrap (e.$l_1$); ...; $l_n$ = $\llbracket \syntype_n \rrbracket$.wrap e.$l_n$ } }
      \end{lstlisting}\end{minipage}
      \item $\llbracket \mkmiu{\beta}{\syntype} \rrbracket$ = 
      \begin{minipage}{.75\linewidth}\begin{lstlisting}[gobble=4]
      Y (fun self -> fun _ -> 
        { gen = $\ldots$; check = $\ldots$; wrap = fun e -> (fun $\beta$ -> $\llbracket\syntype\rrbracket$.wrap e) (self 0) })
      \end{lstlisting}\end{minipage}
    \end{enumerate}
  \end{definition}

  \subsubsection*{Wrapping Parametric Polymorphism.}
  
  
Support of wrapping for polymorphic functions also requires an extension of the syntax beyond what we used in the theory: uses of polymorphic functions must instantiate the type.  Since we have no type inference, our implementation syntax requires users to supply the type instantiations, and additionally all polymorphic type variables must be declared up-front. Polymorphic functions have the type \lstinline|forall 'a ... 'b. $\syntype_1$ -> $\syntype_2$|, and we also support the sugar \lstinline|let f (type a b ...) $ $ x = $ $e1 in e2|.

Our polymorphism use checking is similar to how dynamic contracts extended with parametric polymorphism \cite{ParametricContracts} performs use checks.  We now define the types-as-expressions form of polymorphic types; we also include the \lstinline`gen`/\lstinline`check` clauses since we need to apply type parameters when invoking functions in the implementation syntax.

  \begin{definition}[Parametric Polymorphism with Wrappers]
    \label{def:wrap_poly}
    \ \par 
    \begin{enumerate}
      \setcounter{enumi}{5}
      \item $\llbracket$\lstinline|forall| $\alpha_1 \cdots\alpha_n$ . $\mkfun{\syntype_1}{\syntype_2} \rrbracket$ = 
      \begin{lstlisting}
  { gen = fun _ -> fun $\alpha_1 \ldots \alpha_n$ -> fun arg -> if $\pickb$ then 
        if ($\llbracket\syntype_1\rrbracket$.check arg) then ($\llbracket\syntype_2\rrbracket$.gen 0) else ERROR else ($\llbracket\syntype_2\rrbracket$.gen 0);
    check = fun e ->  
      if e ~ fun then let arg = ($\llbracket\syntype_1\rrbracket$.gen 0) in ($\llbracket\syntype_2\rrbracket$.check) (e V($\alpha_1$) ... V($\alpha_n$) arg) 
        else false;
    wrap = fun e ->
      fun $\alpha_1$ ... $\alpha_n$ -> fun x' -> if $\pickb$ then 
        if ($\llbracket\syntype_1\rrbracket$.check x') then let x = $\llbracket\syntype_1\rrbracket$.wrap x' in $\llbracket\syntype_2\rrbracket$.wrap (e V($\alpha_1$) ... V($\alpha_n$) x) 
        else ERROR 
      else 
        let x = $\llbracket\syntype_1\rrbracket$.wrap x' in $\llbracket\syntype_2\rrbracket$.wrap (e V($\alpha_1$) ... V($\alpha_n$) x) }
      \end{lstlisting}
    \item $\llbracket \alpha\rrbracket$ = \lstinline|{$$ gen = $$fun _ ->$$ V($\alpha$); check = $$fun e -> e $\simeq \alpha$; wrap = fun e -> e }|
    \end{enumerate}
  \end{definition}
  
\subsubsection*{Wrapping Record Subtyping.}

We need to update the wrapper definition in Definition~\ref{def:sem_of_types} to add a re-tagging which maintains the invariant that after a type is put on a record value, all subsequent access on this record will only be valid for labels declared in the type. 

\begin{definition}
  We modify Definition~\ref{def:sem_of_types_wrapped} by replacing \texttt{wrap} for record types as follows:
  \ \par
  \begin{enumerate}
  \setcounter{enumi}{7}
    \item $\llbracket\ttob\hastype{l_1}{\syntype_1}; \ldots; \hastype{l_n}{\syntype_n}\ttcb\rrbracket$ = 
    \begin{minipage}{.75\linewidth}
    \begin{lstlisting}[gobble=4]
    { ...;$\ $wrap =$\ $fun e -> 
      let r' = { $l_1$ = $\llbracket \syntype_1 \rrbracket$.wrap (e.$l_1$); ...; $l_n$ = $\llbracket \syntype_n \rrbracket$.wrap e.$l_n$ } in 
      retag$\ttop$r', $\{l_1; \cdots; l_n\}\ttcp$ }
    \end{lstlisting}\end{minipage}
   \end{enumerate}
\end{definition}

Consider the following function, which demonstrates how the new wrapper addresses subtyping:

\begin{lstlisting}
  let f (r : {a : int}) : int = match r with
  | {a; b} -> r.b
  | {a} -> r.a
  in f {a = 5; b = true}
\end{lstlisting}

The definition of \lstinline|f| will not be refuted by our checker because all inputs created by the generator of type \lstinline|{a : int}| will go under the second match case. The use \lstinline|f {a = $ $ 5; b =$ $ true}| will also typecheck because the re-tagging in the call to the \fnstyle{Wrapped} version of \lstinline!f! will hide the \lstinline|b| field, and the second match branch will be taken.  This case is a bit subtle for programmers used to dynamically casing on record/object structure: they need to be aware that a record type restriction is also a runtime restriction to avoid seeing behind the type interface.

\subsection{A Selection of Examples}
In this section, we will showcase some examples that demonstrate the implementation's capabilities.

\subsubsection*{Dependent + Refinement + Polymorphic Types: {\tt prepend}.} 

Consider the implementation of a list prepend function, given in Figure \ref{fig:list_prepend}. This implementation is buggy, because it adds the same element into the second list twice, thus violating the predicate on the return value (i.e. the resulting list's length must be the sum of the lengths of the two original lists). Our type checker provides the following error message:


\bgroup\scriptsize
\begin{verbatim}
** Bluejay Type Errors **
- Found at clause : let rec prepend ... in prepend
--------------------
* Value    : prepend
* Expected : ((x : [a]) -> ((y : [a]) -> 
          {[a] | fun r -> length r 0 == length x 0 + length y 0}))
* Actual   : ((x : [a]) -> ((y : [a]) -> hd :: (hd :: acc)))
\end{verbatim}
\egroup

\subsubsection*{Type Function: mk\_student.} Consider the following program, whose return types are two different record types depending on the integer value of the argument, \texttt{age}. 

\begin{lstlisting}
let mk_rec age = 
    if age >= 18 then { age : int; employed : bool } else { age : int } 
in
let mk_student (n : int) : bool -> (mk_rec n) = fun employed ->
    if n > 18 then {age = n; employed = employed} else {age = n}
in mk_student
\end{lstlisting}

According to the \texttt{mk\_rec} type function, records should only contain the extra field \texttt{employed} if the input \texttt{age} is greater than or equal to 18. However, in \texttt{mk\_student}, the function mistakenly sets the threshold to strictly greater than 18, meaning that the return value will be missing the \texttt{employed} field if \texttt{age}'s value is exactly 18. Our type checker is able to spot this error as well, but the error message itself is not as understandable as the previous example.

\bgroup\scriptsize
\begin{verbatim}
** Bluejay Type Errors **
- Found at clause : 
  let mk_student ... in mk_student
--------------------
* Value    : mk_student
* Expected : ((n : int) -> (bool -> mk_rec n))
* Actual   : TypeError: Type unknown
\end{verbatim}
\egroup

\subsubsection*{Recursive + Record Types.} Consider the following program which combines recursive and record types to encode a binary tree. (\lstinline`is_bst` is a user-defined function checking whether a given tree meets the binary search tree criteria; its code is omitted for brevity.)

\begin{lstlisting}
let is_bst = ... in
let tree_type =  Mu tt. (Node { left : tt; right : tt; item : int } || Leaf { leaf : bool }) 
in let (bad_tree : { tree_type | is_bst }) = 
  Node  { left = 
          Node { left = Leaf { leaf = true };
                right = Leaf { leaf = true };
                item = 6 };
            right = Leaf { leaf = true }; item = 2 }
in bad_tree
\end{lstlisting}

According to the type declaration, this tree should be a binary search tree. However, although this value conforms to the form of a binary tree, it violates the invariance of a binary search tree: its root node's left child is bigger than the root's value. It thus fails to type check.  Our type checker reports the following error:

\bgroup\scriptsize
\begin{verbatim}
** Bluejay Type Errors **
- Found at clause : let (bad_tree : { tree_type | is_bst }) = ... in bad_tree
--------------------
* Value    : bad_tree
* Expected : {tree_type | is_bst}
* Actual   : Node {item = 2;
                  left = Node {item = 6; left = Leaf {leaf = true}; right = Leaf {leaf = true}};
                  right = Leaf {leaf = true}}
\end{verbatim}
\egroup

\subsubsection*{Record + Function Subtyping.} Consider the following program which requires record and function subtyping to successfully type check:

\begin{lstlisting}
let (r : { a : int; b : int }) = {a = 1; b = 2} in
let transform_record (i : { a : int }) : { a : int; c : bool } = 
    {a = r.a; c = r.a > 0}
in
let (new_record : { c : bool }) = transform_record r in new_record
\end{lstlisting}

There are two points here that require subtyping: (1) the application, \lstinline`transform_record r`, where the argument's type is a subtype of the input type, and (2) type checking the application result, \lstinline`new_record`, where the value itself is a subtype of the declared type. Our checker correctly outputs ``\lstinline!No errors found!'' for this program.

\subsection{Back End Implementations}
\label{sec_backend}

We experimented with two different back-ends to search for \lstinline|ERROR|s in instrumented programs.  

First we built a demand-driven symbolic evaluator based on \cite{DDSE} by forking that code base.  Demand-driven symbolic evaluation evaluates programs in reverse, and we hypothesized that working back from known \lstinline|ERROR|s in the translated programs would give a more goal-directed, efficient search.  In practice, this symbolic evaluator ran far too slowly; there is considerable run-time overhead, which could not be easily optimized away.

We then implemented a concolic symbolic evaluator back end to get faster results.  This implementation was on average several orders of magnitude faster than the demand symbolic evaluator and was also strictly faster across-the-board, so here we report runtime results on the concolic evaluator only. The translation system has yet to be optimized for speed and is currently inefficient, so we separately report translation time and \texttt{ERROR} search time below. Concolic symbolic evaluators \cite{DART} initially proceed like fuzzers and property-based testers, running programs on random inputs.  These runs are instrumented to record which conditional branches are taken, and an SMT solver is then used to infer inputs which will exercise an as-yet-untaken branch.  We could find no pre-existing concolic evaluator for functional programs that we could use, so we implemented our own concolic evaluator from scratch.

\begin{wraptable}{r}{.40\textwidth}
  \caption{Bluejay features with the number of tests using the features (``uses'') and number of tests in which the feature is a reason for the type error (``errors'').}
  \begin{center}\small
    \begin{tabular}{r|cc}
      Feature & uses & errors\\
      \hline
      Polymorphic types (P)       &  29  &  20 \\
      Variants (V)                &  19  &  6 \\
      Intersection types (I)      &  6   &  4 \\
      Recursive functions (R)     &  59  &  21 \\
      Mu types (M)                &  13  &  10 \\
      Higher order functions (H)  &  50  &  9 \\
      Subtyping (S)               &  6   &  1 \\
      Type casing (T)             &  4   &  2 \\
      OOP-style (O)               &  15  &  7 \\
      Refinement types (F)        &  44  &  34 \\
      Dependent types (D)         &  16  &  6 \\
      Parametric types (A)        &  11  &  4 \\
      Records (C)                 &  42  &  19 \\
      Wrap required (W)           &  11  &  11 \\
      Assertions (N)              &  9   &  8 \\
      Operator misuse (U)         &  10  &  10 \\
      Return type (Y)             &  40  &  34 \\
      Match (X)                   &  45  &  4 \\
    \end{tabular}
  \end{center}
  \label{table_bjy_features}
  \end{wraptable}

The concolic evaluator concretely interprets the program, substituting random values for all input clauses not yet seen in any previous interpretations. For each conditional branch taken, the negation of the branch is pushed to two queues: a depth-first search horizon and a breadth-first search horizon. At the conclusion of the interpretation, if an \lstinline|ERROR| was not found, a target branch is randomly popped from either the depth-first or breadth-first horizon, and the Z3 SMT solver checks satisfiability of the target branch. If the branch is unsatisfiable, another target is randomly popped. Once a satisfiable target is found, the program is interpreted again, where inputs are decided by the solver for known clauses or are again random for newly-seen clauses, and the process repeats.

  We use several heuristics to encourage an efficient search. Each interpretation is terminated after it reaches a fixed max number of steps. If this step count is reached, the next target branch is dequeued from the breadth-first horizon to guide the evaluator away from branches that might be more likely to hit the max step count again. The tree of executed and potential program paths is pruned at a max depth, and the evaluator quits after it exhausts all execution paths up to that depth or finds an \lstinline|ERROR|. The bookkeeping is expensive for a large tree, so we incrementally increase the depth of the tree if no \lstinline|ERROR|s are found at shallower depths. We draw from the horizons in this way without other heuristics because the program is instrumented with frequent \lstinline|ERROR| clauses, and most conditional branches have a short distance to an \lstinline|ERROR|; prioritizing branches based on their proximity to uncovered code is therefore not expected to be fruitful in these instrumented programs, and most code will remain unreachable in well-typed programs. A na\"ive implementation of such a heuristic yielded no obvious improvement, but we intend to revisit the idea in future work. The default parameters in the current implementation are as follows: program paths are explored up to sixty conditional branches, and the max depth of the path tree is incremented to sixty in six equal steps. Interpretations are cut off at 50\,000 steps, and the concolic evaluator quits after 90 seconds if no \texttt{ERROR}s are found.

  \begin{table}[h]
    \caption{The BlueJay benchmarks. Run, translation, and total times are in ms. LOC is lines of code.  Errors were succesfully found in the benchmarks above the lines; the search was unsuccessful below the lines.  ?? indicates timeout and $^*$ indicates the branching depth limit was exhausted. Letters are used for readability to indicate which features from Table \ref{table_bjy_features} are in the test. Black font indicates the feature is used, \red{red font} indicates the feature is key to the type error, and -- indicates the feature is not present in the test. }
    \begin{center}\footnotesize
      \begin{tabular}{@{}r|c@{\hspace{3pt}}c@{\hspace{3pt}}c|@{\hspace{3pt}}c@{\hspace{3pt}}|c@{\hspace{3pt}}c@{\hspace{3pt}}c@{\hspace{3pt}}c@{\hspace{3pt}}c@{\hspace{3pt}}c@{\hspace{3pt}}c@{\hspace{3pt}}c@{\hspace{3pt}}c@{\hspace{3pt}}c@{\hspace{3pt}}c@{\hspace{3pt}}c@{\hspace{3pt}}c@{\hspace{3pt}}c@{\hspace{3pt}}c@{\hspace{3pt}}c@{\hspace{3pt}}c@{\hspace{3pt}}c@{\hspace{3pt}}}
      Test Name & Run & Transl & Total & LOC & \rot{\underline{P}olymorphic types} & \rot{\underline{V}ariants} & \rot{\underline{I}ntersection types} & \rot{\underline{R}ecursive functions} & \rot{\underline{M}u types} & \rot{\underline{H}igher order functions} & \rot{\underline{S}ubtyping} & \rot{\underline{T}ype casing} & \rot{\underline{O}OP-style} & \rot{Re\underline{f}inement types} & \rot{\underline{D}ependent types} & \rot{P\underline{a}rametric types} & \rot{Re\underline{c}ords} & \rot{\underline{W}rap required} & \rot{Assertio\underline{n}s} & \rot{Operator mis\underline{u}se} & \rot{Return t\underline{y}pe} & \rot{Match (X)}\\
        \hline
        \texttt{ad\_hoc\_polymorphism}      &  254   &  13     &  267    &  34   &  \red{P}  &  --       &  --       &  R        &  --       &  --       &  --       &  \red{T}  &  --       &  F        &  --       &  A        &  --       &  --       &  --  &  --       &  --       &  X \\
        \texttt{adapter\_record}            &  4     &  23     &  27     &  30   &  --       &  --       &  --       &  --       &  --       &  H        &  --       &  --       &  O        &  --       &  --       &  --       &  \red{C}  &  --       &  --  &  --       &  Y        &  -- \\
        \texttt{avl\_tree\_instance}        &  2     &  106    &  107    &  176  &  --       &  V        &  --       &  \red{R}  &  \red{M}  &  --       &  --       &  --       &  --       &  \red{F}  &  --       &  --       &  C        &  --       &  --  &  --       &  --       &  X \\
        \texttt{continuation\_bind1}        &  1     &  63     &  64     &  12   &  \red{P}  &  --       &  --       &  --       &  --       &  H        &  --       &  --       &  --       &  --       &  --       &  --       &  C        &  --       &  --  &  --       &  --       &  X \\
        \texttt{continuation\_bind2}        &  2     &  52     &  53     &  32   &  \red{P}  &  --       &  --       &  --       &  --       &  H        &  --       &  --       &  --       &  --       &  --       &  \red{A}  &  C        &  --       &  --  &  --       &  \red{Y}  &  -- \\
        \texttt{continuation\_bind\_usage}  &  3     &  51     &  53     &  32   &  \red{P}  &  --       &  --       &  --       &  --       &  \red{H}  &  --       &  --       &  --       &  --       &  --       &  A        &  \red{C}  &  --       &  --  &  --       &  --       &  -- \\
        \texttt{decorator\_timer}           &  2     &  21     &  22     &  24   &  --       &  --       &  --       &  --       &  --       &  H        &  --       &  --       &  O        &  --       &  --       &  --       &  C        &  --       &  --  &  \red{U}  &  --       &  -- \\
        \texttt{duck\_typing\_colors1}      &  2     &  34     &  35     &  45   &  --       &  --       &  --       &  --       &  --       &  H        &  --       &  --       &  \red{O}  &  F        &  D        &  --       &  \red{C}  &  --       &  --  &  --       &  --       &  -- \\
        \texttt{duck\_typing\_colors2}      &  1     &  34     &  34     &  45   &  --       &  --       &  --       &  --       &  --       &  H        &  --       &  --       &  \red{O}  &  F        &  D        &  --       &  \red{C}  &  --       &  --  &  --       &  --       &  -- \\
        \texttt{flyweight\_color1}          &  15    &  28     &  43     &  58   &  --       &  --       &  --       &  --       &  --       &  H        &  --       &  --       &  \red{O}  &  F        &  D        &  --       &  \red{C}  &  --       &  --  &  --       &  --       &  -- \\
        \texttt{flyweight\_color2}          &  394   &  31     &  424    &  57   &  --       &  --       &  --       &  --       &  --       &  H        &  --       &  --       &  O        &  F        &  D        &  --       &  C        &  --       &  --  &  \red{U}  &  --       &  -- \\
        \texttt{flyweight\_color3}          &  16    &  33     &  49     &  61   &  --       &  --       &  --       &  --       &  --       &  \red{H}  &  --       &  --       &  O        &  F        &  D        &  --       &  C        &  --       &  --  &  --       &  --       &  \red{X} \\
        \texttt{insertion\_sort1}           &  514   &  34     &  547    &  44   &  --       &  --       &  --       &  \red{R}  &  --       &  --       &  --       &  --       &  --       &  \red{F}  &  --       &  --       &  --       &  --       &  --  &  --       &  Y        &  X \\
        \texttt{insertion\_sort2}           &  295   &  35     &  329    &  44   &  --       &  --       &  --       &  \red{R}  &  --       &  --       &  --       &  --       &  --       &  \red{F}  &  --       &  --       &  --       &  --       &  --  &  --       &  Y        &  X \\
        \texttt{list\_take\_n}              &  226   &  262    &  487    &  19   &  P        &  --       &  --       &  R        &  --       &  --       &  --       &  --       &  --       &  F        &  D        &  --       &  --       &  --       &  --  &  --       &  \red{Y}  &  X \\
        \texttt{poly\_casting\_applied}     &  204   &  4      &  207    &  12   &  \red{P}  &  --       &  --       &  --       &  --       &  --       &  --       &  --       &  --       &  --       &  --       &  --       &  --       &  --       &  --  &  \red{U}  &  --       &  -- \\
        \texttt{poly\_compose}              &  3     &  14     &  17     &  16   &  \red{P}  &  --       &  --       &  --       &  --       &  \red{H}  &  --       &  --       &  --       &  --       &  --       &  --       &  --       &  \red{W}  &  --  &  --       &  --       &  -- \\
        \texttt{poly\_fold2\_map2}          &  15    &  545    &  559    &  23   &  \red{P}  &  --       &  --       &  \red{R}  &  --       &  H        &  --       &  --       &  --       &  --       &  --       &  --       &  --       &  --       &  --  &  --       &  --       &  X \\
        \texttt{poly\_sgn}                  &  14    &  174    &  188    &  22   &  \red{P}  &  V        &  --       &  \red{R}  &  --       &  H        &  --       &  --       &  --       &  --       &  --       &  --       &  --       &  --       &  --  &  \red{U}  &  --       &  X \\
        \texttt{poly\_type\_casing}         &  1     &  24     &  24     &  13   &  \red{P}  &  --       &  I        &  --       &  --       &  H        &  S        &  \red{T}  &  --       &  --       &  --       &  --       &  --       &  --       &  --  &  --       &  --       &  X \\
        \texttt{rec\_dep\_polymorphism}     &  87    &  25     &  112    &  12   &  \red{P}  &  --       &  --       &  \red{R}  &  --       &  --       &  --       &  --       &  --       &  --       &  \red{D}  &  \red{A}  &  --       &  --       &  --  &  --       &  --       &  -- \\
        \texttt{rec\_polymorphism}          &  58    &  4      &  62     &  6    &  \red{P}  &  --       &  --       &  R        &  --       &  --       &  --       &  --       &  --       &  --       &  --       &  --       &  --       &  --       &  --  &  --       &  \red{Y}  &  -- \\
        \texttt{record\_intersection}       &  2     &  436    &  438    &  22   &  --       &  --       &  \red{I}  &  --       &  --       &  H        &  --       &  T        &  --       &  --       &  --       &  --       &  C        &  --       &  --  &  \red{U}  &  --       &  X \\
        \texttt{set\_module}                &  295   &  27618  &  27912  &  37   &  --       &  --       &  --       &  R        &  --       &  H        &  --       &  --       &  --       &  --       &  --       &  A        &  \red{C}  &  --       &  --  &  --       &  \red{Y}  &  -- \\
        \texttt{sub\_fun\_intersection}     &  15    &  95     &  109    &  19   &  --       &  --       &  \red{I}  &  --       &  --       &  H        &  S        &  T        &  --       &  F        &  --       &  --       &  --       &  --       &  --  &  --       &  \red{Y}  &  X \\
        \texttt{sub\_higher\_order\_fun}    &  16    &  1891   &  1906   &  11   &  --       &  --       &  --       &  --       &  --       &  \red{H}  &  S        &  --       &  --       &  --       &  --       &  --       &  \red{C}  &  --       &  --  &  --       &  --       &  X \\
        \texttt{sub\_implied\_type}         &  313   &  16     &  329    &  12   &  --       &  --       &  --       &  R        &  --       &  \red{H}  &  \red{S}  &  --       &  --       &  \red{F}  &  --       &  --       &  --       &  \red{W}  &  --  &  --       &  --       &  -- \\
        \texttt{variant\_obj}               &  5     &  147    &  151    &  17   &  --       &  V        &  \red{I}  &  --       &  --       &  \red{H}  &  --       &  --       &  \red{O}  &  --       &  --       &  --       &  C        &  --       &  --  &  --       &  \red{Y}  &  X \\
        \texttt{visitor\_accept}            &  2     &  33     &  35     &  29   &  --       &  --       &  --       &  --       &  M        &  H        &  --       &  --       &  \red{O}  &  --       &  --       &  --       &  \red{C}  &  --       &  --  &  --       &  --       &  -- \\
        \texttt{visitor\_tree1}             &  1207  &  65     &  1271   &  63   &  --       &  V        &  --       &  R        &  \red{M}  &  H        &  --       &  --       &  \red{O}  &  --       &  --       &  A        &  C        &  \red{W}  &  --  &  --       &  --       &  X \\
        \texttt{visitor\_tree2}             &  7007  &  54     &  7061   &  63   &  --       &  V        &  --       &  R        &  \red{M}  &  H        &  --       &  --       &  O        &  --       &  --       &  A        &  C        &  --       &  --  &  \red{U}  &  --       &  X \\
        \texttt{visitor\_tree3}             &  192   &  62     &  254    &  63   &  --       &  V        &  --       &  R        &  \red{M}  &  H        &  --       &  --       &  O        &  --       &  --       &  A        &  \red{C}  &  --       &  --  &  \red{U}  &  --       &  X \\
        \texttt{visitor\_variant}           &  262   &  56     &  318    &  25   &  --       &  \red{V}  &  --       &  --       &  M        &  H        &  --       &  --       &  O        &  --       &  --       &  --       &  C        &  --       &  --  &  --       &  --       &  \red{X} \\
  
        \hline\hline
        \texttt{avl\_tree}                 &  ??     &  138   &  ??     &  164  &  --       &  V   &  --  &  \red{R}  &  \red{M}  &  --  &  --  &  --  &  --  &  \red{F}  &  --  &  --  &  C        &  --  &  --  &  --  &  --       &  X \\
        \texttt{long\_recursion}$^*$       &  5905   &  10    &  5914   &  11   &  --       &  --  &  --  &  R        &  --       &  --  &  --  &  --  &  --  &  F        &  --  &  --  &  --       &  --  &  --  &  --  &  \red{Y}  &  -- \\
        \texttt{ngrams}$^*$                &  56505  &  2612  &  59116  &  96   &  \red{P}  &  --  &  --  &  R        &  --       &  H   &  --  &  --  &  --  &  F        &  D   &  A   &  \red{C}  &  --  &  N   &  --  &  --       &  X \\
        \texttt{recursively\_refined}$^*$  &  16343  &  58    &  16401  &  40   &  --       &  V   &  --  &  \red{R}  &  --       &  --  &  --  &  --  &  --  &  \red{F}  &  D   &  --  &  C        &  --  &  --  &  --  &  --       &  X \\
        \texttt{self\_returning}$^*$       &  1008   &  17    &  1025   &  7    &  --       &  --  &  --  &  --       &  M        &  H   &  --  &  --  &  O   &  --       &  --  &  --  &  C        &  --  &  --  &  --  &  \red{Y}  &  -- \\
        \texttt{sequential\_defs}$^*$      &  754    &  4     &  758    &  7    &  --       &  --  &  --  &  R        &  --       &  --  &  --  &  --  &  --  &  --       &  --  &  --  &  --       &  --  &  --  &  --  &  \red{Y}  &  -- \\
  
      \end{tabular}
    \end{center}
    \label{table_bjy_benchmarks}
  \end{table}

  \subsection{Performance Evaluation}
  \label{sec_eval}
  
We benchmarked our implementation on a set of nontrivial ill-typed test programs that aim to cover the key features of Bluejay. To show the coverage of the tests, in Table \ref{table_bjy_features} we list key features along with the total number of ill-typed tests that use the feature some way.  Our tests include both unit tests for particular features as well as benchmarks which solve particular programming problems.  
The benchmarks and their performance are listed in Table \ref{table_bjy_benchmarks}.  For each benchmark we include the time to translate, the time for the concolic evaluator to find an error, the features the benchmark uses (see Table \ref{table_bjy_features} for the mapping of capital letters in the table to particular features), and those features which are critically involved in the particular type error.  Appendix \infull{\ref{app_tests}}\inshort{D.6 of the supplement} extends Table \ref{table_bjy_benchmarks} to cover all the unit tests we used.
%
For the benchmarks below the lines in Table \ref{table_bjy_benchmarks}, the system exhausted the maximum tree depth of 60 in the indicated time, or timed out after 300 seconds.  We also ran a well-typed version of every benchmark to check correctness of the translation and the concolic evaluator, and in all cases, the evaluator found no error before timing out or exhausting the search. 

\begin{wraptable}{r}{.33\textwidth}
      \caption{Bluejay vs static contract running times (times are in ms).}
      \begin{center}\footnotesize
         \scalebox{.99}{
            \begin{tabular}{r|c|c}
          Benchmark  &  Bluejay & Contracts\\
          \hline
            \texttt{all}            & 85     & 22       \\  
            \texttt{append}         & 372    & 7        \\  
            \texttt{boolflip-e}     & 3      & 22       \\  
            \texttt{braun-tree}     & 219    & 6427     \\  
            \texttt{flatten}        & 465    & 22       \\  
            \texttt{fold-fun-list}  & 143    & 395      \\ 
            \texttt{foldl}          & 93     & 20       \\  
            \texttt{foldl1}         & 13     & 20       \\  
            \texttt{foldr}          & 99     & 21       \\  
            \texttt{foldr1}         & 13     & 19       \\  
            \texttt{hors}           & 18     & 50       \\  
            \texttt{hrec}           & 1      & 144      \\  
            \texttt{intro1}         & 72     & 145      \\  
            \texttt{intro3}         & 13     & 24       \\  
            \texttt{last}           & 13     & 19       \\  
            \texttt{lastpair}       & 83     & 16       \\  
            \texttt{max}            & 28     & 148      \\  
            \texttt{mem}            & 100    & 328      \\  
            \texttt{member}         & 13     & 18       \\  
            \texttt{mult-all-e}     & 71     & 249      \\  
            \texttt{mult-cps-e}     & 1      & 35       \\  
            \texttt{mult-e}         & 1      & 21       \\  
            \texttt{mult}           & 35     & 160      \\  
            \texttt{nth0}           & 56     & 316      \\  
            \texttt{r-lock}         & 745    & 57       \\  
            \texttt{reverse}        & 1      & 197      \\  
            \texttt{sum-acm-e}      & 1      & 842      \\  
            \texttt{sum-all-e}      & 66     & 203      \\  
            \texttt{sum-e}          & 1      & 20       \\  
            \texttt{tree-depth}     & 1      & 12       \\  
        \end{tabular}
       }
      \end{center} 
      
      \label{table_scheme_benchmarks}
  \end{wraptable}
  We now highlight some of the patterns in these benchmarks.
  \paragraph{OOP-Style Programming}
  We designed tests that implement common structural and behaviorial design patterns of objects.  Objects were implemented as records of functions or as function intersections with self-passing.
  %
  \paragraph{Polymorphism}
  Common functional programming patterns such as monads, polymorphism, recursion, and higher order functions were tested extensively. Our system efficiently finds type errors in both the creation (e.g. \texttt{continuation\_bind1}) and usage (e.g. \texttt{continuation\_bind\_usage}) of such programming patterns.
  %
  %
  %
\paragraph{OCaml-Style Modules}
Bluejay supports OCaml-style modules without hidden types by using dependently-typed records containing types-as-values. An example module is the \texttt{set\_module} program, where we perform type checking on a simple set module implementation with an incorrect \lstinline|remove| function. Our analysis was able to catch the error quickly, but inefficiencies in the translation increase the total running time significantly.
\paragraph{Deep Errors}
The failing cases such as \texttt{avl\_tree} and \texttt{long\_recursion} contain errors that are very deep in the program path tree. Since the instrumented program is currently just run top-to-bottom, these deep errors might not be reached due to time-out on earlier declaration checks. One advantage of the refutation approach is there is no need to start at the top in checking types; earlier types could be ignored in a run, and only later functions are checked. So in the event of a timeout on an earlier declaration, the checker could re-run starting right below that definition; we plan to implement this improvement, and we expect that most or all of these errors should then be found.
  
  
\subsection{Performance Comparison}
\label{sec_eval_compare}
  
Along with our own benchmarks, we also ran benchmarks developed for related systems to verify our performance is reasonable. There were no ideal existing benchmarks as our system contains a unique mix of types, and we also need to find benchmarks of programs \emph{failing} to type check as successful programs will time-out in our tool. The most relevant benchmarks we found were static contract verification benchmarks \cite{RelCompleteCounterexamples}.  These benchmarks include both type correct and type erroring versions of numerous examples involving refinement types. We hand-ported these examples to our Bluejay syntax. The static contract semantics differs from Bluejay in several respects, perhaps the largest being that contracts rely heavily on the \texttt{any/c} contract, which we do not support.  In porting the examples to Bluejay, we changed \texttt{any/c} to a parametrically polymorphic type when possible and to some other type when not.  These contract examples also included no parametric polymorphism, records, variants, intersections, or subtyping because their language does not support those features, and that is one reason why we developed our own set of benchmarks in the previous section. On the other hand, our language currently lacks string and floating-point types, and benchmarks using those types were not ported.
    
The directory \texttt{benchmark-verification/fail-ce/} of the static contract paper artifact \cite{RelCompleteCounterexamplesRepo} contains the benchmarks which had type errors.  We successfully ported many of those benchmarks to Bluejay and ran them, and we also compared those running times with the paper artifact implementation which we downloaded and built; the results appear in Table \ref{table_scheme_benchmarks}.  The benchmarks listed in the table are all such benchmarks in this directory, which we concluded were reasonable to port.\footnote{For completeness, in this footnote we enumerate all the benchmarks in \texttt{benchmark-verification/fail-ce/} that were skipped and why. Several were not portable because they included string or floating-point types; these included \texttt{ack}, \texttt{argmin}, \texttt{fold-div}, \texttt{get-path}, \texttt{id-dependent}, \texttt{inc-or-greet}, \texttt{r-file}, \texttt{recip-contract}, and all programs in \texttt{octy/}.  \texttt{fhnhn} appeared specific to their contract syntax.  There were a few larger examples that contained no interesting type refinements, and we also skipped those: \texttt{snake}, \texttt{tetris}, and \texttt{zombie}.  There were also many numbered examples \texttt{142...} harvested from their web interface that we did not attempt to port.} Bluejay runtimes do not include the time to translate and instrument the program in order to benchmark only the search for the \texttt{ERROR}. Full translation times and features used in these programs are found in Appendix \infull{\ref{app_tests}}\inshort{D.6 of the supplementary material}.
In each case, our type checker terminated and found an error in the time indicated. The static contract artifact also contained successful examples which contained no type errors in the \texttt{../safe/} directory; all of these type-safe versions were also ported to Bluejay, and none triggered any type errors in our system: the system either reported that all paths were exhausted, reported that all paths were exhausted up to the allowed max depth, or timed out.

When we examine the running times, we see neither system was consistently faster. Bluejay was faster by mean (94ms to 332ms), and they were comparable by median, with the static contract system marginally faster (29.5ms to 31.5ms). 



\section{Related Work}

Our approach shares the same general philosophy as Incorrectness Logic \cite{10.1145/3371078,10.1145/3527325}, which emphasizes underapproximating possible program states to eliminate false positives, as opposed to the standard approach which overapproximates program states to enforce soundness.  Incorrectness Logic applies to first-order, imperative languages, but there has been a recent extension of this philosophy to functional language type systems for provable ill-typing \cite{10.1145/3632909}.  All of these works are primarily in the proof-theory space, whereas with semantic types we take a purely model-theoretic approach to incorrectness.

We share the ``falsification'' view of success typing \cite{10.1145/1140335.1140356}: i.e. we focus on finding type errors rather than proving their absence. However, instead of inferring types as with success typing, we only perform checks on explicit type declarations provided by users.

Modeling types via an inductive definition over type structure in terms of (untyped) program semantics originates with \cite{TaitRealizability}, where it was used to prove properties of a type system.  Use of unary logical relations as the sole semantic basis of program meaning originates with \cite{MILNER1978348}, and this approach is called a \emph{semantic typing} approach to type meaning \cite{DreyerBlog}.  Other more recent works that have a semantic typing basis include \cite{StepIndexedLogicalRelations,10.5555/1037736,StepIndexedKripke,10.1145/3676954}.

Current works on semantic subtyping, such as the CDuce Project \cite{SemanticSubtyping}, are built upon a largely syntactic foundation. In particular, the subtyping relations themselves are deduced semantically, but the types which the relations operate on are still syntactically determined. We are using a semantic notion for all typing in the system, rather than for subtyping only. Because we are modeling all typing and need to be semi-decidable, we face additional challenges in making the type generators for types in the presence of subtyping; CDuce is not modeling typing semantically and so does not run into this complexity. 

We embed higher-order type checking in the programs themselves, and this approach is in the spirit of runtime contract checking \cite{Findler-Contracts}. Unlike contracts we aim for static, modular verification and aim to be sound and complete with respect to a semantic types basis.

Closer to our approach are works which take a program with runtime contracts and either verify some of those contracts hold statically using a conservative abstract interpretation \cite{SoftContractVerif,Meunier:2006,nguyen_tobin-hochstadt_vanhorn_2017} or refute the contracts via symbolic model checking \cite{HaskellStaticContracts,RelCompleteCounterexamples}.  In particular, \cite{RelCompleteCounterexamples} also constructs counterexamples exhibiting contract failures via symbolic evaluation. A primary difference is that we build on known type foundations, showing our system is sound and complete with respect to a semantic typing basis.
Technically, we also take a different approach in how unknown functions of a given declared type are modeled.  Their approach is to view the function as a mapping that is initially empty and to incrementally extend it, in effect expanding a canonical syntactic definition of the function as symbolic execution proceeds. Our approach is ``semantic'' as opposed to syntactic, and we provide a fixed definition for the unknown function using built-in nondeterminism and let the SMT solver in effect infer the mapping.  
Their implementation includes numerous extensions, including a partially-successful extension to mutable state.  On the other hand, our implementation has a broader range of types including general parameterized types, parametric polymorphism, variants, intersections, and records with subtyping.

We also share some features with \cite{LazyCounterfactual}.  Their paper's stated purpose is to find errors in the \emph{specifications} themselves, but the symbolic evaluation and verification tools developed in the paper are of independent interest. They take a similar approach to how arbitrary function arguments are generated, using a symbolic value $s$ which is similar to our $\picki/\pickb$.  They restrict these values to be of first-order type only, however.

Our approach has overlapping goals with gradual typing \cite{gradual-typing-original}, but has a different methodology and trade-offs.  Our primary aim is purely static type checking but with a complete semantic typing basis.  The checking/generating code is not designed to be present at runtime of the actual program, unlike the case with gradual types, which may need dynamic runtime contracts at the interface.  Additionally, in gradual typing, it relies on a type system for typing programs and so will be incomplete; with this incompleteness there is more necessity to escape to purely dynamically-typed code.  In our approach, there is less of a need to combine with purely dynamic code as it should be more feasible to type check all the code (modulo time-out) due to the completeness of our type checker.  

Our approach also draws inspiration from property-based testing. Most notably, our construction of generators follows roughly the same philosophy outlined in the seminal QuickCheck paper by Claessen and Hughes \cite{10.1145/351240.351266}, where generators of bigger, more complex structures can be derived from smaller, simpler generators for basic types. However, our current approach to function value generation is fundamentally different from QuickCheck's method: the output of our function generator does not depend on the argument value, which means that we currently cannot generate pure function values.  

\section{Conclusions}

In this paper, we presented a theory which focuses on establishing type incorrectness for a functional language. Our approach shares the approach of Incorrectness Logic in that we only report type errors, and we prove that our type checker is sound and complete with respect to a semantic typing notion. 
We demonstrated the potential of this approach by showing how various complex types can be naturally incorporated, including refinement, dependent, polymorphic, and recursive types, as well as subtyping, and by providing a preliminary implementation which quickly finds many relatively complex type errors.
\subsubsection*{Future Work.}
The overall aim of basing a practical type checker on semantic typing still has several hurdles to overcome, in particular on the performance side: since the search space is exponentially expanding it may not practially find errors on some realistic programs.  The translator code is now overly verbose, but it should be possible to simplify away unnecessary branching constructs, which will in turn shrink the size of the control flow tree that must be searched and greatly improve performance.  Addition of an abstract interpretation back end will allow our type checker to show well-typed programs that lack complex refinements to be type-correct, avoiding the need for the symbolic evaluator to run to time-out.  The type checker is also capable in principle of type checking functions individually if all variables used have declared types; adapting the tool to do this will factor the workload and should lead to better performance.

Additional forms of typing need be incorporated, such as bounded polymorphism and existential types. Supporting languages with side effects is important but will be challenging; as was previously pointed out \cite{RelCompleteCounterexamples}, it is challenging to simulate an arbitrary function that can perform side effects.  It would also be desirable to support pure functional programming; currently, the user language includes the impure \lstinline!pick_i! operator as it is also needed in the generators and checkers to generate all possible functions.  On the foundations side, soundness and completeness need to be proven for the extensions in Section~\ref{sec_additional_extension}.

\begin{acks} 
    Thanks to the anonymous OOPSLA reviewers for extensive comments that have helped to significantly improve the paper.  Thanks to Zach Palmer for many discussions.  Thanks to Sean Murray for his initial implementation of the concolic evaluator.
\end{acks}
\section*{Data-Availability Statement}
The implementation of the type checker described in Section \ref{sec_impl} is archived and is publicly available on Zenodo \cite{Artifact24Zenodo} and Software Heritage \cite{Artifact24SoftwareHeritage}, and it is in an evolving open-source GitHub repository \cite{Artifact24GitHub}.

\bibliographystyle{ACM-Reference-Format}
\bibliography{main}
\infull{
\appendix
\newpage
\section{Proofs}
\label{app_proofs}

This Appendix contains the proofs of all the non-trivial lemmas in the paper.

\begin{definition}[Size of Types]The size of types defined in Definitions~\ref{def:lr_core}, \ref{def:lr_refinement_dependent}, and \ref{def:lr_poly} are calculated as follows:
  \ \par
  \begin{enumerate}
    \item $\textsc{size($\tint$)} = 1$.
    \item $\textsc{size($\tbool$)} = 1$.
    \item $\textsc{size($\tpoly$)} = 1$.
    \item $\textsc{size($\mkfun{\syntype_1}{\syntype_2}$)} = \textsc{size($\syntype_1$)} + \textsc{size($\syntype_2$)}$ + 1.
    \item $\textsc{size($\mktset{\syntype}{\expr_p}$)} = \textsc{size($\syntype$)} + 1$.
    \item $\textsc{size($\mkfun{\syntype_1}{\syntype_2}$)} = \textsc{size($\syntype_1$)} + \textsc{size($\syntype_2$)}$ + 1.
  \end{enumerate}
\end{definition}

\subsection{Proofs for Core Language}

This subsection contains proofs for the core Language introduced in Figure~\ref{fig:core_grammar}.

\completenessextra* 

\begin{proof}

  We prove this lemma by induction on the size of $\syntype$.

  \textbf{Base case}: $\syntype = \tint$

  Consider \textbf{clause (1)}. We need show that $\models \hastype{\mkgen{\tint}}{\tint}$.
  
  Since \mkgen{\tint} = \lstinline|pick_i|, by definition of \lstinline|$\picki$|, we know that $\forall \eval.$ if \lstinline|$\picki$| $\smallsteps[] \eval$, then $\eval \in \mathbb{Z}$. Thus, we have shown that $\models \hastype{\mkgen{\tint}}{\tint}$.

  Consider \textbf{clause (2)}. We need to prove that for an arbitrary $\expr$, if \mkche{\tint}{\expr} $\smallsteps[]$ \lstinline|ERROR|, or if \mkche{\syntype}{\expr}$\smallsteps[] \gtfalse$, then $\not\models \hastype{\expr}{\tint}$. By definition, \mkche{\tint}{\expr} = \matches{\expr}{\tint}. Proceed by case analysis on the evaluation result of the pattern match.

  \begin{enumerate}
    \item $\expr \smallsteps[]$ \lstinline|ERROR|. In this case, $\not\models\hastype{\expr}{\tint}$ is trivially true. This is also the only way that \mkche{\tint}{\expr} can evaluate to \lstinline|ERROR|, since pattern match will not return an \lstinline|ERROR| value.
    \item $\expr \smallsteps[] \eval_0$. In this case, by the operational semantics, we know that $\matches{\eval_0}{\tint} \smallsteps[] \gtfalse$, and this implies that $\eval_0 \notin \mathbb{Z}$. Therefore, we have $\not\models\hastype{\expr}{\tint}$ by definition. 
  \end{enumerate}

  The proof for the case where $\syntype = \tbool$ is very similar, so we will omit it here for brevity.

  \textbf{Inductive step:} \syntype\ =\ $\mkfun{\syntype_1}{\syntype_2}$

  Consider \textbf{clause (1)}. We need to show that $\models \hastype{\mkgen{\mkfun{\syntype_1}{\syntype_2}}}{\mkfun{\syntype_1}{\syntype_2}}$. To prove this, we need $\forall\eval.$ if $\models \hastype{\eval}{\syntype_1}$, then $\models \hastype{(\mkgen{\mkfun{\syntype_1}{\syntype_2}} \ \eval)}{\syntype_2}$.
  
  By definition, $\mkgen{\mkfun{\syntype_1}{\syntype_2}}$ is nondeterministic. There are two cases to consider here:
  
  \begin{enumerate}
    \item \sloppy Argument is checked: In this case, we have $\mkgen{\mkfun{\syntype_1}{\syntype_2}} \ \eval \smallsteps[]$ \lstinline|if|$\ \mkche{\syntype_1}{\eval}$ \lstinline|then| $\mkgen{\syntype_2}$ \lstinline|else ERROR|. By induction hypothesis on \textbf{clause (2)}, $\forall\eval.$ if $\models \hastype{\eval}{\syntype_1}$, then $\mkche{\syntype_1}{\eval}\nsmallsteps[] \gtfalse$ and $\mkche{\syntype_1}{\eval}\nsmallsteps[] \eerror$. This implies that either $\mkche{\syntype_1}{\eval}$ diverges or $\mkche{\syntype_1}{\eval} \smallsteps[] \gttrue$. If $\mkche{\syntype_1}{\eval}$ diverges, $\mkgen{\mkfun{\syntype_1}{\syntype_2}} \ \eval$ will diverge, too, making the statement $\models \hastype{\mkgen{\mkfun{\syntype_1}{\syntype_2}} \ \eval}{\syntype_2}$ trivially true. If $\mkche{\syntype_1}{\eval} \smallsteps[] \gttrue$, we only have to consider $\mkgen{\syntype_2}$. By induction hypothesis on \textbf{clause (1)}, we know that $\models \hastype{\mkgen{\syntype_2}}{\syntype_2}$. 
    \item Argument is not checked: In this case, we have $\mkgen{\mkfun{\syntype_1}{\syntype_2}} \ \eval \smallsteps[] \mkgen{\syntype_2}$. By induction hypothesis, we know that $\models \hastype{\mkgen{\syntype_2}}{\syntype_2}$.
  \end{enumerate}
  Therefore, we have shown that $\forall\eval.$ if $\hastype{\eval}{\syntype_1}$, then $\models \hastype{\mkgen{\mkfun{\syntype_1}{\syntype_2}} \ \eval}{\syntype_2}$. We can now safely conclude that $\models \hastype{\mkgen{\mkfun{\syntype_1}{\syntype_2}}}{\mkfun{\syntype_1}{\syntype_2}}$.

  Consider \textbf{clause (2)}. we need to prove that for any $\expr$, if $\mkche{\mkfun{\syntype_1}{\syntype_2}}{\expr} \smallsteps[] \eerror$ or if $\mkche{\mkfun{\syntype_1}{\syntype_2}}{\expr}\smallsteps[] \gtfalse$, then $\not\models \hastype{\expr}{\mkfun{\syntype_1}{\syntype_2}}$. If $\expr \smallsteps[] \eerror$, we have $\not\models\hastype{\expr}{\mkfun{\syntype_1}{\syntype_2}}$ trivially. We will only consider the case where $\expr \not\smallsteps[] \eerror$ in the discussion below.
  
  Let us first consider the case, $\mkche{\mkfun{\syntype_1}{\syntype_2}}{\expr} \smallsteps[] \eerror$. Expanding the checker definition, we get $\mkche{\mkfun{\syntype_1}{\syntype_2}}{\expr} = $\lstinline|if $\expr\ \tttilde$ fun then let arg =|\ $\mkgen{\syntype_1}$\ \lstinline|in| $\mkche{\syntype_2}{\expr \ \mkgen{\syntype_1}}$\ \lstinline| else false|. 
  
  Since $\mkgen{\syntype_1}$ does not evaluate to \lstinline|ERROR|, we know it must come from either $\matches{\expr}{\gtfun}$ or $\mkche{\syntype_2}{\expr \ \mkgen{\syntype_1}}$. 
  
  If $\matches{\expr}{\gtfun} \smallsteps[] \eerror$, it implies that $\expr \smallsteps[] \texttt{V}\ttop\tpoly\ttcp$ for some polymorphic variable $\tpoly$, from which we can conclude that $\not\models\hastype{\expr}{\mkfun{\syntype_1}{\syntype_2}}$. 
  
  If $\mkche{\syntype_2}{\expr \ \mkgen{\syntype_1}} \smallsteps[] \eerror$, by induction hypothesis on \textbf{clause (2)}, we know that $\not\models \hastype{(\expr \ \mkgen{\syntype_1})}{\syntype_2}$, and by induction hypothesis, we have \textbf{clause (1)} $\models \hastype{\mkgen{\syntype_1}}{\syntype_1}$. Thus we have found a witness $\models \hastype{\mkgen{\syntype_1}}{\syntype_1}$ such that $\not\models \hastype{(\expr \ \mkgen{\syntype_1})}{\syntype_2}$, proving that $\not\models \hastype{\expr}{\mkfun{\syntype_1}{\syntype_2}}$.
  
  Now, consider the case, $\mkche{\mkfun{\syntype_1}{\syntype_2}}{\expr}\smallsteps[] \gtfalse$. By the \textsc{checker} definition, we have the following possible cases:
  \begin{enumerate}
    \item $\matches{\expr}{\gtfun} \smallsteps[] \gtfalse$. In this case, by the operational semantics, we know that $\expr \smallsteps[] \eval_0$ and $\eval_0$ is not a function. Therefore, we have $\not\models\hastype{\expr}{\mkfun{\syntype_1}{\syntype_2}}$ trivially.
    \item $\matches{\expr}{\gtfun} \smallsteps[] \gttrue$. In this case, we know that $\expr \smallsteps[] \eval_0$ and $\eval_0$ is a function. However, this suggests that $\mkche{\syntype_2}{\expr \ \mkgen{\syntype_1}} \smallsteps[] \gtfalse$. By induction hypothesis on \textbf{clause (2)}, we know that $\not\models \hastype{(\expr \ \mkgen{\syntype_1})}{\syntype_2}$, and by induction hypothesis on \textbf{clause (1)}, we have $\models \hastype{\mkgen{\syntype_1}}{\syntype_1}$. Thus we have found a witness $\models \hastype{\mkgen{\syntype_1}}{\syntype_1}$ such that $\not\models \hastype{(\expr \ \mkgen{\syntype_1})}{\syntype_2}$, proving that $\not\models \hastype{\expr}{\mkfun{\syntype_1}{\syntype_2}}$.
  \end{enumerate}

\end{proof}

\subsubsection*{Reduction Contexts}
We now define notation needed in the soundness proof for expressing uniform compuation over hole(s).

\begin{wrapfigure}{r}{.6\textwidth}
    \begin{grammar}
      \grule[value context]{\evalcon}{
        \circ
        \gor \mathbb{Z}
        \gor \mathbb {B}
        \gor
        \mkfunv{\ev}{\exprcon}
      }
      \grule[expression context]{\exprcon}{
          \expr
          \gor \evalcon
          \gor \exprcon\ \exprcon
          \gor \exprcon\ \binop\ \exprcon
          \gline 
          \gor \ife{\exprcon}{\exprcon}{\exprcon}
          \gor \matches{\exprcon}{p}
      }
      \grule[redex context]{\redexcon}{
        \evalcon \binop \evalcon
        \gor \evalcon \ \evalcon
        \gor \ife{\vtrue}{\exprcon}{\exprcon}
        \gline
        \gor \ife{\vfalse}{\exprcon}{\exprcon}
        \gline 
        \gor \matches{\evalcon}{p}
        \gor \picki \gor \pickb
      }
      \grule[RC context]{\rccon}{
        \bar{\bullet}
        \gor \reductioncon \ \exprcon
        \gor \exprcon \ \reductioncon
        \gor \ife{\reductioncon}{\exprcon}{\exprcon}
        \gline
        \gor \reductioncon \binop \exprcon
        \gor \evalcon \binop \reductioncon
        \gor \matches{\reductioncon}{p}
    }
    \end{grammar}
    \caption{Definitions for Contexts}
    \label{fig:contexts}
  \end{wrapfigure}
  
  The next definition introduces context holes (denoted as $\circ$, or "white holes") into values, expressions, redexes, and reduction contexts.
  
  We use standard notation for context substitution, $\evalcon[\eval]$ or $\exprcon[\eval]$ which replaces $\circ$ with $\eval$ in $\evalcon$ or $\exprcon$.
  Another notable construct is the black hole context, $\bar{\bullet}$. Intuitively, it is $\bullet$ which additionally may contain white holes in the substituted expression.
  
  Sometimes we need to talk about reduction contexts that also contain normal holes $\circ$, which we will denote as $\redcon{\circ}{\bar{\bullet}}$, where the first set of brackets will indicate what value will be substituted uniformly into all of the white holes, and the second set will indicate what expression will be substituted into $\bar{\bullet}$. If $\bar{\bullet}$ itself contains white holes, they will be filled by the white-hole-substituting value as well. For example, if we have $\rccon = \bar{\bullet} + \circ + 2$, then $\redcon{0}{\circ + 1} = (0 + 1) + 0 + 2$.
  
  We will also be using Lemma~3.4 from \cite{FromOperationalToDenotational}:
  
  \begin{restatable}[Unique context factorization]{lemma}{uniquefactor}
    \label{lemma:unique_factor}
    If $\exprcon[\expr] \smallstep[] \expr'$, there exists unique $\rccon$ such that $\exprcon[\circ] = \redcon{\circ}{\exprcon'[\circ]}$, where $\exprcon'[\circ] = \redexcon[\circ]$ for some redex context $\redexcon$, or $\exprcon'[\circ] = \circ$.
  \end{restatable}
  
  \begin{proof}
    This is Corollary 3.5 in the aforecited paper.
  \end{proof}
  
  Furthermore, we will define what it means for a redex context to be \textit{parametric} with respect to white holes, $\circ$. Intuitively, it refers to a redex context which, when filled, its contractum will be independent of what values we use to fill in the $\circ$. Its formal definition is as follows:
  
  \begin{definition}
    \label{def:parametric_redcon}
    Redex context $\redexcon$ is \emph{parametric} iff $\redexcon=$ (\lstinline|fun x -> |$\ \exprcon) \ \exprcon'$, $\redexcon=$ \lstinline|if true then $\exprcon$ else $\exprcon'$|, or $\redexcon=$ \lstinline|if false then $\exprcon$ else $\exprcon'$|.
  \end{definition}
  
  We establish the following lemma about parametric redex contexts.
  
  \begin{restatable}{lemma}{parametricredcon}
    \label{lemma:parametric_redcon}
    $\forall \exprcon.$ if $\exprcon[\circ] = \redcon{\circ}{\redexcon[\circ]}$ and $\redexcon$ is parametric, then $\forall \eval, \eval'.$ if $\eval$ is closed and $\exprcon[\eval] \smallstep[] \exprcon'[\eval]$, then $\exprcon[\eval'] \smallstep[] \exprcon'[\eval']$.
  \end{restatable}

\begin{proof}
  Proceed by case analysis on $\redexcon$.
  \begin{enumerate}
    \item \sloppy \lstinline|if true then $\exprcon_1$ else $\exprcon_2$|: In this case, $\redexcon[\eval] = \ife{\vtrue}{\exprcon_1[\eval]}{\exprcon_2[\eval]}$. By operational semantics, $\redexcon[\eval] \smallstep[] \exprcon_1[\eval]$. Similarly, we have $\redexcon[\eval'] \smallstep[] \exprcon_1[\eval']$. Since $\exprcon[\eval] = \redcon{\eval}{\redexcon[\eval]}$ and $\exprcon[\eval'] = \redcon{\eval'}{\redexcon[\eval']}$, we can conclude $\exprcon[\eval] \smallstep[] \redcon{\eval}{\exprcon_1[\eval]}$ and $\exprcon[\eval'] \smallstep[] \redcon{\eval'}{\exprcon_1[\eval']}$, where $\redcon{\circ}{\exprcon_1[\circ]} = \exprcon'$.
    \item \lstinline|if false then $\exprcon_1$ else $\exprcon_2$|: Similar to the \lstinline|true| case.
    \item (\lstinline|fun x -> |$\ \exprcon_1) \ \exprcon_2$: In this case, $\redexcon[\eval] = (\mkfunv{\ev}{\exprcon_1[\eval]}) \ \exprcon_2[\eval]$. By operational semantics, $\redexcon[\eval] \smallstep[] \substitute{(\exprcon_1[\eval])}{\exprcon_2[\eval]}{\ev}$. Similarly, we have $\redexcon[\eval'] \smallstep[] \substitute{(\exprcon_1[\eval'])}{\exprcon_2[\eval']}{\ev}$. Since $\eval$ is closed, we know that there won't be any substitution happening to $\eval$ itself in $\exprcon_1$, which means we can factor a new context, $\exprcon' = \substitute{(\exprcon_1[\circ])}{\exprcon_2[\circ]}{\ev}$. Since $\exprcon[\eval] = \redcon{\eval}{\redexcon[\eval]}$ and $\exprcon[\eval'] = \redcon{\eval'}{\redexcon[\eval']}$, we can conclude $\exprcon[\eval] \smallstep[] \redcon{\eval}{\exprcon'[\eval]}$ and $\exprcon[\eval'] \smallstep[] \redcon{\eval'}{\exprcon'[\eval']}$.
  \end{enumerate}
\end{proof}

\subsubsection*{Soundness}
With this new notation we now prove a generalized Lemma from which soundness will be an easy corollary.
  
\soundnessextra*

\begin{proof}
  We will prove this lemma by induction on the size of $\syntype$.

  \textbf{Base case: }$\syntype = \tint$

  Consider \textbf{clause (1)}. Because $\mkgen{\tint} = \picki$, by operational semantics, it follows that $\picki \smallstep[] \eval$. Therefore, we have $\exprcon[\mkgen{\tint}] \smallsteps[] \exprcon[\eval]$, thus we can conclude $\exprcon[\mkgen{\tint}] \smallsteps[] \eerror$.

  Consider \textbf{clause (2)}. Given $\not \models \hastype{\expr}{\tint}$ and $\expr \smallsteps[] \eval$, we can conclude that $\eval \not\in \mathbb{Z}$. By definition, $\mkche{\tint}{\expr}$ = $\matches{\expr}{\tint}$. Since $\expr \smallsteps[] \eval$ and $\eval \not\in \mathbb{Z}$, we can conclude $\matches{\expr}{\tint} \smallsteps[] \gtfalse$, establishing this case.

  The proof for the case where $\syntype = \tbool$ is very similar, so we omit it here for brevity.

  \textbf{Inductive step: }$\syntype = \mkfun{\syntype_1}{\syntype_2}$

    Consider \textbf{clause (1)}. 

    Given that $\models \hastype{\eval}{\syntype}$, and $\exprcon[\eval] \smallsteps[] \eerror$, we need to show that $\exprcon[\mkgen{\syntype}] \smallsteps[] \eerror$. We will prove this by induction on the length of the $\exprcon[\eval] \smallsteps[] \eerror$ computation.

    We omit the trivial base case where $\exprcon = \eerror$, in which $\exprcon[\eval]$ will step to \lstinline|ERROR| in zero steps. 

    \textbf{Base case: }$\exprcon[\eval] \smallstep[] \eerror$

    By Lemma~\ref{lemma:unique_factor}, we know that there exists unique $\rccon$ and $\exprcon'$ such that $\exprcon[\circ] = \redcon{\circ}{\exprcon'[\circ]}$, where $\exprcon'[\circ] = \redexcon[\circ]$ or $\exprcon'[\circ] = \circ$. However, here we do not have to consider the case where $\exprcon'[\circ] = \circ$; since we're filling white holes with a value, a value by itself cannot be a redex. 

    There are two cases here we need two consider: 
    \begin{enumerate}
      \item $\redexcon$ is parametric: $\exprcon[\eval]$ will have to take more than one step to evaluate to an \lstinline|ERROR| expression, thus it cannot be in the base case.
      \item $\redexcon$ is not parametric: We proceed by case analysis on the redex context, $\redexcon$.

      \begin{enumerate}
        \item $\redexcon = \evalcon_1 + \evalcon_2$: According to the operational semantics, addition will evaluate to \lstinline|ERROR| if and only if one of the operands is a non-integer. If $\evalcon_1 = \circ$ or $\evalcon_2 = \circ$, since $\eval$ and $\mkgen{\syntype}$ both are function values, we can safely conclude that $\redexcon[\eval]$ and $\redexcon[\mkgen{\syntype}]$ will both evaluate to \lstinline|ERROR|. The other scenario is where at least one of the operands takes the form of \lstinline|fun x -> $\ \exprcon_f$| or $b$, where $b$ is a boolean value; in both cases, the expression will evaluate to \lstinline|ERROR| regardless of which value gets substituted into $\circ$.
        \item $\redexcon = \evalcon_1 \ \evalcon_2$: According to the operational semantics, function application will evaluate to \lstinline|ERROR| if and only if $\evalcon_1$ is an integer or boolean value. This implies $\evalcon_1$ doesn't contain any $\circ$, since the relevant values being substituted here are functions. Therefore, this can be reduced to the parametric case above.
        \item $\redexcon = \ife{\evalcon}{\exprcon_2}{\exprcon_3}$: Since the expression steps to an \lstinline|ERROR|, we know that $\evalcon[\eval]$ is not a boolean value. This implies that $\evalcon = \circ$, or $\evalcon = n$ where $n \in \mathbb{Z}$, or $\evalcon =\ $\lstinline|fun x -> $\ \exprcon_f$|. In all three cases, when we substitute $\mkgen{\syntype}$ into $\circ$, $\evalcon[\mkgen{\syntype}]$ still will not be a boolean value. Therefore, we can conclude that $\exprcon[\mkgen{\syntype}] \smallstep[] \eerror$.
        \item $\redexcon = \matches{\evalcon}{p}$: This is an impossible case. Since $\matches{\eval}{p}$ always returns a value, we know that the overall expression, $\exprcon[\eval]$, will have to take more than one step to evaluate to an \lstinline|ERROR| expression in this case. 
      \end{enumerate}

    \end{enumerate}

    \textbf{Inductive step: }$\exprcon[\eval] \smallsteps[] \eerror$

    Let us examine the first step in the given computation, which is effectively $\exprcon[\eval] \smallstep[] \expr_1$ for some intermediate evaluation result, $\expr_1$.
    
    As with the base case, we know that there exists unique $\rccon$ and $\redexcon$ such that $\exprcon[\circ] = \redcon{\circ}{\redexcon[\circ]}$. There are two main scenarios to consider:

    \begin{enumerate}
      \item $\redexcon$ is parametric: By Lemma~\ref{lemma:parametric_redcon}, we know that $\exprcon[\eval] \smallstep[] \exprcon'[\eval]$ implies $\exprcon[\mkgen{\syntype}] \smallstep[] \exprcon'[\mkgen{\syntype}]$. The result then follows by the inner induction hypothesis. 
      
      \item $\redexcon$ is not parametric: We proceed by case analysis on the redex context, $\redexcon$.
      \begin{enumerate}

        \item $\redexcon = \evalcon_1 + \evalcon_2$: According to the operational semantics, there are two cases to consider here. If both operands are integers, it implies that $\evalcon_1$ and $\evalcon_2$ don't contain any holes, which means $\redexcon[\eval] = \redexcon[\mkgen{\syntype}]$, and the rest follows by the induction hypothesis. The case where either one of the operands is not an integer is handled by the base case, since this will result in an immediate \lstinline|ERROR| expression after a single step of evaluation. 
        
        \item $\redexcon = \evalcon_1 \ \evalcon_2$: We've already covered the case of applying to a non-function in the base case. Now, we only need to consider the cases where $\evalcon_1 = \circ$, since $\evalcon_1 =\ $\lstinline|fun x -> $\ \exprcon_f$| for some $\exprcon_f$ is one of the cases where the redex is parametric. 
        
        Let $\eval_g = \mkgen{\mkfun{\syntype_1}{\syntype_2}}$. Consider the application, $\redexcon[\eval_g] = \evalcon_1[\eval_g] \ \evalcon_2[\eval_g] = \eval_g \ \evalcon_2[\eval_g]$. 
        Now consider the type of $\evalcon_2[\eval_g]$. There are two cases to consider here:
      
        \begin{enumerate}
          \item Case $\not\models\hastype{\evalcon_2[\eval_g]}{\syntype_1}$, since we have
          \begin{lstlisting}
$\eval_g$ = fun x -> if $\pickb$ then 
  if $\mkche{\syntype_1}{\texttt{x}}$ then $\mkgen{\syntype_2}$ else $\eerror$ 
else $\mkgen{\syntype_2}$
          \end{lstlisting}
        
          We know that 
          \begin{lstlisting}
$\eval_g$ $\evalcon_2[\eval_g]$ = if $\pickb$ then 
if $\mkche{\syntype_1}{\evalcon_2[\eval_g]}$ then $\mkgen{\syntype_2}$ else $\eerror$ 
else $\mkgen{\syntype_2}$
          \end{lstlisting}
        
          Since $\not\models\hastype{\evalcon_2[\eval_g]}{\syntype_1}$, and $\syntype_1$ is a smaller type, by induction hypothesis \textbf{clause (2)}, we know that $\mkche{\syntype_1}{\evalcon_2[\eval_g]} \smallsteps[] \gtfalse$ or $\mkche{\syntype_1}{\evalcon_2[\eval_g]} \smallsteps[] \eerror$. Either way, it will result in $\eval_g \ \evalcon_2[\eval_g] \smallsteps[] \eerror$ where $\pickb \smallstep[] \gttrue$, thus $\exprcon[\eval_g] \smallsteps[] \eerror$. 
              
        \item Case $\models\hastype{\evalcon_2[\eval_g]}{\syntype_1}$: Since $\eval$ is a function, we know that $\eval = \mkfunv{\ev}{\exprcon_v}$ for some context, $\exprcon_v$. Consider w.l.o.g. evaluations where $\pickb \smallstep[] \gtfalse$, i.e. the cases where the argument is not checked, we have $\redexcon[\eval] = \eval \ \evalcon_2[\eval] \smallsteps[] \substitute{\exprcon_v}{\evalcon_2[\eval]}{\ev}$. Putting it back into the overall context, we have $\redcon{\eval}{\substitute{\exprcon_v}{\evalcon_2[\eval]}{\ev}} \smallsteps[] \eerror$. Since this is smaller than the original computation, $\redcon{\eval}{\eval \ \evalcon_2[\eval]} \smallsteps[] \eerror$, by inner induction hypothesis, we have $\redcon{\eval_g}{\substitute{\exprcon_v}{\evalcon_2[\eval_g]}{\ev}} \smallsteps[] \eerror$. Now, consider the application, $\eval \ \evalcon_2[\eval_g]$ where $\pickb \smallstep[] \gtfalse$. By operational semantics, we have $\eval \ \evalcon_2[\eval_g] \smallsteps[] \substitute{\exprcon_v}{\evalcon_2[\eval_g]}{\ev}$. Since we have $\models\hastype{\evalcon_2[\eval_g]}{\syntype_1}$, and that $\models\hastype{\eval}{\mkfun{\syntype_1}{\syntype_2}}$, we can conclude that $\models\hastype{\substitute{\exprcon_v}{\evalcon_2[\eval_g]}{\ev}}{\syntype_2}$. Since $\syntype_2$ is a smaller type and that $\redcon{\eval_g}{\substitute{\exprcon_v}{\evalcon_2[\eval_g]}{\ev}} \smallsteps[] \eerror$, by induction hypothesis on $\textbf{clause (1)}$, we can conclude that $\redcon{\eval_g}{\mkgen{\syntype_2}} \smallsteps[] \eerror$. By operational semantics, we know that $\redexcon[\eval_g] = \eval_g \ \evalcon_2[\eval_g] \smallsteps[] \mkgen{\syntype_2}$ when $\pickb \smallstep[] \gtfalse$, therefore we can conclude that $\exprcon[\eval_g] = \redcon{\eval_g}{\redexcon[\eval_g]} \smallsteps[] \eerror$.
        
      \end{enumerate}

        \item $\redexcon =\ $\lstinline|if $\evalcon_1$ then $\exprcon_2$ else $\exprcon_3$|: Since the expression doesn't immediately step to an \lstinline|ERROR|, we know that $\evalcon[\eval]$ is a boolean value. This implies that $\evalcon_1 = \vtrue$ or $\evalcon_1 = \vfalse$, reducing it to the case where $\redexcon$ is parametric.
        
        \item $\redexcon = \matches{\evalcon}{p}$: There are three possible cases to consider here: $\evalcon[\eval]$ is an integer, or $\evalcon[\eval]$ is a boolean, or $\evalcon[\eval]$ is a function value. In the first two cases, we can conclude that $\evalcon$ doesn't contain any holes, which means $\evalcon[\eval] = \evalcon[\mkgen{\syntype}]$. Therefore, we only need to consider the scenario where $\evalcon[\eval]$ is a function value, which means either $\evalcon = \circ$ or $\evalcon = \mkfunv{\ev}{\exprcon_f[\circ]}$. In both cases, it doesn't matter if we're substituting $\eval$ or $\mkgen{\syntype}$ into the hole, because the substitution result will always be functions, giving us the same pattern-matching results for both $\redexcon[\eval]$ and $\redexcon[\mkgen{\syntype}]$. 
      \end{enumerate}
    \end{enumerate}
  
  \sloppy
  Consider \textbf{clause (2)}. We need to prove that if $\expr \smallsteps[] \eval_f$ and that $\not \models \hastype{\eval_f}{\mkfun{\syntype_1}{\syntype_2}}$, then $\mkche{\mkfun{\syntype_1}{\syntype_2}}{\eval_f} \smallsteps[] \eerror$ or $\mkche{\mkfun{\syntype_1}{\syntype_2}}{\eval_f} \smallsteps[] \gtfalse$. Expanding the definition, we have $\mkche{\mkfun{\syntype_1}{\syntype_2}}{\eval_f} = $ \lstinline|if $\expr\ \tttilde$ fun then let arg = |\mkgen{$\syntype_1$}\ \lstinline| in |\mkche{$\syntype_2$}{\lstinline|$\expr\ $ arg|} \lstinline| else false|. If $\eval_f$ is not a function value, we will have $\matches{\expr}{\gtfun} \smallsteps[] \gtfalse$, which means $\mkche{\mkfun{\syntype_1}{\syntype_2}}{\eval_f} \smallstep[] \gtfalse$. Now, let us consider the case where $\eval_f$ is a function value. Since $\not \models \hastype{\eval_f}{\mkfun{\syntype_1}{\syntype_2}}$, we know that there must exist some $\eval_0$ such that $\models \hastype{\eval_0}{\syntype_1}$ and $\not \models \hastype{\eval_f \ \eval_0}{\syntype_2}$. Since $\syntype_2$ is a smaller type, by induction hypothesis on \textbf{clause (2)}, we can conclude that $\neg\tc{\eval_f \ \eval_0}{\syntype_2}$, which means that $\mkche{\syntype_2}{\eval_f \ \eval_0} \smallsteps[] \gtfalse$ or $\mkche{\syntype_2}{\eval_f \ \eval_0} \smallsteps[] \eerror$. Let $\exprcon = $ \lstinline|if |$\mkche{\syntype_2}{\eval_f \ \circ}$ \lstinline| then 1 else ERROR|. We know that $\exprcon[\eval_0] \smallsteps[] \eerror$. Since $\syntype_1$ is a smaller type, by induction hypothesis on \textbf{clause (1)}, we can conclude that $\exprcon[\mkgen{\syntype_1}] \smallsteps[] \eerror$. This implies that $\mkche{\syntype_2}{\eval_f \ \mkgen{\syntype_2}} \smallsteps[] \gtfalse$ or $\mkche{\syntype_2}{\eval_f \ \mkgen{\syntype_2}} \smallsteps[] \eerror$. Since $\mkche{\mkfun{\syntype_1}{\syntype_2}}{\eval_f} \smallsteps[] \mkche{\syntype_2}{\eval_f \ \mkgen{\syntype_2}}$ by operational semantics, we can conclude that $\mkche{\mkfun{\syntype_1}{\syntype_2}}{\eval_f} \smallsteps[] \gtfalse$ or $\mkche{\mkfun{\syntype_1}{\syntype_2}}{\eval_f} \smallsteps[] \eerror$, i.e. $\neg\tc{\eval_f}{\mkfun{\syntype_1}{\syntype_2}}$.
  

\end{proof}
 
\soundness*

\begin{proof}
  This is equivalent to showing: if $\not\models \hastype{\expr}{\syntype}$, then \mkche{\syntype}{\expr} $\smallsteps[] \eerror$ or \mkche{\syntype}{\expr} $\smallsteps[] \gtfalse$. Since we know that $\not \models \hastype{\expr}{\syntype}$, we can safely assume that $\expr \not\Uparrow$. 

  Consider the case where $\expr \smallsteps[] \eerror$. By the operational semantics, we know that if $\expr \smallsteps[] \eerror$, then \mkche{\syntype}{\expr} $\smallsteps[] \eerror$. Therefore, it suffices to show that $\forall \expr.$ if $\expr \smallsteps[] \eval$ and $\not \models \hastype{\eval}{\syntype}$, then $\neg\tc{\expr}{\syntype}$, which is case (2) in Lemma~\ref{lemma:soundness_extra}.

\end{proof}

\subsection{Proof for Extensions}
\label{app_proofs_ext}

This section contains proofs for the core language extended with features introduced in Definitions~\ref{def:lr_refinement_dependent} and \ref{def:lr_poly}. The full extended grammar is defined in Figure~\ref{fig:ext_grammar}, and the operational semantics for this language is defined in Figure~\ref{fig:extended_op_sem}. The definitions for bound, free, and closed are standard, and definitions of $\smallstep[]$, $\smallsteps[]$, and $\Uparrow$ are the same as Definition~\ref{def:eval_relation}. The full extended semantic typing as well as checker/generator/wrapper definitions are provided in Definitions~\ref{def:lr_ext}, \ref{def:checker_ext}, \ref{def:generator_ext}, and \ref{def:wrap_ext}.

\begin{figure}[hbt!]
  \begin{grammar}
  \grule[values]{\eval}{
    \mathbb{Z}
    \gor \mathbb {B}
    \gor
    \mkfunv{\ev}{\expr}
    \gor \texttt{\small V(}\alpha\texttt{\small)}
  }
  \grule[expressions]{\expr}{
      \eval
      \gor \ev
      \gor \expr\ \expr
      \gor \expr\ \binop\ \expr
      \gline 
      \gor \ife{\expr}{\expr}{\expr}
      \gor \matches{\expr}{p}
      \gline
      \gor \picki 
      \gor \pickb
      \gor \expr \simeq \tpoly
      \gor \mzero
      \gor \eerror
  }
  \grule[variables]{\ev}{
              \textit{(identifiers)}
  }
  \grule[patterns]{p}{
    \tint
    \gor \tbool
    \gor \pfun
  }
  \grule[binops]{\binop}{
    \gtplus
    \gor \gtminus
    \gor \gtless
    \gor \gteq
    \gor \gtand
    \gor \gtor
    \gor \gtxor
  }
  \grule[type variables]{\tpoly}{
              $\lstinline`'a`\ $
              \gor $\lstinline`'b`\ $
              \gor \ldots
  }
  \grule[types]{\syntype}{
    \tint
    \gor \tbool
    \gor \tfun
    \gline
    \gor \mktset{\syntype}{\expr}
    \gor \mkdfun{\ev}{\syntype}{\syntype}
    \gor \tpoly
  }
  \grule[redexes]{r}{
    \eval \binop \eval
    \gor \eval \ \eval
    \gor \ife{\vtrue}{\expr}{\expr}
    \gline
    \gor \ife{\vfalse}{\expr}{\expr}
    \gline
    \gor \matches{\eval}{p}
    \gor \picki \gor \pickb
  }
  \grule[reduction contexts]{\reductioncon}{
    \bullet
    \gor \reductioncon \ \expr
    \gor \expr \ \reductioncon
    \gor \ife{\reductioncon}{\expr}{\expr}
    \gline
    \gor \reductioncon \binop \expr
    \gor \eval \binop \reductioncon
    \gor \matches{\reductioncon}{p}
    \gor \reductioncon \simeq \tpoly
}
\end{grammar}
\caption{Extended Language Grammar}
\label{fig:ext_grammar}
\end{figure}

\begin{figure}[ht]
  \begin{mathpar}

      \bbrule{Red}{
        r \smallstep[] \expr
      }{
        R\lbbr r \rbbr \smallstep[] R'\lbbr\expr\rbbr
      }

      \bbrule{Err}{
        \expr \smallstep[] \eerror
      }{
        R\lbbr\expr\rbbr \smallstep[] \eerror
      }

      \bbrule{Add}{
        n_1, n_2 \in \mathbb{Z}
      }{
        n_1 + n_2 \smallstep[] \text{integer sum of } n_1 \text{ and } n_2
      }

      \bbrule{Add-Err}{
        \eval_1 \text{ or } \eval_2 \text{ is not an integer}
      }{
        \eval_1 + \eval_2 \smallstep[] \eerror
      }

      \bbrule{Appl}{
      }{
        (\mkfunv{\ev}{\expr_f}) \ \eval \smallstep[] \substitute{\expr_f}{\eval}{\ev}
      }

      \bbrule{Appl-Err}{
        \eval \text{ is not a function value}
      }{
        \eval \ \eval' \smallstep[] \eerror
      }

      \bbrule{If-True}{
        }{
          \ife{\vtrue}{\expr}{\expr'} \smallstep[] \expr
      }

      \bbrule{If-False}{
        }{
          \ife{\vfalse}{\expr}{\expr'} \smallstep[] \expr'
      }

      \bbrule{If-Err}{
        \eval \text{ is not a boolean value}
      }{
        \ife{\eval}{\expr}{\expr'} \smallstep[] \eerror
      }

      \bbrule{Nondet-int}{
        n \in \mathbb{Z}
      }{
        \picki \smallstep[] n
      }

      \bbrule{Nondet-bool}{
        b \in \mathbb{B}
      }{
        \pickb \smallstep[] b
      }

      \bbrule{Pattern}{
        \beta = \textsc{matches}(\eval, p)
      }{
        \matches{\eval}{p} \smallstep[] \beta
      }

      \bbrule{Opaque Pattern}{
      }{
        \matches{\texttt{V}\ttop\alpha\ttcp}{p} \smallstep[] \eerror
      }

      \bbrule{Mzero}{
          \expr \smallstep[] \mzero
      }{
        R\lbbr \expr \rbbr \smallstep[] \mzero
      }

      \bbrule{Poly-Check-True}{
        \eval = \texttt{V}\ttop\alpha\ttcp
      }{
        \eval \simeq \alpha \smallstep[] \vtrue
      }

      \bbrule{Poly-Check-False}{
        \eval \neq \texttt{V}\ttop\alpha\ttcp
      }{
        \eval \simeq \alpha \smallstep[] \vfalse
      }

  \end{mathpar}
  \caption{Operational Semantics for Extended Language}
  \label{fig:extended_op_sem}
\end{figure}

\begin{definition}[Semantic Types for Extended Language]
  \label{def:lr_ext}
  \ \par
  \begin{enumerate}
      \item $\ \models \hastype{\expr}{\tint}$ iff $\expr \nsmallsteps[] \eerror$ and $\forall \eval.$ if $\expr \smallsteps[] \eval$, then $\eval \in \mathbb{Z}$.
      \item $\ \models \hastype{\expr}{\tbool}$ iff $\expr \nsmallsteps[] \eerror$ and $\forall \eval.$ if $\expr \smallsteps[] \eval$, then $\eval \in \mathbb{B}$.
      \item $\ \models \hastype{\expr}{\mkfun{\syntype_1}{\syntype_2}}$ iff $\expr \nsmallsteps[] \eerror$ and $\forall \eval_f.$ if $\expr \smallsteps[] \eval_f$, then $\forall{\eval}.$ if $\models \hastype{\eval}{\syntype_1}$, then $\models \hastype{\eval_f\ \eval}{\syntype_2}$.
      \item $\ \models \hastype{\expr}{\mktset{\syntype}{\expr_p}}$ iff $\expr \nsmallsteps[] \eerror$, $\models \hastype{\expr_p}{\mkfun{\syntype}{\tbool}}$, and $\forall \eval.$ if $\expr \smallsteps[] \eval$, then $\ \models \hastype{\eval}{\syntype}$ and $\forall \eval_p.$ if $\expr_p \ \eval \smallsteps[] \eval_p$, then $\eval_p = \gttrue$.
      \item $\ \models \hastype{\expr}{\mkdfun{\ev}{\syntype_1}{\syntype_2}}$ iff $\expr \nsmallsteps[] \eerror$, and $\forall{\eval_f}.$ if $\expr \smallsteps[] \eval_f$, then $\forall \eval.$ if $\models \hastype{\eval}{\syntype_1}$, then $\models \hastype{\eval_f\ \eval}{\substitute{\syntype_2}{\eval}{\ev}}$.
      \item $\ \models \hastype{\expr}{\tpoly}$ iff $\expr \nsmallsteps[] \eerror$ and $\forall \eval.$ if $\expr \smallsteps[] \eval$, then $\eval \simeq \alpha \smallsteps[] \vtrue$.
   \end{enumerate}
\end{definition}

\begin{definition}[Checker for Extended Language]
  \label{def:checker_ext}
  \ \par
  \begin{enumerate}
      \item \mkche{\tint}{\expr} = $\expr$ $\tttilde$ int
      \item \mkche{\tbool}{\expr} = $\expr$ $\tttilde$ bool
      \item $\mkche{\mkfun{\syntype_1}{\syntype_2}}{\expr}$ = \\ 
      \lstinline|if $\expr\ \tttilde$ fun then let arg = |\mkgen{$\syntype_1$}\ \lstinline| in |\mkche{$\syntype_2$}{\lstinline|$\expr\ $ arg|} \lstinline| else false|
      \item \mkche{$\mktset{\syntype}{\expr_p}$}{\expr} = \mkche{\syntype}{\expr}\ \ \lstinline| and | ($\expr_p$ \ \expr)
      \item \mkche{$\mkdfun{\ev}{\syntype_1}{\syntype_2}$}{\expr} = \\\lstinline|if $\expr\ \tttilde$ fun then let arg =|\ \mkgen{$\syntype_1$}\ \lstinline|in| \mkche{\substitute{$\syntype_2$}{\lstinline|arg|}{\ev}}{\lstinline|($\expr\ $ arg)|}\ \lstinline| else false|
      \item \mkche{$\tpoly$}{$\expr$} = $\expr \simeq \tpoly$
   \end{enumerate}
\end{definition}

\begin{definition}[Generator for Extended Language]
  \label{def:generator_ext}
  \ \par 
  \begin{enumerate}
    \item \mkgen{tint} = \lstinline|pick_i|
    \item \mkgen{\tbool} = \lstinline|pick_b|
    \item \mkgen{$\mkfun{\syntype_1}{\syntype_2}$} = \lstinline|fun x -> $\ $if| \pickb \ \lstinline|then|
    
    \quad \lstinline|if| $\mkche{\syntype_1}{\texttt{x}}$ \lstinline|then| $\mkgen{\syntype_2}$ \lstinline|else| \eerror 
    
    \lstinline|else| $\mkgen{\syntype_2}$
    \item \mkgen{$\mktset{\syntype}{\expr_p}$} = \lstinline|let gend = | \mkgen{\syntype} \lstinline| in if ($\expr_p$ gend) then gend else mzero|
    
    \item \mkgen{$\mkdfun{\ev}{\syntype_1}{\syntype_2}$} = \lstinline|fun $\ev'$ ->$\ $if| \pickb \ \lstinline|then|
    
    \quad \lstinline|if| $\mkche{\syntype_1}{\ev'}$ \lstinline|then| \mkgen{\substitute{$\syntype_2$}{$\ev'$}{$\ev$}} \lstinline|else| \eerror 
    
    \lstinline|else| \mkgen{\substitute{$\syntype_2$}{$\ev'$}{$\ev$}}

    \item \mkgen{$\tpoly$} = \lstinline`V($\alpha$)`
  \end{enumerate}
\end{definition}

\begin{definition}[Wrapper for Extended Language]
  \label{def:wrap_ext}
  \ \par
  \begin{enumerate}
      \item \mkwrap{\tint}{\expr} = \expr
      \item \mkwrap{\tbool}{\expr} = \expr
      \item \mkwrap{$\mkfun{\syntype_1}{\syntype_2}$}{\expr} = \lstinline|fun x ->$\ $if| $\pickb$ \lstinline|then|
      
      \quad \lstinline|if| $\mkche{\syntype_1}{\texttt{x}}$ \lstinline|then| $\mkwrap{\syntype_2}{\expr \ \mkwrap{\syntype_1}{\texttt{x}}}$ \lstinline|else ERROR|\\
      \lstinline|else| $\mkwrap{\syntype_2}{\expr \ \mkwrap{\syntype_1}{\texttt{x}}}$ 

      \item $\mkwrap{\mktset{\syntype}{\expr_p}}{\expr} = \mkwrap{\syntype}{\expr}$
      \item $\mkwrap{\mkdfun{\ev}{\syntype_1}{\syntype_2}}{\expr} = $ \lstinline|fun $\ev'$ ->$\ $if| $\pickb$ \lstinline|then|
        
      \quad \lstinline|if| \mkche{$\syntype_1$}{$\ev'$} \lstinline|then| \mkwrap{\substitute{$\syntype_2$}{$\ev'$}{$\ev$}}{$\expr$ \ \mkwrap{$\syntype_1$}{$\ev'$}} \lstinline|else ERROR|\\
      \lstinline|else| \mkwrap{\substitute{$\syntype_2$}{$\ev'$}{$\ev$}}{$\expr$ \ \mkwrap{$\syntype_1$}{$\ev'$}} 
      \item $\mkwrap{\tpoly}{\expr} = \expr$
   \end{enumerate}
\end{definition}

\begin{lemma}[Size of Type Subsitution]
  \label{lemma:size_of_type_subs}
  $\forall\syntype, \expr, \ev. \textsc{size(\substitute{\syntype}{\eval}{\ev})} \leq \textsc{size($\syntype$)}$.
\end{lemma}

\begin{proof}
  We will prove this lemma by induction on the size of $\syntype$.

  \textbf{Base case:} $\syntype = \tint$.

  In this case, since $\ev$ is not free in $\tint$, the substitution result will still have size $1$. 

  The proof for the case where $\syntype = \tbool$ and $\syntype = \tpoly$ is very similar, so we omit it here for brevity.

  \textbf{Inductive step:} 
  
  \begin{enumerate}
    \item $\syntype = \mkfun{\syntype_1}{\syntype_2}$.

    We know that $\textsc{size($\substitute{(\mkfun{\syntype_1}{\syntype_2})}{\eval}{\ev}$)} = \textsc{size($\substitute{\syntype_1}{\eval}{\ev}$)} + \textsc{size($\substitute{\syntype_2}{\eval}{\ev}$)} + 1$. By induction hypothesis, $\textsc{size($\substitute{\syntype_1}{\eval}{\ev}$)} \leq \textsc{size($\syntype_1$)}$, and $\textsc{size($\substitute{\syntype_2}{\eval}{\ev}$)} \leq \textsc{size($\syntype_2$)}$. Therefore, we can conclude that $\textsc{size($\substitute{(\mkfun{\syntype_1}{\syntype_2})}{\eval}{\ev}$)} \leq \textsc{size($\syntype_1$)} + \textsc{size($\syntype_2$)} + 1 = \textsc{size($\mkfun{\syntype_1}{\syntype_2}$)}$.
  
    The proof for the case where $\syntype = \mkdfun{\ev}{\syntype_1}{\syntype_2}$ is very similar, so we omit it here for brevity.

    \item $\syntype = \mktset{\syntype_0}{\expr_p}$.

    We know that $\textsc{size($\substitute{\mktset{\syntype_0}{\expr_p}}{\eval}{\ev}$)} = \textsc{size($\substitute{\syntype_0}{\eval}{\ev}$)} + 1$. By induction hypothesis, $\textsc{size($\substitute{\syntype_0}{\eval}{\ev}$)} \leq \textsc{size($\syntype_0$)}$. Therefore, we can conclude that $\textsc{size($\substitute{\mktset{\syntype_0}{\expr_p}}{\eval}{\ev}$)} \leq \textsc{size($\syntype_0$)} + 1 = \textsc{size($\mktset{\syntype_0}{\expr_p}$)}$.

  \end{enumerate}
  
\end{proof}

\begin{lemma}\label{lemma:completeness_ext_extra}For all types $\syntype$ defined in Definitions~\ref{def:lr_core}, \ref{def:lr_refinement_dependent}, and \ref{def:lr_poly}, 
  \begin{enumerate}
    \item $\models$ \mkgen{\syntype} $\colon$\syntype, and 
    \item $\forall \expr.$ if \mkche{\syntype}{\expr} $\smallsteps[] \eerror$ or if \mkche{\syntype}{\expr}$\smallsteps[] \gtfalse$, then $\not \models$ \expr $\colon$\syntype.
\end{enumerate}
\end{lemma}

\begin{proof}
  We prove this lemma by induction on the size of $\syntype$.

  \textbf{Base case (1)}: $\syntype = \tpoly$

  Consider \textbf{clause (1)}. We need to show that $\models \hastype{\mkgen{\tpoly}}{\tpoly}$.
  
  Since \mkgen{\tpoly} = \texttt{V}\ttop\tpoly\ttcp, by definition, we have $\models \hastype{\mkgen{\tpoly}}{\tpoly}$.

  Consider \textbf{clause (2)}. we have to prove that for an arbitrary $\expr$, if \mkche{\tpoly}{\expr} $\smallsteps[]$ \lstinline|ERROR|, or if \mkche{\tpoly}{\expr}$\smallsteps[] \gtfalse$, then $\not\models \hastype{\expr}{\tpoly}$. By definition, \mkche{\tpoly}{\expr} = $\expr \simeq \tpoly$. Proceed by case analysis on the evaluation result of the checker.

  \begin{enumerate}
    \item $\expr \smallsteps[]$ \lstinline|ERROR|. In this case, $\not\models\hastype{\expr}{\tpoly}$ is trivially true. This is also the only way that \mkche{\tpoly}{\expr} can evaluate to \lstinline|ERROR|, since $\expr \simeq \tpoly$ will not return an \lstinline|ERROR| if $\expr \smallsteps[] \eval$ for some value $\eval$.
    \item $\expr \smallsteps[] \eval_0$. In this case, by the operational semantics, we know that $\eval_0 \simeq \tpoly \smallsteps[] \gtfalse$, and this implies that $\eval_0 \neq \texttt{V}\ttop\tpoly\ttcp$. Therefore, we have $\not\models\hastype{\expr}{\tpoly}$ by definition. 
  \end{enumerate}
  
  The proof for the cases where $\syntype = \tint$ or $\tbool$ is similar, so we will omit it here for brevity.

\textbf{Inductive step (1):} \syntype\ =\ $\mktset{\syntype_0}{\expr_p}$

Consider \textbf{clause (1)}.

\sloppy By definition, $\mkgen{\mktset{\syntype_0}{\expr_p}} =$ \lstinline|let gend = $\ \mkgen{\syntype_0}$ in if ($\expr_p$ gend) then gend| \lstinline|else mzero|. By induction hypothesis on \textbf{clause (1)}, we know that $\models\hastype{\mkgen{\syntype_0}}{\syntype_0}$. Now, consider the evaluation result of the application, $\expr_p\ \mkgen{\syntype_0}$.

Since we require predicates to be total, we do not have to consider the case where $\expr_p\ \mkgen{\syntype_0}$ does not terminate. We also can't have $\expr_p\ \mkgen{\syntype_0}$ $\smallsteps[] \eerror$, because we required the predicate to have the type $\mkfun{\syntype_0}{\tbool}$, which means the predicate cannot return $\eerror$ on any arguments with type $\syntype_0$. 

If $\expr_p\ \mkgen{\syntype_0}$ $\smallsteps[] \gttrue$, we know that the generated value satisfies the predicate. If $\expr_p\ \mkgen{\syntype_0}$ $\smallsteps[] \gtfalse$, we know that the generated value does not satisfy the predicate, and the conditional expression evaluates to \lstinline|mzero|. Consequently, we know that $\mkgen{\mktset{\syntype_0}{\expr_p}}$ either evaluates to a value $\eval$ where $\models\hastype{\eval}{\syntype_0}$ and $\expr_p \ \ \eval \smallsteps[] \gttrue$ or to \lstinline|mzero| (a non-value in our language, which by definition can take on arbitrary type), which implies $\models \hastype{\mkgen{\mktset{\syntype_0}{\expr_p}}}{\mktset{\syntype_0}{\expr_p}}$. 

Consider \textbf{clause (2)}.

We need prove that for any $\expr$, if $\mkche{\mktset{\syntype_0}{\expr_p}}{\expr} \smallsteps[] \eerror$ or if $\mkche{\mktset{\syntype_0}{\expr_p}}{\expr}\smallsteps[] \gtfalse$, then $\not\models \hastype{\expr}{\mktset{\syntype_0}{\expr_p}}$. We only consider the case where $\expr \not\smallsteps[] \eerror$ in the following discussion, since $\expr \smallsteps[] \eerror$ will make $\not\models\hastype{\expr}{\mktset{\syntype_0}{\expr_p}}$ true trivially. 

If $\mkche{\mktset{\syntype_0}{\expr_p}}{\expr} \smallsteps[] \eerror$, it implies that either $\mkche{\syntype_0}{\expr} \smallsteps[] \eerror$ or $\expr_p \ \expr \smallsteps[] \eerror$. In the first case, by induction hypothesis on \textbf{clause (2)}, we have $\not\models\hastype{\expr}{\syntype_0}$, and by definition, it also implies that $\not\models\hastype{\expr}{\mktset{\syntype_0}{\expr_p}}$. In the second case, since we know that $\models\hastype{\expr_p}{\mkfun{\syntype_0}{\tbool}}$, it must be the case that $\not\models\hastype{\expr}{\syntype_0}$, implying that $\not\models\hastype{\expr}{\mktset{\syntype_0}{\expr_p}}$.

If $\mkche{\mktset{\syntype_0}{\expr_p}}{\expr}\smallsteps[] \gtfalse$, it implies that either $\mkche{\syntype_0}{\expr} \smallsteps[] \gtfalse$ or $\expr_p \ \expr \smallsteps[] \gtfalse$. In the first case, by induction hypothesis on \textbf{clause (2)}, we have $\not\models\hastype{\expr}{\syntype_0}$, and by definition, it also implies that $\not\models\hastype{\expr}{\mktset{\syntype_0}{\expr_p}}$. In the second case, since the predicate check failed, we have $\not\models\hastype{\expr}{\mktset{\syntype_0}{\expr_p}}$ by definition.

\textbf{Inductive step (2):} \syntype\ =\ $\mkdfun{\ev}{\syntype_1}{\syntype_2}$

Consider \textbf{clause (1)}. We need to show that $\models \hastype{\mkgen{\mkdfun{\ev}{\syntype_1}{\syntype_2}}}{\mkdfun{\ev}{\syntype_1}{\syntype_2}}$. To prove this, we need $\forall\eval.$ if $\models \hastype{\eval}{\syntype_1}$, then $\models \hastype{(\mkgen{\mkdfun{\ev}{\syntype_1}{\syntype_2}} \ \eval)}{\substitute{\syntype_2}{\eval}{\ev}}$.
  
By definition, $\mkgen{\mkdfun{\ev}{\syntype_1}{\syntype_2}}$ is nondeterministic. There are two cases to consider here:

\begin{enumerate}
  \item \sloppy Argument is checked: In this case, we have $\mkgen{\mkdfun{\ev}{\syntype_1}{\syntype_2}} \ \eval \smallsteps[]$ \lstinline|if|$\ \mkche{\syntype_1}{\eval}$ \lstinline|then| $\mkgen{\substitute{\syntype_2}{\eval}{\ev}}$ \lstinline|else ERROR|. By induction hypothesis on \textbf{clause (2)}, $\forall\eval.$ if $\models \hastype{\eval}{\syntype_1}$, then $\mkche{\syntype_1}{\eval}\nsmallsteps[] \gtfalse$ and $\mkche{\syntype_1}{\eval}\nsmallsteps[] \eerror$. This implies that either $\mkche{\syntype_1}{\eval}$ diverges or $\mkche{\syntype_1}{\eval} \smallsteps[] \gttrue$. If $\mkche{\syntype_1}{\eval}$ diverges, $\mkgen{\mkdfun{\ev}{\syntype_1}{\syntype_2}} \ \eval$ will diverge, too, making the statement $\models \hastype{\mkgen{\mkdfun{\ev}{\syntype_1}{\syntype_2}} \ \eval}{\substitute{\syntype_2}{\eval}{\ev}}$ trivially true. If $\mkche{\syntype_1}{\eval} \smallsteps[] \gttrue$, we only have to consider $\mkgen{\substitute{\syntype_2}{\eval}{\ev}}$. By Lemma~\ref{lemma:size_of_type_subs}, we know that $\substitute{\syntype_2}{\eval}{\ev}$ is a smaller type than $\mkdfun{\ev}{\syntype_1}{\syntype_2}$, and by induction hypothesis on \textbf{clause (1)}, we know that $\models \hastype{\mkgen{\substitute{\syntype_2}{\eval}{\ev}}}{\substitute{\syntype_2}{\eval}{\ev}}$. 
  \item Argument is not checked: In this case, we have $\mkgen{\mkdfun{\ev}{\syntype_1}{\syntype_2}} \ \eval \smallsteps[] \mkgen{\substitute{\syntype_2}{\eval}{\ev}}$. By Lemma~\ref{lemma:size_of_type_subs}, we know that $\substitute{\syntype_2}{\eval}{\ev}$ is a smaller type than $\mkdfun{\ev}{\syntype_1}{\syntype_2}$, and by induction hypothesis, we know that $\models \hastype{\mkgen{\substitute{\syntype_2}{\eval}{\ev}}}{\substitute{\syntype_2}{\eval}{\ev}}$.
\end{enumerate}
Therefore, we have shown that $\forall\eval.$ if $\hastype{\eval}{\syntype_1}$, then $\models \hastype{\mkgen{\mkdfun{\ev}{\syntype_1}{\syntype_2}} \ \eval}{\substitute{\syntype_2}{\eval}{\ev}}$. We can now safely conclude that $\models \hastype{\mkgen{\mkdfun{\ev}{\syntype_1}{\syntype_2}}}{\mkdfun{\ev}{\syntype_1}{\syntype_2}}$.

Consider \textbf{clause (2)}. we need to prove that for any $\expr$, if $\mkche{\mkdfun{\ev}{\syntype_1}{\syntype_2}}{\expr} \smallsteps[] \eerror$ or if $\mkche{\mkdfun{\ev}{\syntype_1}{\syntype_2}}{\expr}\smallsteps[] \gtfalse$, then $\not\models \hastype{\expr}{\mkdfun{\ev}{\syntype_1}{\syntype_2}}$. If $\expr \smallsteps[] \eerror$, we have $\not\models\hastype{\expr}{\mkdfun{\ev}{\syntype_1}{\syntype_2}}$ trivially. We will only consider the case where $\expr \not\smallsteps[] \eerror$ in the discussion below.

Let us first consider the case, $\mkche{\mkdfun{\ev}{\syntype_1}{\syntype_2}}{\expr} \smallsteps[] \eerror$. Expanding the checker definition, we get $\mkche{\mkdfun{\ev}{\syntype_1}{\syntype_2}}{\expr} = $\lstinline|if $\expr\ \tttilde$ fun then let arg =|\ $\mkgen{\syntype_1}$\ \lstinline|in| $\mkche{\substitute{\syntype_2}{\eval}{\ev}}{\expr \ \mkgen{\syntype_1}}$\ \lstinline| else false|. 

Since $\mkgen{\syntype_1}$ does not evaluate to \lstinline|ERROR|, we know it must come from either $\matches{\expr}{\gtfun}$ or $\mkche{\substitute{\syntype_2}{\eval}{\ev}}{\expr \ \mkgen{\syntype_1}}$. 

If $\matches{\expr}{\gtfun} \smallsteps[] \eerror$, it implies that $\expr \smallsteps[] \texttt{V}\ttop\tpoly\ttcp$ for some polymorphic variable $\tpoly$, from which we can conclude that $\not\models\hastype{\expr}{\mkdfun{\ev}{\syntype_1}{\syntype_2}}$. 

If $\mkche{\substitute{\syntype_2}{\eval}{\ev}}{\expr \ \mkgen{\syntype_1}} \smallsteps[] \eerror$, by induction hypothesis on \textbf{clause (2)}, we know that $\not\models \hastype{(\expr \ \mkgen{\syntype_1})}{\substitute{\syntype_2}{\eval}{\ev}}$, and by induction hypothesis, we have \textbf{clause (1)} $\models \hastype{\mkgen{\syntype_1}}{\syntype_1}$. Thus we have found a witness $\models \hastype{\mkgen{\syntype_1}}{\syntype_1}$ such that $\not\models \hastype{(\expr \ \mkgen{\syntype_1})}{\substitute{\syntype_2}{\eval}{\ev}}$, proving that $\not\models \hastype{\expr}{\mkdfun{\ev}{\syntype_1}{\syntype_2}}$.

Now, consider the case, $\mkche{\mkdfun{\ev}{\syntype_1}{\syntype_2}}{\expr}\smallsteps[] \gtfalse$. By the \textsc{checker} definition, we have the following possible cases:
\begin{enumerate}
  \item $\matches{\expr}{\gtfun} \smallsteps[] \gtfalse$. In this case, by the operational semantics, we know that $\expr \smallsteps[] \eval_0$ and $\eval_0$ is not a function. Therefore, we have $\not\models\hastype{\expr}{\mkdfun{\ev}{\syntype_1}{\syntype_2}}$ trivially.
  \item $\matches{\expr}{\gtfun} \smallsteps[] \gttrue$. In this case, we know that $\expr \smallsteps[] \eval_0$ and $\eval_0$ is a function. However, this suggests that $\mkche{\substitute{\syntype_2}{\eval}{\ev}}{\expr \ \mkgen{\syntype_1}} \smallsteps[] \gtfalse$. By induction hypothesis on \textbf{clause (2)}, we know that $\not\models \hastype{(\expr \ \mkgen{\syntype_1})}{\substitute{\syntype_2}{\eval}{\ev}}$, and by induction hypothesis on \textbf{clause (1)}, we have $\models \hastype{\mkgen{\syntype_1}}{\syntype_1}$. Thus we have found a witness $\models \hastype{\mkgen{\syntype_1}}{\syntype_1}$ such that $\not\models \hastype{(\expr \ \mkgen{\syntype_1})}{\substitute{\syntype_2}{\eval}{\ev}}$, proving that $\not\models \hastype{\expr}{\mkdfun{\ev}{\syntype_1}{\syntype_2}}$.
\end{enumerate}

The proof for $\syntype = \mkfun{\syntype_1}{\syntype_2}$ is similar, so we omit it here for brevity.

\end{proof}

\begin{restatable}[Completeness of Extended System]{lemma}{completenessExt}
  \label{lemma:completeness_ext}For all types $\syntype$ defined in Definitions~\ref{def:lr_core}, \ref{def:lr_refinement_dependent}, and \ref{def:lr_poly}, $\forall \expr.$ if $\models \hastype{\expr}{\syntype}$, then $\tc{\expr}{\syntype}$.
\end{restatable}

\begin{proof}
  This is equivalent to showing: if \mkche{\syntype}{\expr} $\smallsteps[] \eerror$ or if $\exists\eval. \mkche{\syntype}{\expr} \smallsteps[] \eval$ and $\eval \neq \gttrue$, then $\not\models \hastype{\expr}{\syntype}$. By examining Definitions~\ref{def:checker_core}, \ref{def:checker_refinement_and_dependent} and \ref{def:checker_poly}, we can see that \textsc{checker} can only return \lstinline|ERROR| or boolean values, making the completeness statement follow immediately from Lemma~\ref{lemma:completeness_ext_extra}.
\end{proof}


\begin{restatable}[Soundness of Extended System]{lemma}{soundnessExtExtra}
  \label{lemma:soundness_ext_extra}
  For all types $\syntype$ defined in Definitions~\ref{def:lr_core}, \ref{def:lr_refinement_dependent}, and \ref{def:lr_poly},
  \begin{enumerate}
    \item $\forall \eval.$ if $\models\hastype{\eval}{\syntype}$, then $\forall\exprcon.$ if $\exprcon[\eval] \smallsteps[] \eerror$, then $\exprcon[\mkgen{\syntype}] \smallsteps[] \eerror$.
    \item $\forall\expr.$ if $\expr \smallsteps[] \eval$ and $\not\models\hastype{\eval}{\syntype}$, then $\neg\tc{\expr}{\syntype}$.
\end{enumerate}
\end{restatable}

\begin{proof}
  We will prove this lemma by induction on the size of $\syntype$.

  \textbf{Base case: }$\syntype = \tpoly$

  Consider \textbf{clause (1)}. Since $\models\hastype{\eval}{\tpoly}$, we know that $\eval = \texttt{V}\ttop\tpoly\ttcp$. Since $\mkgen{\tpoly} = \texttt{V}\ttop\tpoly\ttcp$, it implies that $\mkgen{\tpoly} = \eval$. Because we already have $\exprcon[\eval] \smallsteps[] \eerror$, we can safely conclude that $\exprcon[\mkgen{\tpoly}] \smallsteps[] \eerror$.

  Consider \textbf{clause (2)}. Given $\not \models \hastype{\expr}{\tpoly}$ and $\expr \smallsteps[] \eval$, we can conclude that $\eval \neq \texttt{V}\ttop\tpoly\ttcp$. By definition, $\mkche{\tpoly}{\expr}$ = $\expr \simeq \tpoly$. Since $\expr \smallsteps[] \eval$ and $\eval \neq \texttt{V}\ttop\tpoly\ttcp$, we can conclude $\eval \simeq \tpoly \smallsteps[] \gtfalse$, establishing this case.

  The proof for the case where $\syntype = \tint$ and $\syntype = \tbool$ is very similar, so we omit it here for brevity.

  \textbf{Inductive step (1): }$\syntype = \mktset{\syntype_0}{\expr_p}$

    Consider \textbf{clause (1)}.
    
    Since $\models\hastype{\eval}{\mktset{\syntype_0}{\expr_0}}$, we know that $\forall \eval_p.$ if $\expr_p \ \eval \smallsteps[] \eval_p$, then $\eval_p = \gttrue$. Furthermore, we know this must be true for all evaluations of $\expr_p \ \eval$, since $\expr_p$ must be a total function. Consequently, defining $\exprcon' = \ife{\expr_p \ \circ}{\eerror}{1}$, we have $\exprcon'[\eval] \smallsteps[] \eerror$. Since $\models\hastype{\eval}{\syntype_0}$ and $\syntype_0$ is a smaller type, we can conclude by induction on \textbf{clause (1)} that $\exprcon'[\mkgen{\syntype_0}] \smallsteps[] \eerror$, which implies that there exists some evaluation such that $\expr_p \ \mkgen{\syntype_0} \smallsteps[] \gttrue$. By definition, we have $\mkgen{\mktset{\syntype_0}{\expr_p}} = $ \lstinline|let gend = $\ \mkgen{\syntype_0}$ in if $\expr_p$ gend then gend| \lstinline|else mzero|. Since we have proven that $\expr_p \ \mkgen{\syntype_0}$ must evaluate to $\gttrue$ for some execution, we can conclude that there must exist some evaluation such that $\mkgen{\mktset{\syntype_0}{\expr_p}} \smallsteps[] \mkgen{\syntype_0}$. 
    
    Since $\syntype_0$ is a smaller type, $\models\hastype{\eval}{\syntype_0}$, and that $\exprcon[\eval] \smallsteps[] \eerror$, by induction hypothesis on \textbf{clause (1)}, we have $\exprcon[\mkgen{\syntype_0}] \smallsteps[] \eerror$. Since we have shown that $\mkgen{\mktset{\syntype_0}{\expr_p}} \smallsteps[] \mkgen{\syntype_0}$, we can safely conclude that $\exprcon[\mkgen{\mktset{\syntype_0}{\expr_p}}] \smallsteps[] \eerror$, establishing the case.

    Consider \textbf{clause (2)}. We need to prove that if $\expr \smallsteps[] \eval$ and that $\not \models \hastype{\eval}{\mktset{\syntype_0}{\expr_p}}$, then $\mkche{\mktset{\syntype_0}{\expr_p}}{\eval} \smallsteps[] \eerror$ or $\mkche{\mktset{\syntype_0}{\expr_p}}{\eval} \smallsteps[] \gtfalse$. Expanding the definition, we have $\mkche{\mktset{\syntype_0}{\expr_p}}{\eval} = \mkche{\syntype_0}{\eval}$ \lstinline|and| $\expr_p \ \eval$.

    There are two causes for $\not\models\hastype{\eval}{\mktset{\syntype_0}{\expr_p}}$:

    \begin{enumerate}
      \item \sloppy $\not\models\hastype{\eval}{\syntype_0}$: In this case, by induction hypothesis on \textbf{clause (2)}, we have $\mkche{\syntype_0}{\eval} \smallsteps[] \gtfalse$ or $\mkche{\syntype_0}{\eval} \smallsteps[] \eerror$, leading to $\mkche{\mktset{\syntype_0}{\expr_p}}{\eval} \smallsteps[] \gtfalse$ or $\mkche{\mktset{\syntype_0}{\expr_p}}{\eval} \smallsteps[] \eerror$ respectively.
      \item $\expr_p \ \eval \smallsteps[] \gtfalse$: In this case, we have $\mkche{\mktset{\syntype_0}{\expr_p}}{\eval} \smallsteps[] \gtfalse$ directly.
    \end{enumerate}

  \textbf{Inductive step (2):} $\syntype = \mkdfun{\ev}{\syntype_1}{\syntype_2}$. 


  Consider \textbf{clause (1)}. 

    Given that $\models \hastype{\eval}{\syntype}$, and $\exprcon[\eval] \smallsteps[] \eerror$, we need to show that $\exprcon[\mkgen{\syntype}] \smallsteps[] \eerror$. We will prove this by induction on the length of the $\exprcon[\eval] \smallsteps[] \eerror$ computation.

    We omit the trivial base case where $\exprcon = \eerror$, in which $\exprcon[\eval]$ will step to \lstinline|ERROR| in zero steps. 

    \textbf{Base case: }$\exprcon[\eval] \smallstep[] \eerror$

    By Lemma~\ref{lemma:unique_factor}, we know that there exists unique $\rccon$ and $\exprcon'$ such that $\exprcon[\circ] = \redcon{\circ}{\exprcon'[\circ]}$, where $\exprcon'[\circ] = \redexcon[\circ]$ or $\exprcon'[\circ] = \circ$. However, here we do not have to consider the case where $\exprcon'[\circ] = \circ$; since we're filling white holes with a value, a value by itself cannot be a redex. 

    There are two cases here we need two consider: 
    \begin{enumerate}
      \item $\redexcon$ is parametric: $\exprcon[\eval]$ will have to take more than one step to evaluate to an \lstinline|ERROR| expression, thus it cannot be in the base case.
      \item $\redexcon$ is not parametric: We proceed by case analysis on the redex context, $\redexcon$.

      \begin{enumerate}
        \item $\redexcon = \evalcon_1 + \evalcon_2$: According to the operational semantics, addition will evaluate to \lstinline|ERROR| if and only if one of the operands is a non-integer. If $\evalcon_1 = \circ$ or $\evalcon_2 = \circ$, since $\eval$ and $\mkgen{\syntype}$ both are function values, we can safely conclude that $\redexcon[\eval]$ and $\redexcon[\mkgen{\syntype}]$ will both evaluate to \lstinline|ERROR|. The other scenario is where at least one of the operands takes the form of \lstinline|fun x -> $\ \exprcon_f$| or $b$, where $b$ is a boolean value; in both cases, the expression will evaluate to \lstinline|ERROR| regardless of which value gets substituted into $\circ$.
        \item $\redexcon = \evalcon_1 \ \evalcon_2$: According to the operational semantics, function application will evaluate to \lstinline|ERROR| if and only if $\evalcon_1$ is an integer or boolean value. This implies $\evalcon_1$ doesn't contain any $\circ$, since the relevant values being substituted here are functions. Therefore, this can be reduced to the parametric case above.
        \item $\redexcon = \ife{\evalcon}{\exprcon_2}{\exprcon_3}$: Since the expression steps to an \lstinline|ERROR|, we know that $\evalcon[\eval]$ is not a boolean value. This implies that $\evalcon = \circ$, or $\evalcon = n$ where $n \in \mathbb{Z}$, or $\evalcon =\ $\lstinline|fun x -> $\ \exprcon_f$|. In all three cases, when we substitute $\mkgen{\syntype}$ into $\circ$, $\evalcon[\mkgen{\syntype}]$ still will not be a boolean value. Therefore, we can conclude that $\exprcon[\mkgen{\syntype}] \smallstep[] \eerror$.
        \item $\redexcon = \matches{\evalcon}{p}$: This is an impossible case. Since $\matches{\eval}{p}$ always returns a value, we know that the overall expression, $\exprcon[\eval]$, will have to take more than one step to evaluate to an \lstinline|ERROR| expression in this case. 
      \end{enumerate}

    \end{enumerate}

    \textbf{Inductive step: }$\exprcon[\eval] \smallsteps[] \eerror$

    Let us examine the first step in the given computation, which is effectively $\exprcon[\eval] \smallstep[] \expr_1$ for some intermediate evaluation result, $\expr_1$.
    
    As with the base case, we know that there exists unique $\rccon$ and $\redexcon$ such that $\exprcon[\circ] = \redcon{\circ}{\redexcon[\circ]}$. There are two main scenarios to consider:

    \begin{enumerate}
      \item $\redexcon$ is parametric: By Lemma~\ref{lemma:parametric_redcon}, we know that $\exprcon[\eval] \smallstep[] \exprcon'[\eval]$ implies $\exprcon[\mkgen{\syntype}] \smallstep[] \exprcon'[\mkgen{\syntype}]$. The result then follows by the inner induction hypothesis. 
      
      \item $\redexcon$ is not parametric: We proceed by case analysis on the redex context, $\redexcon$.
      \begin{enumerate}

        \item $\redexcon = \evalcon_1 + \evalcon_2$: According to the operational semantics, there are two cases to consider here. If both operands are integers, it implies that $\evalcon_1$ and $\evalcon_2$ don't contain any holes, which means $\redexcon[\eval] = \redexcon[\mkgen{\syntype}]$, and the rest follows by the induction hypothesis. The case where either one of the operands is not an integer is handled by the base case, since this will result in an immediate \lstinline|ERROR| expression after a single step of evaluation. 
        
        \item $\redexcon = \evalcon_1 \ \evalcon_2$: We've already covered the case of applying to a non-function in the base case. Now, we only need to consider the cases where $\evalcon_1 = \circ$, since $\evalcon_1 =\ $\lstinline|fun x -> $\ \exprcon_f$| for some $\exprcon_f$ is one of the cases where the redex is parametric. 
        
        Let $\eval_g = \mkgen{\mkdfun{\ev}{\syntype_1}{\syntype_2}}$. Consider the application, $\redexcon[\eval_g] = \evalcon_1[\eval_g] \ \evalcon_2[\eval_g] = \eval_g \ \evalcon_2[\eval_g]$. 
        Now consider the type of $\evalcon_2[\eval_g]$. There are two cases to consider here:
      
        \begin{enumerate}
          \item Case $\not\models\hastype{\evalcon_2[\eval_g]}{\syntype_1}$, since we have
          \begin{lstlisting}
$\eval_g$ = fun x' -> if $\pickb$ then 
  if $\mkche{\syntype_1}{\texttt{x}}$ then $\mkgen{\substitute{\syntype_2}{\ev'}{\ev}}$ else $\eerror$ 
else $\mkgen{\substitute{\syntype_2}{\ev'}{\ev}}$
          \end{lstlisting}
        
          We know that 
          \begin{lstlisting}
$\eval_g$ $\evalcon_2[\eval_g]$ = if $\pickb$ then 
if $\mkche{\syntype_1}{\evalcon_2[\eval_g]}$ then $\mkgen{\substitute{\syntype_2}{\evalcon_2[\eval_g]}{\ev}}$ else $\eerror$ 
else $\mkgen{\substitute{\syntype_2}{\evalcon_2[\eval_g]}{\ev}}$
          \end{lstlisting}
        
          Since $\not\models\hastype{\evalcon_2[\eval_g]}{\syntype_1}$, and $\syntype_1$ is a smaller type, by induction hypothesis \textbf{clause (2)}, we know that $\mkche{\syntype_1}{\evalcon_2[\eval_g]} \smallsteps[] \gtfalse$ or $\mkche{\syntype_1}{\evalcon_2[\eval_g]} \smallsteps[] \eerror$. Either way, it will result in $\eval_g \ \evalcon_2[\eval_g] \smallsteps[] \eerror$ where $\pickb \smallstep[] \gttrue$, thus $\exprcon[\eval_g] \smallsteps[] \eerror$. 
              
        \item Case $\models\hastype{\evalcon_2[\eval_g]}{\syntype_1}$: Since $\eval$ is a function, we know that $\eval = \mkfunv{\ev}{\exprcon_v}$ for some context, $\exprcon_v$. Consider w.l.o.g. evaluations where $\pickb \smallstep[] \gtfalse$, i.e. the cases where the argument is not checked, we have $\redexcon[\eval] = \eval \ \evalcon_2[\eval] \smallsteps[] \substitute{\exprcon_v}{\evalcon_2[\eval]}{\ev}$. Putting it back into the overall context, we have $\redcon{\eval}{\substitute{\exprcon_v}{\evalcon_2[\eval]}{\ev}} \smallsteps[] \eerror$. Since this is smaller than the original computation, $\redcon{\eval}{\eval \ \evalcon_2[\eval]} \smallsteps[] \eerror$, by inner induction hypothesis, we have $\redcon{\eval_g}{\substitute{\exprcon_v}{\evalcon_2[\eval_g]}{\ev}} \smallsteps[] \eerror$. Now, consider the application, $\eval \ \evalcon_2[\eval_g]$ where $\pickb \smallstep[] \gtfalse$. By operational semantics, we have $\eval \ \evalcon_2[\eval_g] \smallsteps[] \substitute{\exprcon_v}{\evalcon_2[\eval_g]}{\ev}$. Since we have $\models\hastype{\evalcon_2[\eval_g]}{\syntype_1}$, and that $\models\hastype{\eval}{\mkdfun{\ev}{\syntype_1}{\syntype_2}}$, we can conclude that $\models\hastype{\substitute{\exprcon_v}{\evalcon_2[\eval_g]}{\ev}}{\substitute{\syntype_2}{\evalcon_2[\eval_g]}{\ev}}$. By Lemma~\ref{lemma:size_of_type_subs}, we know $\substitute{\syntype_2}{\evalcon_2[\eval_g]}{\ev}$ is a smaller type. Since we also know that $\redcon{\eval_g}{\substitute{\exprcon_v}{\evalcon_2[\eval_g]}{\ev}} \smallsteps[] \eerror$, by induction hypothesis on $\textbf{clause (1)}$, we can conclude that $\redcon{\eval_g}{\mkgen{\substitute{\syntype_2}{\evalcon_2[\eval_g]}{\ev}}} \smallsteps[] \eerror$. By operational semantics, we know that $\redexcon[\eval_g] = \eval_g \ \evalcon_2[\eval_g] \smallsteps[] \mkgen{\substitute{\syntype_2}{\evalcon_2[\eval_g]}{\ev}}$ when $\pickb \smallstep[] \gtfalse$, therefore we can conclude that $\exprcon[\eval_g] = \redcon{\eval_g}{\redexcon[\eval_g]} \smallsteps[] \eerror$.
        
      \end{enumerate}

        \item $\redexcon =\ $\lstinline|if $\evalcon_1$ then $\exprcon_2$ else $\exprcon_3$|: Since the expression doesn't immediately step to an \lstinline|ERROR|, we know that $\evalcon[\eval]$ is a boolean value. This implies that $\evalcon_1 = \vtrue$ or $\evalcon_1 = \vfalse$, reducing it to the case where $\redexcon$ is parametric.
        
        \item $\redexcon = \matches{\evalcon}{p}$: There are three possible cases to consider here: $\evalcon[\eval]$ is an integer, or $\evalcon[\eval]$ is a boolean, or $\evalcon[\eval]$ is a function value. In the first two cases, we can conclude that $\evalcon$ doesn't contain any holes, which means $\evalcon[\eval] = \evalcon[\mkgen{\syntype}]$. Therefore, we only need to consider the scenario where $\evalcon[\eval]$ is a function value, which means either $\evalcon = \circ$ or $\evalcon = \mkfunv{\ev}{\exprcon_f[\circ]}$. In both cases, it doesn't matter if we're substituting $\eval$ or $\mkgen{\syntype}$ into the hole, because the substitution result will always be functions, giving us the same pattern-matching results for both $\redexcon[\eval]$ and $\redexcon[\mkgen{\syntype}]$. 
      \end{enumerate}
    \end{enumerate}
  
  \sloppy
  Consider \textbf{clause (2)}. We need to prove that if $\expr \smallsteps[] \eval_f$ and that $\not \models \hastype{\eval_f}{\mkdfun{\ev}{\syntype_1}{\syntype_2}}$, then $\mkche{\mkdfun{\ev}{\syntype_1}{\syntype_2}}{\eval_f} \smallsteps[] \eerror$ or $\mkche{\mkdfun{\ev}{\syntype_1}{\syntype_2}}{\eval_f} \smallsteps[] \gtfalse$. Expanding the definition, we have $\mkche{\mkdfun{\ev}{\syntype_1}{\syntype_2}}{\eval_f} = $ \lstinline|if $\expr\ \tttilde$ fun then let arg = |\mkgen{$\syntype_1$}\ \lstinline| in |\mkche{$\substitute{\syntype_2}{\texttt{arg}}{\ev}$}{\lstinline|$\expr\ $ arg|} \lstinline| else false|. If $\eval_f$ is not a function value, we will have $\matches{\expr}{\gtfun} \smallsteps[] \gtfalse$, which means $\mkche{\mkdfun{\ev}{\syntype_1}{\syntype_2}}{\eval_f} \smallstep[] \gtfalse$. Now, let us consider the case where $\eval_f$ is a function value. Since $\not \models \hastype{\eval_f}{\mkdfun{\ev}{\syntype_1}{\syntype_2}}$, we know that there must exist some $\eval_0$ such that $\models \hastype{\eval_0}{\syntype_1}$ and $\not \models \hastype{\eval_f \ \eval_0}{\substitute{\syntype_2}{\eval_0}{\ev}}$. Since $\syntype_2$ is a smaller type, by induction hypothesis on \textbf{clause (2)}, we can conclude that $\neg\tc{\eval_f \ \eval_0}{\substitute{\syntype_2}{\eval_0}{\ev}}$, which means that $\mkche{\substitute{\syntype_2}{\eval_0}{\ev}}{\eval_f \ \eval_0} \smallsteps[] \gtfalse$ or $\mkche{\substitute{\syntype_2}{\eval_0}{\ev}}{\eval_f \ \eval_0} \smallsteps[] \eerror$. Let $\exprcon = $ \lstinline|if |$\mkche{\substitute{\syntype_2}{\circ}{\ev}}{\eval_f \ \circ}$ \lstinline| then 1 else ERROR|. We know that $\exprcon[\eval_0] \smallsteps[] \eerror$. Since $\syntype_1$ is a smaller type, by induction hypothesis on \textbf{clause (1)}, we can conclude that $\exprcon[\mkgen{\syntype_1}] \smallsteps[] \eerror$. This implies that $\mkche{\substitute{\syntype_2}{\mkgen{\syntype_1}}{\ev}}{\eval_f \ \mkgen{\syntype_1}} \smallsteps[] \gtfalse$ or $\mkche{\substitute{\syntype_2}{\mkgen{\syntype_1}}{\ev}}{\eval_f \ \mkgen{\syntype_1}} \smallsteps[] \eerror$. Since $\mkche{\mkdfun{\ev}{\syntype_1}{\syntype_2}}{\eval_f} \smallsteps[] \mkche{\substitute{\syntype_2}{\mkgen{\syntype_1}}{\ev}}{\eval_f \ \mkgen{\syntype_1}}$ by operational semantics, we can conclude that $\mkche{\mkdfun{\ev}{\syntype_1}{\syntype_2}}{\eval_f} \smallsteps[] \gtfalse$ or $\mkche{\mkdfun{\ev}{\syntype_1}{\syntype_2}}{\eval_f} \smallsteps[] \eerror$, i.e. $\neg\tc{\eval_f}{\mkdfun{\ev}{\syntype_1}{\syntype_2}}$.  

\end{proof}

Soundness is then direct from the above.

\begin{restatable}[Soundness of Extended System]{lemma}{soundnessExt}
  \label{lemma:soundness_ext}
  For all types $\syntype$ defined in Definitions~\ref{def:lr_core}, \ref{def:lr_refinement_dependent}, and \ref{def:lr_poly}, $\forall \expr.$ if $\tc{\expr}{\syntype}$, then $\models \hastype{\expr}{\syntype}$.
\end{restatable}

So, we finally have both soundness and completeness.
\begin{proof}[Proof of Theorem \ref{thm:sound_and_complete_ext}]
  The forward implication follows from Lemma~\ref{lemma:completeness_ext} and the reverse from Lemma~\ref{lemma:soundness_ext}.
\end{proof}

\section{Full Bluejay Syntax}
\label{app_bjy_syntax}

Figure \ref{fig:bluejay_grammar} contains the full syntax of the Bluejay Language discussed in Section~\ref{sec_impl}.

\begin{figure}[hbt!]
      \begin{grammar}
        \grule[values]{\eval}{
          \mathbb{Z}
          \gor \mathbb {B}
          \gor
          \mkfunv{\ev}{\expr}
          \gline
          \gor \ttob l_1 = \eval_1; \ \cdots; \ l_n = \eval_n\ttcb
          \gor [\eval_1; \cdots; \eval_n]
        }
        \grule[expressions]{\expr}{
            \eval
            \gor \ev
            \gor \expr\ \expr
            \gor \expr\ \binop\ \expr
            \gline
            \gor \ttob l_1 = \expr_1; \ \cdots; \ l_n = \expr_n\ttcb
            \gor \expr.l 
            \gline
            \gor [\expr_1; \cdots; \expr_n]
            \gor \expr :: \expr
            \gline
            \gor \ife{\expr}{\expr}{\expr}
            \gline 
            \gor \texttt{\small match }\expr \texttt{\small\ with }p \texttt{\small\ -> }\expr \cdots
            \gline
            \gor \letin{\ev}{\expr}{\expr}
            \gor \letint{\ev}{\syntype}{\expr}{\expr}
            \gline
            \gor \letfun
            \gor \letfunt
            \gline
            \gor \texttt{\small let f (type a ... b) (x : $\syntype$) : $\syntype$ = $\expr$}
            \gline
            \gor \texttt{\small assert }\expr
            \gor \texttt{\small assume }\expr
            \gline
            \gor \syntype
            \gor \texttt{\small input}
            \gor \texttt{\small ERROR}
            \gor \mzero
        }
        \grule[variables]{\ev}{
          \textit{(identifiers)}
        }
        \grule[patterns]{p}{
          \tint
          \gor \tbool
          \gor \gtfun
          \gor \ttob l_1; \cdots; l_n \ttcb
          \gor \gtany
          \gline
          \gor x :: y
          \gor []
        }
        \grule[labels]{l}{
          \textit{(identifiers)}
        }
        \grule[poly variables]{\alpha}{
          \texttt{'a} 
          \gor \texttt{'b} 
          \gor \cdots
        }
        \grule[type variables]{\beta}{
          \textit{(identifiers)}
        }
        \grule[types]{\syntype}{
          \tint
          \gor \tbool
          \gor \tfun
          \gline
          \gor \mktset{\syntype}{\expr}
          \gor \mkdfun{\ev}{\syntype}{\syntype}
          \gline
          \gor (\mkvariant{\syntype_1}{\syntype_n})
          \gline 
          \gor \mkintersect{(\mkfun{(V_1\ \texttt{of}\ \syntype_{1})}{\syntype_{1}'})}{(\mkfun{(V_n\ \texttt{of}\ \syntype_{n})}{\syntype_{n}'})}
          \gline
          \gor \alpha \gor \beta
          \gor \mkmiu{\alpha}{\syntype}
          \gor \ttob \hastype{l_1}{\syntype_1}; \ \cdots; \ \hastype{l_n}{\syntype_n} \ttcb
          \gor \texttt{\small list }\syntype
        }
      \end{grammar}
      \caption{Bluejay Language Grammar}
      \label{fig:bluejay_grammar}
    \end{figure} 

\section{Record Subtyping Extension}
\label{app_record}

Figure \ref{fig:records_op_sem} contains the full set of additional operational semantics rules necessary for record subtyping as discussed in Section \ref{subsec:rec_and_subtyping}.

\begin{figure}
  \begin{mathpar}

    \bbrule{Record}{
    }{
      \ttob l_1 = \eval_1; \ldots; l_n = \eval_n \ttcb \smallstep[] \ttob l_1 = \eval_1; \ldots; l_n = \eval_n \ttcb^{\ttob l_1; \ldots; l_n \ttcb}
    }

    \bbrule{Projection}{i \leq m \leq n
    }{
      \ttob l_1 = \eval_1; \ldots; l_i = \eval_i; \ldots; l_n = \eval_n \ttcb^{\ttob l_1; \ldots; l_m\ttcb}.l_i \smallstep[] \eval_i
    }

    \bbrule{Proj-Error}{
      l \notin \ttob l_1; \ldots; l_m \ttcb
    }{
      \ttob l_1 = \eval_1; \ldots; l_n = \eval_n \ttcb^{\ttob l_1; \ldots; l_m\ttcb}.l \smallstep[] \eerror
    }

    \bbrule{Retag}{
      \eval = \ttob l_1 = \eval_1; \ldots; l_n = \eval_n \ttcb^{\ttob l'_1; \ldots; l'_k \ttcb} \\
      m \leq n
    }{
      \retag \ttop\eval, \ttob l_1; \ldots; l_m \ttcb\ttcp \smallstep[]  \ttob l_1 = \eval_1; \ldots; l_n = \eval_n \ttcb^{\ttob l_1; \ldots; l_m \ttcb}
    }

    \bbrule{Retag-Error1}{
      \eval \text{ is not a record value}
    }{
      \retag \ttop\eval, \ttob l_1; \ldots; l_m \ttcb\ttcp \smallstep[] \eerror
    }

    \bbrule{Retag-Error2}{
      \eval = \ttob l_1 = \eval_1; \ldots; l_n = \eval_n \ttcb^{\ttob l_1; \ldots; l_m \ttcb} \\
      \ttob l_1'; \ldots; l_m' \ttcb \not\subseteq \ttob l_1; \ldots; l_n \ttcb
    }{
      \retag \ttop\eval, \ttob l_1'; \ldots; l_m' \ttcb\ttcp \smallstep[] \eerror
    }
\end{mathpar}
\caption{Additional Operational Semantics Rules for Records with Subtyping}
  \label{fig:records_op_sem} 
\end{figure}

\section{Concolic evaluator implementation overview}
\label{app_concolic}

\ifdefined\gh
\else
  \newcommand\gh[1]{\href{#1}{\color{black}\tiny\faLink}}
\fi

In this section, we discuss how the concolic evaluator is implemented. It is covered at a high level, and there are frequent pointers to the code so that exact implementation can be seen. Refer to Section \ref{sec_backend} for the behavior of the concolic evaluator and for its heuristics.

\subsection{Interpreter}

\subsubsection{Jayil} The concolic evaluator interprets the Jayil language. Jayil was created for DDSE \cite{DDSE} and is in administrative normal form (ANF). For simplicity, the concolic evaluator works over the Jayil language as well. The grammar for Jayil is described in OCaml \gh{https://archive.softwareheritage.org/swh:1:cnt:6b61de06aa8d945714c9a31cdd4fd0b17f1d0f6a;origin=https://github.com/JHU-PL-Lab/jaylang;visit=swh:1:snp:d10322f66a3531ab7251120395770910fe5980dd;anchor=swh:1:rev:54ef2e63cefcaecab809bf1fe73ff5e1795e6f29;path=/src/lang-jayil/ast.ml}.

\subsubsection{Environments} The interpreter is environment-based with an immutable tree map from an identifier to a value. The concolic evaluator overlays the default interpreter, so the behavior of the interpreter can be seen inside the concolic evaluator \gh{https://archive.softwareheritage.org/swh:1:cnt:9079b3d3ddd04c930645e94fc4bfbe6b41436bef;origin=https://github.com/JHU-PL-Lab/jaylang;visit=swh:1:snp:d10322f66a3531ab7251120395770910fe5980dd;anchor=swh:1:rev:54ef2e63cefcaecab809bf1fe73ff5e1795e6f29;path=/src/concolic/evaluator.ml;lines=75}.

\subsection{Solver}

The concolic evaluator uses the Z3 version 4.12.5 SMT solver to solve the symbolic expressions and target new program paths.

\subsubsection{Representations} Jayil has only four data types: 1) integers, 2) booleans, 3) functions, and 4) records. Integers and booleans are a primitive sort in Z3, and we represent functions with a string identifier. Jayil clauses correspond nicely with Z3 formulas except for pattern matching clauses. This is easy for all patterns except record labels. For this reason, records are represented as bitvectors in Z3 \gh{https://archive.softwareheritage.org/swh:1:cnt:54694433d1e7c0bd80916a0b2c849e975a51b702;origin=https://github.com/JHU-PL-Lab/jaylang;visit=swh:1:snp:d10322f66a3531ab7251120395770910fe5980dd;anchor=swh:1:rev:54ef2e63cefcaecab809bf1fe73ff5e1795e6f29;path=/src-vendor/sudu/z3_api.ml;lines=82-85}, where each bit is an indicator for the presence of a label in the record \gh{https://archive.softwareheritage.org/swh:1:cnt:09f8afbc874ac8f6108bd78ba6f293a1091a92f8;origin=https://github.com/JHU-PL-Lab/jaylang;visit=swh:1:snp:d10322f66a3531ab7251120395770910fe5980dd;anchor=swh:1:rev:54ef2e63cefcaecab809bf1fe73ff5e1795e6f29;path=/src/from_dbmc/solve/riddler.ml;lines=11-84}. For example, a Jayil program with record labels \texttt{a}, \texttt{b}, and \texttt{c} might have bits 0, 1, and 2 indicate the presence of \texttt{a}, \texttt{b}, and \texttt{c} respectively. Then \texttt{r = \{ a = 1 ; c = 2 \}} is represented as the bitvector \texttt{101} in Z3. To check whether \texttt{r} has strict record pattern \texttt{p} (i.e. \texttt{r} has all the labels in \texttt{p} and no other labels), the Z3 formula asserts \texttt{r = p}. To check that \texttt{r} has at least the labels in \texttt{p}, the Z3 formula asserts that \texttt{p = r $\land$ p} \gh{https://archive.softwareheritage.org/swh:1:cnt:09f8afbc874ac8f6108bd78ba6f293a1091a92f8;origin=https://github.com/JHU-PL-Lab/jaylang;visit=swh:1:snp:d10322f66a3531ab7251120395770910fe5980dd;anchor=swh:1:rev:54ef2e63cefcaecab809bf1fe73ff5e1795e6f29;path=/src/from_dbmc/solve/riddler.ml;lines=280-305}.

\subsubsection{Keys} Since Jayil is in ANF, all clauses in the program have a unique string identifier. A hash map assigns a unique integer to each identifier and lets this integer identify the clause in Z3 \gh{https://archive.softwareheritage.org/swh:1:cnt:09f8afbc874ac8f6108bd78ba6f293a1091a92f8;origin=https://github.com/JHU-PL-Lab/jaylang;visit=swh:1:snp:d10322f66a3531ab7251120395770910fe5980dd;anchor=swh:1:rev:54ef2e63cefcaecab809bf1fe73ff5e1795e6f29;path=/src/from_dbmc/solve/riddler.ml;lines=97-135}. This is done to avoid Z3's internal handling of strings. This representation would be sufficient if there were no recursion; however, since variables in recursive functions can have different values depending on the recursive depth, the actual key for a clause is the clause's string identifier and the number of functions that have been entered on the path to that clause \gh{https://archive.softwareheritage.org/swh:1:cnt:d7daef5e3497108f8e4f8115be9ef37de4c84556;origin=https://github.com/JHU-PL-Lab/jaylang;visit=swh:1:snp:d10322f66a3531ab7251120395770910fe5980dd;anchor=swh:1:rev:54ef2e63cefcaecab809bf1fe73ff5e1795e6f29;path=/src/concolic/structure/concolic_key.ml}. Because the concolic evaluator solves for a condition along exactly one path at a time, this method is sufficient to uniquely identify each runtime clause in the solver.

\subsubsection{Mutation} Since the concolic evaluator is implemented functionally (except for occasional mutation that is behind an interface) and the Z3 SMT solver is mutable, there is an inherent incompatibility. We choose to create solver instances transiently and not make use of the mutable state \gh{https://archive.softwareheritage.org/swh:1:cnt:f3037b9e2b16697538dac8e248a86a66d06594ab;origin=https://github.com/JHU-PL-Lab/jaylang;visit=swh:1:snp:d10322f66a3531ab7251120395770910fe5980dd;anchor=swh:1:rev:54ef2e63cefcaecab809bf1fe73ff5e1795e6f29;path=/src/concolic/session.ml;lines=147-151}. We are therefore inefficient in our use of the solver, but we get the benefits of correctness from functional code.

\subsection{Path tree}

The concolic evaluator aims to execute all possible programs paths up to some fixed number of conditional branches. To do this, we must store all possible program paths. This is done with a path tree.

\subsubsection{Structure} The path tree is implemented with recursive modules and functional types \gh{https://archive.softwareheritage.org/swh:1:cnt:227eaae664f89cca38bf5497ddcf4d26d69565b8;origin=https://github.com/JHU-PL-Lab/jaylang;visit=swh:1:snp:d10322f66a3531ab7251120395770910fe5980dd;anchor=swh:1:rev:54ef2e63cefcaecab809bf1fe73ff5e1795e6f29;path=/src/concolic/structure/path_tree.mli}. Each node in the tree represents the set of clauses between the latest conditional branch (or the start of the program) and the next conditional branch (or the end of the program). Therefore, a node has a child for each direction of the next conditional branch \gh{https://archive.softwareheritage.org/swh:1:cnt:074a434052dae6b7480806c5a60199bee440e3b6;origin=https://github.com/JHU-PL-Lab/jaylang;visit=swh:1:snp:d10322f66a3531ab7251120395770910fe5980dd;anchor=swh:1:rev:54ef2e63cefcaecab809bf1fe73ff5e1795e6f29;path=/src/concolic/structure/path_tree.ml;lines=173}. Each child has a status: unsolved, unknown, unsatisfiable, or hit \gh{https://archive.softwareheritage.org/swh:1:cnt:074a434052dae6b7480806c5a60199bee440e3b6;origin=https://github.com/JHU-PL-Lab/jaylang;visit=swh:1:snp:d10322f66a3531ab7251120395770910fe5980dd;anchor=swh:1:rev:54ef2e63cefcaecab809bf1fe73ff5e1795e6f29;path=/src/concolic/structure/path_tree.ml;lines=306-310}. Children are \texttt{unsolved} when they are not yet hit, and the SMT solver has not determined their satisfiability. A child is \texttt{unknown} if the solver timed out when solving. This has not yet happened in any practical use. The child is \texttt{unsatisfiable} if its condition constraint cannot be satisfied alongside all the formulas in the path to that child. In all three of these cases, the child is a leaf. Otherwise, the child has been \texttt{hit} in some program execution.

\subsubsection{Data} The formulas acquired through interpretation of a node's clauses are stored at the node \gh{https://archive.softwareheritage.org/swh:1:cnt:074a434052dae6b7480806c5a60199bee440e3b6;origin=https://github.com/JHU-PL-Lab/jaylang;visit=swh:1:snp:d10322f66a3531ab7251120395770910fe5980dd;anchor=swh:1:rev:54ef2e63cefcaecab809bf1fe73ff5e1795e6f29;path=/src/concolic/structure/path_tree.ml;lines=121}. A child node also has constraints: formulas that must be satisfied to take the branch from the node's parent \gh{https://archive.softwareheritage.org/swh:1:cnt:074a434052dae6b7480806c5a60199bee440e3b6;origin=https://github.com/JHU-PL-Lab/jaylang;visit=swh:1:snp:d10322f66a3531ab7251120395770910fe5980dd;anchor=swh:1:rev:54ef2e63cefcaecab809bf1fe73ff5e1795e6f29;path=/src/concolic/structure/path_tree.ml;lines=242}. These constraints are formulas for the node's branch condition and any additional formulas needed to satisfy \texttt{assume} or \texttt{assert} statements at the node.

\subsection{Target paths}

While the evaluator interprets the program to execute the target path, it acquires new targets: the negation of each branch taken along the path.

\subsubsection{Target representation} Targets are represented by a path of conditionals (conditional variables with their boolean values) \gh{https://archive.softwareheritage.org/swh:1:cnt:897738a3ab8aea12cc16d0b20930e9dd779a7e84;origin=https://github.com/JHU-PL-Lab/jaylang;visit=swh:1:snp:d10322f66a3531ab7251120395770910fe5980dd;anchor=swh:1:rev:54ef2e63cefcaecab809bf1fe73ff5e1795e6f29;path=/src/concolic/structure/target.ml}. The target is solved by traversing the path tree corresponding to the target's path and adding all formulas found at the nodes to the Z3 SMT solver. The solver checks compatibility of these formulas with the target's condition constraint.

\subsubsection{Target queues} Targets are stored in tree functional priority search queues: one where the targets are given a priority such that they can be popped in a depth-first search manner, another in a breadth-first search manner, and one that is uniformly random \gh{https://archive.softwareheritage.org/swh:1:cnt:291d602ef65d8c9e74b06fc54cddd971cdac3450;origin=https://github.com/JHU-PL-Lab/jaylang;visit=swh:1:snp:d10322f66a3531ab7251120395770910fe5980dd;anchor=swh:1:rev:54ef2e63cefcaecab809bf1fe73ff5e1795e6f29;path=/src/concolic/structure/target_queue.ml;lines=162}. We use priority search queues instead of stacks or queues so that when a target is pushed, it is efficiently erased from the queue and doesn't exist multiple times.

\subsubsection{Target acquisition} When the interpreter finishes, the program path is visited in the path tree to acquire new targets \gh{https://archive.softwareheritage.org/swh:1:cnt:e197f8b7eb409c418809d912562dd567082c337a;origin=https://github.com/JHU-PL-Lab/jaylang;visit=swh:1:snp:d10322f66a3531ab7251120395770910fe5980dd;anchor=swh:1:rev:54ef2e63cefcaecab809bf1fe73ff5e1795e6f29;path=/src/concolic/symbolic_session.ml;lines=71-105}. For each branch in the path, the negation of the branch is checked in the tree, and if it is an \texttt{unsolved} branch or previously terminated on a failed assert or assume statement \gh{https://archive.softwareheritage.org/swh:1:cnt:074a434052dae6b7480806c5a60199bee440e3b6;origin=https://github.com/JHU-PL-Lab/jaylang;visit=swh:1:snp:d10322f66a3531ab7251120395770910fe5980dd;anchor=swh:1:rev:54ef2e63cefcaecab809bf1fe73ff5e1795e6f29;path=/src/concolic/structure/path_tree.ml;lines=334}, then it is added to the target queue.

\subsection{Integration with interpreter}

The concolic evaluator is overlayed on the default Jayil interpreter. It repeatedly interprets the program while it consults several ``session'' modules described here.

\subsubsection{Sessions} There are several ``session'' modules that track information for the concolic evaluator.

\begin{itemize}
  \item Concrete session: holds only the information needed to concretely run the interpreter \gh{https://archive.softwareheritage.org/swh:1:cnt:f3037b9e2b16697538dac8e248a86a66d06594ab;origin=https://github.com/JHU-PL-Lab/jaylang;visit=swh:1:snp:d10322f66a3531ab7251120395770910fe5980dd;anchor=swh:1:rev:54ef2e63cefcaecab809bf1fe73ff5e1795e6f29;path=/src/concolic/session.ml;lines=10}.
  \item Symbolic session: tracks all symbolic information during a single run of the interpreter \gh{https://archive.softwareheritage.org/swh:1:cnt:b6ee1c07f731325af6eacf974c7149f4b96604fa;origin=https://github.com/JHU-PL-Lab/jaylang;visit=swh:1:snp:d10322f66a3531ab7251120395770910fe5980dd;anchor=swh:1:rev:54ef2e63cefcaecab809bf1fe73ff5e1795e6f29;path=/src/concolic/symbolic_session.mli}.
  \item Session (main): handles inter-run information, e.g. accumulates the information from the symbolic session into the path tree \gh{https://archive.softwareheritage.org/swh:1:cnt:f3037b9e2b16697538dac8e248a86a66d06594ab;origin=https://github.com/JHU-PL-Lab/jaylang;visit=swh:1:snp:d10322f66a3531ab7251120395770910fe5980dd;anchor=swh:1:rev:54ef2e63cefcaecab809bf1fe73ff5e1795e6f29;path=/src/concolic/session.ml;lines=69-177}. 
\end{itemize}

With this setup, the interpreter runs while interfacing with the concrete and symbolic sessions, and it interacts with the main session to begin the next run. This means the meat of the concolic evaluator's logic is found in the session modules, and the modules described in the previous sections are simply used by the sessions.


\subsubsection{Optional arguments} There are many settings for the concolic evaluator (e.g. max tree depth, program max step, program timeout, etc.) to be chosen by the user, and these are all optional arguments to the evaluator's \texttt{test} function \gh{https://archive.softwareheritage.org/swh:1:cnt:6309abee6f2a3d2711c4e72dd54fe9b1f79c3384;origin=https://github.com/JHU-PL-Lab/jaylang;visit=swh:1:snp:d10322f66a3531ab7251120395770910fe5980dd;anchor=swh:1:rev:54ef2e63cefcaecab809bf1fe73ff5e1795e6f29;path=/src/concolic/driver.mli;lines=21} using an optional argument module. The arguments can be wrapped as a record \gh{https://archive.softwareheritage.org/swh:1:cnt:008d588a1586fe2b546e49167e7ae89a99786b8c;origin=https://github.com/JHU-PL-Lab/jaylang;visit=swh:1:snp:d10322f66a3531ab7251120395770910fe5980dd;anchor=swh:1:rev:54ef2e63cefcaecab809bf1fe73ff5e1795e6f29;path=/src/concolic/options.ml;lines=5} for internal use, and they extend nicely to \texttt{Argparse} command-line arguments using \texttt{ref} cells \gh{https://archive.softwareheritage.org/swh:1:cnt:008d588a1586fe2b546e49167e7ae89a99786b8c;origin=https://github.com/JHU-PL-Lab/jaylang;visit=swh:1:snp:d10322f66a3531ab7251120395770910fe5980dd;anchor=swh:1:rev:54ef2e63cefcaecab809bf1fe73ff5e1795e6f29;path=/src/concolic/options.ml;lines=20}. Optional arguments can be applied to the function as a record or as OCaml optional arguments, and compositions and mappings are supported in a monad-like way for functions with the same optional arguments \gh{https://archive.softwareheritage.org/swh:1:cnt:008d588a1586fe2b546e49167e7ae89a99786b8c;origin=https://github.com/JHU-PL-Lab/jaylang;visit=swh:1:snp:d10322f66a3531ab7251120395770910fe5980dd;anchor=swh:1:rev:54ef2e63cefcaecab809bf1fe73ff5e1795e6f29;path=/src/concolic/options.ml;lines=48-95}.

\subsubsection{Lwt} We use \texttt{Lwt} for timeouts with the assumption that the solver and interpreter are each independently fast on a single program path, so \texttt{Lwt} can cause the evaluator to quit between runs \gh{https://archive.softwareheritage.org/swh:1:cnt:9079b3d3ddd04c930645e94fc4bfbe6b41436bef;origin=https://github.com/JHU-PL-Lab/jaylang;visit=swh:1:snp:d10322f66a3531ab7251120395770910fe5980dd;anchor=swh:1:rev:54ef2e63cefcaecab809bf1fe73ff5e1795e6f29;path=/src/concolic/evaluator.ml;lines=324}, and these runs are dense enough in time space that the time to quit is near what is desired.

\subsection{Benchmarks and tests}
\label{app_tests}

\subsubsection{Testing} The concolic evaluator went through several versions when building it up from scratch, and to ensure correctness in the earlier versions, we tested hand-written Jayil programs and exact resulting branch statuses. In later versions, once the correctness of earlier versions had been established, the testing evolved to evaluating translated Bluejay programs and checking for the ability to find a path that leads to an \texttt{ERROR} clause \gh{https://archive.softwareheritage.org/swh:1:cnt:26396b1f23ecffce5c8f579c6f227b16fae20f7a;origin=https://github.com/JHU-PL-Lab/jaylang;visit=swh:1:snp:d10322f66a3531ab7251120395770910fe5980dd;anchor=swh:1:rev:54ef2e63cefcaecab809bf1fe73ff5e1795e6f29;path=/src-test/concolic/test_concolic.ml;lines=13}. This relies on the correctness of the translation, but it allows us to test significantly larger and more complex programs that cannot easily be worked out manually on the Jayil code, but we can convince ourselves easily of the existence or nonexistence of a type error in the Bluejay code. Language features are tested in a greater number of combinations this way, even though the testing has been truncated to only check a binary result.

Further, \texttt{assert} statements are sprinkled throughout the code (e.g. here \gh{https://archive.softwareheritage.org/swh:1:cnt:e197f8b7eb409c418809d912562dd567082c337a;origin=https://github.com/JHU-PL-Lab/jaylang;visit=swh:1:snp:d10322f66a3531ab7251120395770910fe5980dd;anchor=swh:1:rev:54ef2e63cefcaecab809bf1fe73ff5e1795e6f29;path=/src/concolic/symbolic_session.ml;lines=288}) for logical impossibilities to check the correctness of the implementation. These asserts, along with the binary testing on BlueJay programs, has us convinced that the concolic evaluator is correct.

\subsubsection{Landmarks} We use \texttt{landmarks} \cite{landmarks} to profile the performance of the evaluator. The results convinced us to use function depth in a clause's identifier to avoid hashing stacks, which was inherited from DDSE, among other small changes to improve performance. It also indicates that the functional data structures we use (sometimes with poor time complexity) only negligibly impact efficiency, and our attention is better directed at reducing calls to the SMT solver and creating fewer formulas.

\subsubsection{Benchmarks} The concolic evaluator is benchmarked by timing the call to read, parse, and translate a Bluejay program into a Jayil program separately from the call to concolically evaluate the Jayil program \gh{https://archive.softwareheritage.org/swh:1:cnt:a8776bcbcd5da0f645b952dfb8a316875422db0e;origin=https://github.com/JHU-PL-Lab/jaylang;visit=swh:1:snp:d10322f66a3531ab7251120395770910fe5980dd;anchor=swh:1:rev:54ef2e63cefcaecab809bf1fe73ff5e1795e6f29;path=/benchmark/concolic/cbenchmark.ml;lines=62-118}. Both processes are run repeatedly, and the average of the ten trials is reported. The results are printed as a LaTeX-formatted table.

In Table \ref{table_unit_tests}, we extend Table \ref{table_bjy_benchmarks} to include the smaller unit tests not included in that table. The soft contract benchmarks also have their detailed runtime and features presented in Table \ref{table_scheme_benchmarks_full}.

\begin{table}
  \begin{center}
    \scalebox{0.72}{
    \begin{tabular}{@{}r|c@{\hspace{3pt}}c@{\hspace{3pt}}c|@{\hspace{3pt}}c@{\hspace{3pt}}|c@{\hspace{3pt}}c@{\hspace{3pt}}c@{\hspace{3pt}}c@{\hspace{3pt}}c@{\hspace{3pt}}c@{\hspace{3pt}}c@{\hspace{3pt}}c@{\hspace{3pt}}c@{\hspace{3pt}}c@{\hspace{3pt}}c@{\hspace{3pt}}c@{\hspace{3pt}}c@{\hspace{3pt}}c@{\hspace{3pt}}c@{\hspace{3pt}}c@{\hspace{3pt}}c@{\hspace{3pt}}c@{\hspace{3pt}}}
    Test Name & Run & Transl & Total & LOC & \rot{\underline{P}olymorphic types} & \rot{\underline{V}ariants} & \rot{\underline{I}ntersection types} & \rot{\underline{R}ecursive functions} & \rot{\underline{M}u types} & \rot{\underline{H}igher order functions} & \rot{\underline{S}ubtyping} & \rot{\underline{T}ype casing} & \rot{\underline{O}OP-style} & \rot{Re\underline{f}inement types} & \rot{\underline{D}ependent types} & \rot{P\underline{a}rametric types} & \rot{Re\underline{c}ords} & \rot{\underline{W}rap required} & \rot{Assertio\underline{n}s} & \rot{Operator mis\underline{u}se} & \rot{Return t\underline{y}pe} & \rot{Match (X)}\\
      \hline
      \texttt{balanced\_tree}                &  4    &  140  &  143  &  49  &  --       &  V        &  --       &  R        &  \red{M}  &  --       &  --  &  --  &  --       &  \red{F}  &  --       &  --       &  C        &  --       &  --  &  --       &  --       &  X \\
      \texttt{bst\_instance}                 &  3    &  165  &  167  &  45  &  --       &  V        &  --       &  R        &  \red{M}  &  --       &  --  &  --  &  --       &  \red{F}  &  --       &  --       &  C        &  --       &  --  &  --       &  --       &  X \\
      \texttt{dep\_fun\_test\_1}             &  14   &  14   &  27   &  10  &  --       &  --       &  --       &  \red{R}  &  --       &  --       &  --  &  --  &  --       &  \red{F}  &  \red{D}  &  --       &  --       &  --       &  --  &  --       &  --       &  X \\
      \texttt{dep\_type\_test\_1}            &  48   &  8    &  55   &  11  &  --       &  --       &  --       &  --       &  --       &  \red{H}  &  --  &  --  &  --       &  --       &  \red{D}  &  --       &  --       &  --       &  --  &  --       &  Y        &  -- \\
      \texttt{flow\_sensitive\_1}            &  11   &  9    &  20   &  7   &  --       &  --       &  --       &  --       &  --       &  --       &  --  &  --  &  --       &  \red{F}  &  --       &  --       &  --       &  --       &  --  &  --       &  --       &  -- \\
      \texttt{intersection\_type\_1}         &  1    &  2    &  3    &  2   &  --       &  --       &  \red{I}  &  --       &  --       &  --       &  --  &  --  &  --       &  --       &  --       &  --       &  --       &  --       &  --  &  --       &  Y        &  -- \\
      \texttt{intersection\_type\_2}         &  1    &  21   &  21   &  3   &  --       &  --       &  I        &  --       &  --       &  H        &  --  &  --  &  --       &  --       &  --       &  --       &  --       &  --       &  --  &  --       &  \red{Y}  &  -- \\
      \texttt{let\_fun\_test\_1}             &  1    &  2    &  2    &  2   &  --       &  --       &  --       &  --       &  --       &  --       &  --  &  --  &  --       &  --       &  --       &  --       &  --       &  --       &  --  &  --       &  \red{Y}  &  -- \\
      \texttt{let\_fun\_test\_2}             &  1    &  2    &  2    &  2   &  --       &  --       &  --       &  --       &  --       &  --       &  --  &  --  &  --       &  --       &  --       &  --       &  --       &  --       &  --  &  --       &  \red{Y}  &  -- \\
      \texttt{let\_fun\_test\_4}             &  1    &  2    &  2    &  2   &  --       &  --       &  --       &  --       &  --       &  --       &  --  &  --  &  --       &  --       &  --       &  --       &  --       &  --       &  --  &  \red{U}  &  --       &  -- \\
      \texttt{let\_fun\_test\_5}             &  1    &  2    &  2    &  5   &  --       &  --       &  --       &  --       &  --       &  --       &  --  &  --  &  --       &  --       &  --       &  --       &  --       &  --       &  --  &  \red{U}  &  --       &  -- \\
      \texttt{let\_fun\_test\_7}             &  1    &  2    &  2    &  4   &  --       &  --       &  --       &  --       &  --       &  --       &  --  &  --  &  --       &  --       &  --       &  --       &  --       &  --       &  --  &  --       &  \red{Y}  &  -- \\
      \texttt{let\_fun\_test\_8}             &  1    &  8    &  8    &  7   &  --       &  --       &  --       &  --       &  --       &  H        &  --  &  --  &  --       &  --       &  --       &  --       &  --       &  --       &  --  &  --       &  \red{Y}  &  -- \\
      \texttt{list\_test\_1}                 &  13   &  8    &  20   &  7   &  --       &  --       &  --       &  --       &  --       &  --       &  --  &  --  &  --       &  --       &  --       &  --       &  --       &  --       &  --  &  --       &  \red{Y}  &  X \\
      \texttt{list\_test\_3}                 &  1    &  4    &  4    &  2   &  --       &  --       &  --       &  --       &  --       &  --       &  --  &  --  &  --       &  --       &  --       &  --       &  --       &  --       &  --  &  --       &  \red{Y}  &  -- \\
      \texttt{list\_test\_4}                 &  1    &  4    &  4    &  4   &  --       &  --       &  --       &  --       &  --       &  --       &  --  &  --  &  --       &  --       &  --       &  --       &  --       &  --       &  --  &  --       &  \red{Y}  &  -- \\
      \texttt{list\_test\_5}                 &  1    &  4    &  4    &  4   &  --       &  --       &  --       &  --       &  --       &  --       &  --  &  --  &  --       &  --       &  --       &  --       &  --       &  --       &  --  &  --       &  \red{Y}  &  -- \\
      \texttt{mutually\_rec\_1}              &  1    &  5    &  5    &  10  &  --       &  --       &  --       &  \red{R}  &  --       &  --       &  --  &  --  &  --       &  --       &  --       &  --       &  --       &  --       &  --  &  --       &  \red{Y}  &  -- \\
      \texttt{mutually\_rec\_2}              &  1    &  5    &  5    &  10  &  --       &  --       &  --       &  \red{R}  &  --       &  --       &  --  &  --  &  --       &  --       &  --       &  --       &  --       &  --       &  --  &  --       &  \red{Y}  &  -- \\
      \texttt{mutually\_rec\_dep\_types\_1}  &  340  &  6    &  345  &  15  &  --       &  --       &  --       &  \red{R}  &  --       &  --       &  --  &  --  &  --       &  --       &  \red{D}  &  --       &  --       &  --       &  --  &  --       &  \red{Y}  &  -- \\
      \texttt{parametric\_id}                &  2    &  2    &  3    &  6   &  \red{P}  &  --       &  --       &  --       &  --       &  --       &  --  &  --  &  --       &  --       &  --       &  \red{A}  &  --       &  \red{W}  &  --  &  --       &  --       &  -- \\
      \texttt{pattern\_match\_1}             &  1    &  1    &  1    &  3   &  --       &  --       &  --       &  --       &  --       &  --       &  --  &  --  &  --       &  --       &  --       &  --       &  \red{C}  &  --       &  --  &  --       &  --       &  \red{X} \\
      \texttt{poly\_apply}                   &  1    &  6    &  6    &  4   &  \red{P}  &  --       &  --       &  --       &  --       &  H        &  --  &  --  &  --       &  --       &  --       &  --       &  --       &  --       &  --  &  --       &  --       &  -- \\
      \texttt{poly\_casting}                 &  1    &  2    &  2    &  4   &  \red{P}  &  --       &  --       &  --       &  --       &  --       &  --  &  --  &  --       &  --       &  --       &  --       &  --       &  --       &  --  &  \red{U}  &  --       &  -- \\
      \texttt{poly\_fst}                     &  1    &  2    &  3    &  4   &  \red{P}  &  --       &  --       &  --       &  --       &  --       &  --  &  --  &  --       &  --       &  --       &  --       &  --       &  --       &  --  &  --       &  \red{Y}  &  -- \\
      \texttt{poly\_map}                     &  83   &  28   &  110  &  7   &  \red{P}  &  --       &  --       &  \red{R}  &  --       &  H        &  --  &  --  &  --       &  --       &  --       &  --       &  --       &  --       &  --  &  --       &  \red{Y}  &  X \\
      \texttt{poly\_record}                  &  2    &  4    &  5    &  8   &  \red{P}  &  --       &  --       &  --       &  --       &  --       &  --  &  --  &  --       &  --       &  --       &  \red{A}  &  \red{C}  &  \red{W}  &  --  &  --       &  --       &  -- \\
      \texttt{poly\_specification}           &  1    &  6    &  6    &  7   &  \red{P}  &  --       &  --       &  --       &  --       &  --       &  --  &  --  &  --       &  --       &  --       &  --       &  --       &  --       &  --  &  --       &  --       &  -- \\
      \texttt{project\_non\_record}          &  1    &  1    &  1    &  1   &  --       &  --       &  --       &  --       &  --       &  --       &  --  &  --  &  --       &  --       &  --       &  --       &  \red{C}  &  --       &  --  &  --       &  --       &  -- \\
      \texttt{rec\_fun\_1}                   &  13   &  8    &  20   &  17  &  --       &  --       &  --       &  \red{R}  &  --       &  --       &  --  &  --  &  --       &  \red{F}  &  --       &  --       &  --       &  --       &  --  &  --       &  --       &  X \\
      \texttt{rec\_fun\_2}                   &  71   &  2    &  73   &  6   &  --       &  --       &  --       &  R        &  --       &  --       &  --  &  --  &  --       &  --       &  --       &  --       &  --       &  --       &  --  &  --       &  \red{Y}  &  -- \\
      \texttt{rec\_fun\_4}                   &  70   &  2    &  72   &  6   &  --       &  --       &  --       &  R        &  --       &  --       &  --  &  --  &  --       &  --       &  --       &  --       &  --       &  --       &  --  &  --       &  \red{Y}  &  -- \\
      \texttt{rec\_fun\_5}                   &  74   &  2    &  76   &  6   &  --       &  --       &  --       &  R        &  --       &  --       &  --  &  --  &  --       &  --       &  --       &  --       &  --       &  --       &  --  &  --       &  \red{Y}  &  -- \\
      \texttt{rec\_fun\_6}                   &  68   &  4    &  72   &  10  &  --       &  --       &  --       &  R        &  --       &  --       &  --  &  --  &  --       &  --       &  --       &  --       &  --       &  --       &  --  &  --       &  \red{Y}  &  -- \\
      \texttt{rec\_fun\_7}                   &  15   &  81   &  95   &  7   &  --       &  --       &  --       &  R        &  --       &  H        &  --  &  --  &  --       &  --       &  --       &  --       &  --       &  --       &  --  &  --       &  \red{Y}  &  X \\
      \texttt{record\_1}                     &  1    &  8    &  9    &  4   &  --       &  --       &  --       &  --       &  --       &  --       &  --  &  --  &  --       &  --       &  --       &  --       &  \red{C}  &  --       &  --  &  --       &  \red{Y}  &  -- \\
      \texttt{record\_11}                    &  1    &  3    &  3    &  5   &  --       &  --       &  --       &  --       &  --       &  --       &  --  &  --  &  --       &  --       &  --       &  --       &  \red{C}  &  --       &  --  &  --       &  --       &  \red{X} \\
      \texttt{record\_2}                     &  1    &  4    &  5    &  4   &  --       &  --       &  --       &  --       &  --       &  --       &  --  &  --  &  --       &  --       &  --       &  --       &  \red{C}  &  --       &  --  &  --       &  --       &  -- \\
      \texttt{record\_4}                     &  16   &  326  &  341  &  14  &  --       &  --       &  --       &  --       &  --       &  --       &  --  &  --  &  --       &  \red{F}  &  --       &  --       &  \red{C}  &  --       &  --  &  --       &  --       &  -- \\
      \texttt{record\_5}                     &  15   &  47   &  61   &  10  &  --       &  --       &  --       &  --       &  --       &  --       &  --  &  --  &  --       &  \red{F}  &  --       &  --       &  C        &  --       &  --  &  --       &  --       &  -- \\
      \texttt{record\_6}                     &  1    &  9    &  9    &  4   &  --       &  --       &  --       &  --       &  --       &  --       &  --  &  --  &  --       &  --       &  --       &  --       &  C        &  --       &  --  &  --       &  \red{Y}  &  -- \\
      \texttt{record\_7}                     &  183  &  45   &  228  &  12  &  --       &  --       &  --       &  --       &  --       &  --       &  --  &  --  &  --       &  \red{F}  &  --       &  --       &  C        &  --       &  --  &  --       &  --       &  -- \\
      \texttt{self\_passing}                 &  2    &  7    &  9    &  10  &  --       &  --       &  --       &  --       &  --       &  H        &  --  &  --  &  \red{O}  &  --       &  --       &  --       &  \red{C}  &  --       &  --  &  --       &  --       &  -- \\
      \texttt{set\_type\_1}                  &  14   &  15   &  29   &  11  &  --       &  --       &  --       &  --       &  --       &  --       &  --  &  --  &  --       &  \red{F}  &  --       &  --       &  --       &  --       &  --  &  --       &  --       &  -- \\
      \texttt{set\_type\_2}                  &  1    &  2    &  3    &  4   &  --       &  --       &  --       &  --       &  --       &  --       &  --  &  --  &  --       &  \red{F}  &  --       &  --       &  --       &  --       &  --  &  --       &  --       &  -- \\
      \texttt{set\_type\_4}                  &  1    &  18   &  19   &  14  &  --       &  --       &  --       &  R        &  --       &  --       &  --  &  --  &  --       &  \red{F}  &  --       &  --       &  --       &  --       &  --  &  --       &  --       &  X \\
      \texttt{sub\_simple\_function}         &  1    &  101  &  101  &  8   &  --       &  --       &  --       &  --       &  --       &  \red{H}  &  S   &  --  &  --       &  --       &  --       &  --       &  \red{C}  &  --       &  --  &  --       &  Y        &  -- \\
      \texttt{sub\_variant}                  &  1    &  38   &  38   &  4   &  P        &  \red{V}  &  --       &  --       &  --       &  --       &  S   &  --  &  --       &  --       &  --       &  --       &  --       &  --       &  --  &  --       &  \red{Y}  &  -- \\
      \texttt{union\_type\_1}                &  1    &  20   &  20   &  4   &  --       &  \red{V}  &  --       &  --       &  --       &  --       &  --  &  --  &  --       &  --       &  --       &  --       &  C        &  --       &  --  &  --       &  --       &  -- \\
      \texttt{union\_type\_2}                &  1    &  701  &  701  &  4   &  --       &  \red{V}  &  --       &  --       &  --       &  --       &  --  &  --  &  --       &  --       &  --       &  --       &  C        &  --       &  --  &  --       &  --       &  -- \\
      \texttt{union\_type\_3}                &  1    &  27   &  28   &  4   &  --       &  \red{V}  &  --       &  --       &  --       &  --       &  --  &  --  &  --       &  --       &  --       &  --       &  --       &  --       &  --  &  --       &  --       &  -- \\
      \texttt{variant\_type\_1}              &  1    &  3    &  3    &  4   &  --       &  \red{V}  &  --       &  --       &  --       &  --       &  --  &  --  &  --       &  --       &  --       &  --       &  --       &  --       &  --  &  --       &  --       &  -- \\
    \end{tabular}
  }
  \end{center}
  \captionof{table}{Complete set of tests and benchmarks not include in Table \ref{table_bjy_benchmarks}. Run times and translation times are in ms. Letters are used for readability to indicate which features from Table \ref{table_bjy_features} are in the test. Black font indicates the feature is used, \red{red font} indicates the feature is key to the type error, and -- indicates the feature is not present in the test. }
  \label{table_unit_tests}
\end{table}

\begin{table}
  \begin{center}
    \scalebox{0.9}{
    \begin{tabular}{@{}r|c@{\hspace{3pt}}c@{\hspace{3pt}}c|@{\hspace{3pt}}c@{\hspace{3pt}}|c@{\hspace{3pt}}c@{\hspace{3pt}}c@{\hspace{3pt}}c@{\hspace{3pt}}c@{\hspace{3pt}}c@{\hspace{3pt}}c@{\hspace{3pt}}c@{\hspace{3pt}}c@{\hspace{3pt}}c@{\hspace{3pt}}c@{\hspace{3pt}}c@{\hspace{3pt}}c@{\hspace{3pt}}c@{\hspace{3pt}}c@{\hspace{3pt}}c@{\hspace{3pt}}c@{\hspace{3pt}}c@{\hspace{3pt}}}
    Test Name & Run & Transl & Total & LOC & \rot{\underline{P}olymorphic types} & \rot{\underline{V}ariants} & \rot{\underline{I}ntersection types} & \rot{\underline{R}ecursive functions} & \rot{\underline{M}u types} & \rot{\underline{H}igher order functions} & \rot{\underline{S}ubtyping} & \rot{\underline{T}ype casing} & \rot{\underline{O}OP-style} & \rot{Re\underline{f}inement types} & \rot{\underline{D}ependent types} & \rot{P\underline{a}rametric types} & \rot{Re\underline{c}ords} & \rot{\underline{W}rap required} & \rot{Assertio\underline{n}s} & \rot{Operator mis\underline{u}se} & \rot{Return t\underline{y}pe} & \rot{Match (X)}\\
      \hline
      \texttt{all}              &  85   &  20   &  104  &  11  &  \red{P}  &  --  &  --  &  R        &  --       &  H        &  --  &  --  &  --  &  --       &  --       &  --  &  --  &  --       &  --       &  --  &  \red{Y}  &  -- \\
      \texttt{append}           &  372  &  61   &  432  &  13  &  P        &  --  &  --  &  R        &  --       &  --       &  --  &  --  &  --  &  \red{F}  &  --       &  --  &  --  &  --       &  --       &  --  &  --       &  -- \\
      \texttt{boolflip\_e}      &  3    &  2    &  4    &  17  &  --       &  --  &  --  &  R        &  --       &  --       &  --  &  --  &  --  &  --       &  --       &  --  &  --  &  --       &  \red{N}  &  --  &  --       &  -- \\
      \texttt{braun\_tree}      &  219  &  107  &  325  &  33  &  P        &  V   &  --  &  \red{R}  &  \red{M}  &  --       &  --  &  --  &  --  &  \red{F}  &  --       &  --  &  C   &  --       &  --       &  --  &  --       &  X \\
      \texttt{flatten}          &  465  &  38   &  503  &  18  &  P        &  V   &  --  &  \red{R}  &  \red{M}  &  --       &  --  &  --  &  --  &  --       &  --       &  --  &  --  &  --       &  --       &  --  &  \red{Y}  &  X \\
      \texttt{fold\_fun\_list}  &  143  &  24   &  167  &  20  &  --       &  --  &  --  &  \red{R}  &  --       &  H        &  --  &  --  &  --  &  \red{F}  &  --       &  --  &  --  &  --       &  --       &  --  &  --       &  X \\
      \texttt{foldl}            &  93   &  82   &  175  &  7   &  --       &  --  &  --  &  \red{R}  &  --       &  H        &  --  &  --  &  --  &  --       &  --       &  --  &  --  &  \red{W}  &  --       &  --  &  --       &  X \\
      \texttt{foldl1}           &  13   &  42   &  55   &  11  &  P        &  --  &  --  &  R        &  --       &  H        &  --  &  --  &  --  &  --       &  --       &  --  &  --  &  --       &  \red{N}  &  --  &  --       &  X \\
      \texttt{foldr}            &  99   &  78   &  176  &  7   &  --       &  --  &  --  &  \red{R}  &  --       &  H        &  --  &  --  &  --  &  --       &  --       &  --  &  --  &  \red{W}  &  --       &  --  &  --       &  X \\
      \texttt{foldr1}           &  13   &  39   &  51   &  11  &  P        &  --  &  --  &  R        &  --       &  H        &  --  &  --  &  --  &  --       &  --       &  --  &  --  &  --       &  \red{N}  &  --  &  --       &  X \\
      \texttt{hors}             &  18   &  25   &  43   &  28  &  --       &  --  &  --  &  R        &  --       &  \red{H}  &  --  &  --  &  --  &  \red{F}  &  --       &  --  &  --  &  \red{W}  &  --       &  --  &  --       &  -- \\
      \texttt{hrec}             &  1    &  4    &  5    &  9   &  --       &  --  &  --  &  \red{R}  &  --       &  H        &  --  &  --  &  --  &  \red{F}  &  --       &  --  &  --  &  --       &  --       &  --  &  --       &  -- \\
      \texttt{intro1}           &  72   &  7    &  79   &  14  &  --       &  --  &  --  &  --       &  --       &  H        &  --  &  --  &  --  &  \red{F}  &  --       &  --  &  --  &  \red{W}  &  --       &  --  &  --       &  -- \\
      \texttt{intro3}           &  13   &  15   &  27   &  15  &  --       &  --  &  --  &  --       &  --       &  H        &  --  &  --  &  --  &  \red{F}  &  D        &  --  &  --  &  \red{W}  &  --       &  --  &  --       &  -- \\
      \texttt{last}             &  13   &  10   &  22   &  20  &  P        &  --  &  --  &  R        &  --       &  H        &  --  &  --  &  --  &  --       &  --       &  --  &  --  &  --       &  \red{N}  &  --  &  --       &  X \\
      \texttt{lastpair}         &  83   &  90   &  173  &  13  &  P        &  --  &  --  &  R        &  --       &  --       &  --  &  --  &  --  &  \red{F}  &  --       &  --  &  --  &  --       &  --       &  --  &  --       &  X \\
      \texttt{max}              &  28   &  66   &  94   &  13  &  --       &  --  &  --  &  --       &  --       &  H        &  --  &  --  &  --  &  \red{F}  &  D        &  --  &  --  &  --       &  --       &  --  &  --       &  -- \\
      \texttt{mem}              &  100  &  20   &  119  &  21  &  --       &  --  &  --  &  \red{R}  &  --       &  --       &  --  &  --  &  --  &  \red{F}  &  \red{D}  &  --  &  --  &  --       &  --       &  --  &  --       &  X \\
      \texttt{member}           &  13   &  13   &  26   &  10  &  --       &  --  &  --  &  R        &  --       &  --       &  --  &  --  &  --  &  --       &  --       &  --  &  --  &  --       &  --       &  --  &  \red{Y}  &  X \\
      \texttt{mult}             &  35   &  34   &  68   &  9   &  --       &  --  &  --  &  R        &  --       &  H        &  --  &  --  &  --  &  \red{F}  &  \red{D}  &  --  &  --  &  --       &  --       &  --  &  --       &  -- \\
      \texttt{mult\_all\_e}     &  71   &  5    &  75   &  14  &  --       &  --  &  --  &  R        &  --       &  --       &  --  &  --  &  --  &  \red{F}  &  --       &  --  &  --  &  --       &  --       &  --  &  --       &  -- \\
      \texttt{mult\_cps\_e}     &  1    &  3    &  4    &  15  &  --       &  --  &  --  &  R        &  --       &  H        &  --  &  --  &  --  &  --       &  --       &  --  &  --  &  --       &  \red{N}  &  --  &  --       &  -- \\
      \texttt{mult\_e}          &  1    &  4    &  4    &  9   &  --       &  --  &  --  &  R        &  --       &  --       &  --  &  --  &  --  &  \red{F}  &  --       &  --  &  --  &  --       &  --       &  --  &  --       &  -- \\
      \texttt{nth0}             &  56   &  4    &  60   &  29  &  --       &  --  &  --  &  R        &  --       &  --       &  --  &  --  &  --  &  --       &  --       &  --  &  --  &  --       &  \red{N}  &  --  &  --       &  X \\
      \texttt{r\_lock}          &  745  &  5    &  750  &  23  &  --       &  --  &  --  &  --       &  --       &  --       &  --  &  --  &  --  &  \red{F}  &  --       &  --  &  --  &  \red{W}  &  --       &  --  &  --       &  -- \\
      \texttt{reverse}          &  1    &  4    &  4    &  23  &  --       &  --  &  --  &  R        &  --       &  --       &  --  &  --  &  --  &  --       &  --       &  --  &  --  &  --       &  \red{N}  &  --  &  --       &  X \\
      \texttt{sum\_acm\_e}      &  1    &  1    &  2    &  8   &  --       &  --  &  --  &  R        &  --       &  H        &  --  &  --  &  --  &  --       &  --       &  --  &  --  &  --       &  \red{N}  &  --  &  --       &  -- \\
      \texttt{sum\_all\_e}      &  66   &  4    &  70   &  14  &  --       &  --  &  --  &  R        &  --       &  --       &  --  &  --  &  --  &  \red{F}  &  --       &  --  &  --  &  --       &  --       &  --  &  --       &  -- \\
      \texttt{sum\_e}           &  1    &  4    &  4    &  9   &  --       &  --  &  --  &  R        &  --       &  --       &  --  &  --  &  --  &  \red{F}  &  --       &  --  &  --  &  --       &  --       &  --  &  --       &  -- \\
      \texttt{tree\_depth}      &  1    &  63   &  64   &  13  &  --       &  V   &  --  &  R        &  \red{M}  &  --       &  --  &  --  &  --  &  \red{F}  &  --       &  --  &  C   &  --       &  --       &  --  &  --       &  X \\
    \end{tabular}
    }
  \end{center}
  \captionof{table}{Complete run times and translation times in ms for soft contract benchmarks \cite{RelCompleteCounterexamples}. Letters are used for readability to indicate which features from Table \ref{table_bjy_features} are in the test. Black font indicates the feature is used, \red{red font} indicates the feature is key to the type error, and -- indicates the feature is not present in the test. }
  \label{table_scheme_benchmarks_full}
\end{table}

}
\end{document}